\documentclass[11pt]{article}

\usepackage{dashrule,arydshln,empheq,amsfonts,graphicx,amsmath,amssymb,amsthm,geometry,bbm,hyperref,nicefrac,enumerate,comment,fontawesome,mathtools,cases,algorithmic,algorithm2e,stmaryrd,pict2e,picture,mathrsfs,mathabx,setspace}

\usepackage{color, colortbl}
\definecolor{Gray}{gray}{0.9}

\DeclareMathOperator*{\argmin}{arg\,min}

\newcommand{\Z}{\mathcal{Z}}

\newcommand{\IM}{\mathrm{IM}}
\newcommand{\LGD}{\mathrm{LGD}}
\newcommand{\FC}{\mathrm{FC}}
\newcommand{\TC}{\mathrm{TC}}
\newcommand{\XVA}{\mathrm{XVA}}
\newcommand{\CVA}{\mathrm{CVA}}
\newcommand{\DVA}{\mathrm{DVA}}
\newcommand{\ColVA}{\mathrm{ColVA}}
\newcommand{\MVA}{\mathrm{MVA}}
\newcommand{\FVA}{\mathrm{FVA}}
\newcommand{\I}{\mathbb{I}}
\newcommand{\R}{\mathbb{R}}
\newcommand{\N}{\mathbb{N}}
\newcommand{\E}{\mathbb{E}}

\renewcommand{\P}{\mathbb{P}}

\newcommand{\F}{\mathcal{F}}

\renewcommand{\d}{\text{d}}

\newcommand{\B}{\mathcal{B}}
\newcommand{\C}{\mathcal{C}}
\usepackage{array}
\newcolumntype{P}[1]{>{\centering\arraybackslash}p{#1}}
\newcommand{\triple}{{\vert\kern-0.25ex\vert\kern-0.25ex\vert}}
\usepackage{mathtools}
\newcommand{\defeq}{\vcentcolon=}
\newcommand{\eqdef}{=\vcentcolon}
\usepackage{geometry}
\usepackage{xcolor}

\newtheorem{lemma}{Lemma}[section]
\newtheorem{remark}{Remark}[section]

\newtheorem{assumption}[lemma]{Assumption}
\newtheorem{theorem}{Theorem}[section]

\newtheorem{definition}[lemma]{Definition}
\newtheorem{proposition}[lemma]{Proposition}
\newtheorem{example}{Example}[section]

\theoremstyle{definition}

\newcommand{\footremember}[2]{%
    \footnote{#2}
    \newcounter{#1}
    \setcounter{#1}{\value{footnote}}%
}

\begin{document}
\title{Multi-Layer Deep xVA: Structural Credit Models, Measure Changes and Convergence Analysis}
\author{%
  Kristoffer Andersson\footremember{KA}{ Department of Economics, University of Verona, via Cantarane, 24 - 37129 Verona, Italy.
  Email: \href{mailto:kristoffer.andersson@univr.it}{kristofferherbet.andersson@univr.it}}%
  \and
  Alessandro Gnoatto\footremember{AG}{Department of Economics, University of Verona, via Cantarane, 24 - 37129 Verona, Italy.
  Email: \href{mailto:alessandro.gnoatto@univr.it}{alessandro.gnoatto@univr.it}}%
}

\maketitle
\vspace{0.4cm}
\begin{abstract}
We propose a structural default model for portfolio-wide valuation adjustments (xVAs) and represent it as a system of coupled backward stochastic differential equations. The framework is divided into four layers, each capturing a key component: (i) clean values, (ii) initial margin and Collateral Valuation Adjustment (ColVA), (iii) Credit/Debit Valuation Adjustments (CVA/DVA) together with Margin Valuation Adjustment (MVA), and (iv) Funding Valuation Adjustment (FVA). Since these layers depend on one another through collateral and default effects, a naive Monte Carlo approach would require deeply nested simulations, making the problem computationally intractable.

To address this challenge, we use an iterative deep BSDE approach, handling each layer sequentially so that earlier outputs serve as inputs to the subsequent layers. Initial margin is computed via deep quantile regression to reflect margin requirements over the Margin Period of Risk. We also adopt a change-of-measure method that highlights rare but significant defaults of the bank or counterparty, ensuring that these events are accurately captured in the training process. 

We further extend Han and Long’s (2020) a posteriori error analysis to BSDEs on bounded domains. Due to the random exit from the domain, we obtain an order of convergence of $\mathcal{O}(h^{1/4-\epsilon})$ rather than the usual $\mathcal{O}(h^{1/2})$. 

Numerical experiments illustrate that this method drastically reduces computational demands and successfully scales to high-dimensional, non-symmetric portfolios. The results confirm its effectiveness and accuracy, offering a practical alternative to nested Monte Carlo simulations in multi-counterparty xVA analyses.
\end{abstract}

\section{Introduction}
In this paper we extend the literature on counterparty risk and funding along several directions. First, we provide a new formulation of the discounting cashflow approach of Brigo and co-authors \cite{brigo2024mild,brigo2022nonlinear,brigo2019nonlinear,pallavicini2012funding,pallavicini2011funding}  in a structural credit model, allowing for a formulation of xVA BSDEs solely based on Brownian drivers which are easily tractable from a numerical perspective. Secondly, we propose a novel deep xVA algorithm that extends \cite{gnoatto2020deep} as follows: first, by deriving the system of xVA BSDEs from the discounting approach of Brigo and co-authors, we demonstrate that the algorithm is applicable also in the incomplete market case. Moreover, we include initial margin by means of quantile regressions and provide a full multi-layer specification of the algorithm that addresses all the dependencies among the different value adjustments. Most importantly, we introduce an extension of the deep BSDE solver used in each layer that allows us to address the numerical instability introduced by default events with low probabilities. Moreover, on the theoretical side, we also extend the a posteriori error analysis of \cite{han2020convergence} to BSDEs on a bounded domain, yielding a reduced $\mathcal{O}(h^{1/4-\varepsilon})$ convergence rate for stopping‐time problems. Finally, we show the full power of the algorithm by considering portfolios of derivatives depending on a vector of several risk factors mimicking a realistic netting set between the bank and the counterparty.

The motivation for our study is given by the widespread recognition that, after the 2007--2009 financial crisis, the price of contingent claims should account for counterparty risk and the presence of collateral and multiple sources of funding. Banks and practitioners now routinely adjust the “risk-free” or “clean” values of trades to account for potential default losses, funding costs, and capital requirements. Collectively known as xVA (valuation adjustments), these corrections include credit valuation adjustment (CVA) for counterparty default risk, debt valuation adjustment (DVA) for one’s own default risk, funding valuation adjustment (FVA) for certain funding costs, margin valuation adjustment (MVA) for the funding costs of posting initial margin, and capital valuation adjustment (KVA) for the cost of holding regulatory capital. 

The literature and the market practice on xVA initially focused on the default risk of the counterparty (i.e. on the CVA) under the assumption of a unique risk-free interest rate, see \cite{duhu96,cherubini05,brigoMasetti}. With the onset of the 2007-2009 financial crisis the risk of a default of the bank started to be included e.g. in \cite{bripapa11,bricapa14} in the form of the DVA. The financial crisis induced also new debates regarding the mathematical description of the funding policy of the trading desk. Portfolio dynamics should account for funding via transactions on the repo market for the risky assets, possibly different unsecured rates for the borrowing and lending activity of the ALM desk, the remuneration of collateral in the form of variation and initial margin and the presence of cross currency bases: see \cite{hen07,hen10,pit10,brigo2015note,fushita09,fushita10,fushita10b,fushita10c,pit12,BieRut15,gnosei2021}. Credit and funding aspects were then give a unifying treatment with the proposal of several different xVA frameworks, aiming to include all possible effects and frictions in a single valuation equation. First, we mention the PDE-based contributions of Burgard and Kjaer, where the classical Black-Scholes-type replication is extended to include most valuation adjustments: \cite{bj2011,bj2013}. More general replication approaches can be found in \cite{crepey2015bilateral,crepey2015bilateral2,Crpey2016,bichuch2018arbitrage,biagini2021unified} based on BSDEs that accounts for several types of non-linearities. The already mentioned discounting approach of Brigo and co-authors was simultaneously developed and shown to be consistent with the replication approach in \cite{brigo2022nonlinear}. The important aspect of the discounting approach for our purposes is that the discounted cashflow approach only postulates the martingale property of certain gains processes without assuming replication, allowing for the derivation of xVA BSDEs also in an incomplete market setting as shown in \cite{brigo2024mild}. 

Most xVA approaches in the literature assume that the underlying credit risk model is of reduced-form type. A notable exception to this is given by \cite{BALLOTTA2019}, where a L\'evy-driven structural model is employed. In the present work we also pursue the adoption of a structural model since it can be integrated with our numerical approach in a very natural way. Another advantage of structural models is that they allow for an easy explicit linkage between credit risk and exposure modeling to capture wrong-way risk (WWR) in xVA applications. Since default is endogenously derived from a firm’s asset dynamics, structural models can readily correlate the market factors driving exposure with those that trigger default. 

One of the key challenges when using structural models for derivatives portfolios is that default events, even if rare, must be explicitly modeled. Furthermore, the accurate modeling of these rare credit events and extreme dependencies poses significant computational challenges. As a result, pricing and risk analysis often use a change of measure to streamline the representation of uncertain future events and to highlight particular scenarios. Measure change techniques reweight probabilities in a convenient way without altering real-world outcomes. For instance, large deviation theory and importance sampling rely on choosing an alternative probability measure under which rare events become more likely, to efficiently estimate tail risk. In the context of counterparty risk, \cite{brigo2016disentangling} demonstrate how carefully chosen measure changes can embed complex dependencies into the model in a tractable manner. One prominent example is the treatment of the adverse dependence between exposure and default likelihood, commonly referred to as \emph{wrong-way risk} (WWR). WWR is notoriously difficult to handle because it requires modeling the joint distribution of market exposures and credit events. Traditional simulations that capture WWR directly can be computationally intensive, and indeed regulators like Basel III historically simplified CVA risk charges by ignoring explicit WWR due to this complexity. To address this, \cite{brigo2016disentangling} introduce a change-of-measure approach that incorporates WWR into the drift of the exposure process, effectively adjusting the dynamics so that higher exposure states coincide with higher default risk. This approach entails an infinite sequence of measure changes, yet yields a practical compromise between mathematical rigor and tractability, capturing the essential impact of WWR on CVA figures. In general, measure changes are invaluable for linking physical and pricing measures (especially important in credit, where real-world default probabilities and risk-neutral default intensities can differ substantially) and for structuring complex contingent claims. In our setting, we leverage measure change techniques to ensure that our model appropriately reflects the low-probability, high-impact credit events that drive counterparty losses. We provide a conceptual discussion of these techniques, highlighting their role in tail risk modeling.

Whether we adopt a structural or a reduced-form approach, the full xVA problem is typically cast as a system of coupled conditional expectations that capture all interdependencies. In fact, this framework can be reformulated as a coupled system of high-dimensional BSDEs (or, equivalently, as PDEs via the nonlinear Feynman--Kac formula). For a selective, non-exhaustive overview of key contributions on BSDE approaches for xVA, see \textit{e.g.},  \cite{ crepey2015bilateral,bichuch2018arbitrage,brigo2019nonlinear,crepey2020capital,biagini2021unified,crepey2022positive,gnoatto2020deep,brigo2024mild,abbas2023pathwise}.

However, solving the resulting xVA BSDE system is highly challenging. Traditional numerical methods struggle with both the high dimensionality and the path-dependent features, such as credit triggers and collateral rules that are inherent in xVA. For example, although the Monte Carlo method is widely used in xVA computations, it often requires nested simulations \cite{atcdi2018} or regression techniques \cite{Cesari2009} to accurately estimate conditional exposures, making it computationally expensive. In contrast, PDE-based approaches are typically confined to low-dimensional settings due to the curse of dimensionality. In summary, conventional techniques encounter significant obstacles when tackling the coupled, high-dimensional BSDE system arising from xVA, a challenge that has spurred considerable interest in alternative solution methods within the broader scientific computing community. Recent advances in deep learning have introduced powerful tools for solving high-dimensional BSDEs and related PDEs. Deep learning-based approaches have then been applied, in different ways to xVA computations e.g. in \cite{abbas2023pathwise,atcbb2024,andersson2021deep_2,SheGre18}.

The seminal work by Han, Jentzen, and E~\cite{han2018solving} demonstrated that a deep BSDE solver can efficiently tackle nonlinear PDEs in hundreds of dimensions by recasting them as BSDEs and using neural networks as function approximators. This method leverages the universal approximation properties of deep neural networks and overcomes the curse of dimensionality for certain classes of problems. Since that breakthrough, several machine-learning-based schemes have been proposed to solve coupled BSDEs, (see  \cite{beck2020overview} for an overview), offering a promising alternative to classical methods such as Monte Carlo simulation and PDE-based solvers.

Building on these advances, \cite{gnoatto2020deep} introduced a deep xVA Solver that models the xVA problem as a coupled BSDE system. Their approach solves the BSDEs recursively, handling each adjustment (such as CVA, DVA, FVA, etc.) one at a time. The present paper extends this work by:
\begin{enumerate}
    \item[(i)] Incorporating the exchange of initial margin, thereby enabling the computation of the Margin Valuation Adjustment (MVA);
    \item[(ii)] Adopting a structural modeling framework that enhances the treatment of wrong-way risk and facilitates potential extensions to multiple counterparties;
    \item[(iii)] Simultaneously solving multiple BSDEs at once instead of repeatedly applying the method for each BSDE.
\end{enumerate}
These extensions introduce several non-standard features into the BSDE formulation, which lead to the following enhancements of the original deep BSDE solver:
\begin{enumerate}
    \item[(i)] Simultaneously solving multiple, potentially coupled BSDEs;
    \item[(ii)] Employing a change-of-measure technique that focuses numerical approximations on financially relevant regions (i.e., areas where defaults are more likely to occur);
    \item[(iii)] Incorporating a risk measure into the BSDE driver that accounts for random stopping times, thereby resulting in a system of coupled McKean--Vlasov BSDEs with random stopping times.
\end{enumerate}
In addition, we extend the deep BSDE approach to BSDEs on a bounded domain with explicitly modeled random stopping times. Although deep BSDE methods (or similar techniques) have been applied in domain-restricted settings such as barrier-option valuation (see, \textit{e.g.}, \cite{kremsner2020deep,umeorah2023barrier,ganesan2020pricing}), these works do not provide any associated error analysis. By contrast, numerous papers do offer error analyses for deep BSDE methods in the unbounded (non-stopping) case (see, \textit{e.g.}, \cite{han2020convergence,hutzenthaler2020proof,berner2020analysis,elbrachter2021dnn,Grohs2018APT,jentzen2018proof,jiang2021convergence,gnoatto2025convergence,negyesi2024generalized} for forward-type methods and \cite{fang2009novel,Balint,kapllani2024backward,hure2019some,germain2022approximation} for backward-type methods), yet they do not address domain restrictions or random stopping times. To the best of our knowledge, this paper therefore offers the first such results for deep BSDE methods under random stopping times. In particular, we adapt the a posteriori convergence analysis of Han and Long \cite{han2020convergence}, replacing the scheme of Zhang and Bender \cite{bender2007forward} with the time-discretization results of Bouchard and Menozzi \cite{bouchard2009strong}. This change accounts for boundary exits and yields a reduced $\mathcal{O}\bigl(h^{1/4-\varepsilon}\bigr)$ rate of convergence, reflecting the cost of accurately handling default events within a structural credit model.

We mention that our adoption of a structural approach fully addresses the criticism of \cite{abbas2023pathwise} on the initial deep xVA approach of \cite{gnoatto2020deep}, where the deep xVA algorithm was applied to the diffusive, pre-default BSDEs arising from a reduced-form credit model: our extension of \cite{gnoatto2020deep} can account for a full portfolio-wide view with possible links among several counterparties. By adopting a structural approach it is possible to employ our methodology even in the context of the balance-sheet xVA BSDEs of \cite{crepey2022positive} where FVA is computed over the whole bank's portfolio. 

The remainder of this paper is organized as follows. Section~\ref{sec:market_model} sets up the market environment, distinguishing defaultable and non‐defaultable risk factors through a structural model. Section~\ref{sec:expected_CF} formulates the xVA problem in terms of discounted cash flows and identifies a BSDE system. Section~\ref{sec:Val_BSDEs} introduces the reformulation of xVAs as a system of coupled BSDEs as well as practical reformulations of the BSDEs, including the measure‐change methodology to handle rare default events effectively and the variational problem, which is the foundation of the deep BSDE solver. Section~\ref{seq:discretization} discusses the temporal discretization as well as the neural‐based algorithmic procedure, including architectural details and in Section~\ref{sec:error_analysis} we present bounds for the simulation error. Finally, Section~\ref{sec:numerics} provides numerical results, and Section~\ref{sec:conclusion} concludes with possible extensions and open research topics

\section{Notation and spaces}
For $T\in(0,\infty)$ let $(\Omega,\F,\P)$ be a probability space and $\mathbb{F}=(\F_t)_{t\in[0,T]}$ be a filtration of $\F$, representing the whole available information on an underlying financial market. For any Euclidean space $\mathbb{R}^k$, $k \in \mathbb{N}$, we write $|x|$ for the Euclidean norm of $x \in \mathbb{R}^k$. We denote by $\mathbb{L}^0(\F)$ the space of $\F$-measurable random variables. Denote by $\mathbb{L}^0(\mathbb{F})$ the space of $\mathbb{F}$-progressively measurable processes, \textit{i.e.}, processes $X$ for which, the mapping $[0,t] \times \Omega \ni (s,\omega) \mapsto X_s(\omega)$ is measurable with respect to $\mathcal{B}([0,t]) \otimes \mathcal{F}_t$, for all $t \in [0,T]$.

When more convenient, we view a stochastic process $X\colon[0,T]\times\Omega\to\R^k$ as a family of random variables $X = (X_t)_{t \in [0,T]}$. For a given process $X$, let $\mathbb{F}^X = (\mathcal{F}^X_t)_{t \in [0,T]}$ be the filtration generated by $X$, where $\mathcal{F}^X_t = \sigma(X_s : 0 \le s \le t)$ and $\mathcal{F}^X = \sigma(X_s : 0 \le s \le T)$. Since $\mathcal{F}^X_t \subseteq \mathcal{F}_t$ for all $t$, we have $\mathbb{F}^X \subseteq \mathbb{F}$. A stopping time $\tau \colon \Omega \to [0,T]$ with respect to $\mathbb{F}^X$ (i.e., $\{\tau \le t\} \in \mathcal{F}^X_t$ for all $t$) is also an $\mathbb{F}$-stopping time.

We now introduce the following spaces:
$$
\mathbb{L}^2(\F_t) \coloneq \bigg\{ X \in \mathbb{L}^0(\F_t)\colon \|X\|_{\mathbb{L}^2}^2  \coloneq \mathbb{E}\bigl[ |X|^2 \,\bigr] < \infty \bigg\},
$$
$$
\mathbb{H}^2(\mathbb{F}) \coloneq \bigg\{ X \in \mathbb{L}^0(\mathbb{F}) : \|X\|_{\mathbb{H}^2}^2 \coloneq \mathbb{E}\biggl[\int_0^T |X_t|^2 \,\mathrm{d}t\biggr] < \infty \bigg\},
$$
$$
\mathbb{S}^2(\mathbb{F}) \coloneq \bigg\{ X \in \mathbb{L}^0(\mathbb{F}) : \|X\|_{\mathbb{S}^2}^2  \coloneq \mathbb{E}\biggl[\sup_{t \in [0,T]} |X_t|^2 \biggr] < \infty \bigg\}.
$$

For any $x \in \mathbb{R}$, we define $x^+ = \max\{x,0\}$ and $x^- = \max\{-x,0\}$. We write $\mathbbm{1}_{\{\cdot\}}$ for the indicator function. We denote by $\odot$ and $\oslash$ component wise multiplication and division, respectively.

Throughout this paper, we assume that all formally stated conditions in the \texttt{Assumption} environments hold unless explicitly stated otherwise. The only exception is Section~\ref{sec:error_analysis}, which is intended to be read as a self-contained analysis.
\section{A market model with defaultable counterparties}\label{sec:market_model}
The focus of this investigation is a portfolio of contracts, typically a netting set, between two entities, a \emph{bank}, denoted by $\B$, and a \emph{counterparty}, denoted by 
$\C$. This portfolio consists of $P \in \mathbb{N}$ derivative contracts, each with a respective maturity $T_1,T_2,\ldots,T_P\in(0,\infty)$. 
We define the maturity of the entire portfolio $T=\max\{T_1,T_2,\ldots,T_P\}$ representing the last maturity date among all contracts in the set.  

We classify \emph{risk factors} into two categories, \emph{non-defaultable} and \emph{defaultable}. Non-defaultable risk factors correspond to \emph{tradeable} or \emph{non-tradeable} assets associated with the $P$ derivative contracts, either \emph{directly} as the underlying assets or \emph{indirectly} through related attributes such as volatility and/or interest rate. 

We consider $d\in\N$ non-defaultable risk factors indexed by $\mathcal{A} = \{1,2,\ldots,d\}$. For each $j\in\{1,2,\ldots,P\}$
let $\mathcal{I}_j\subseteq\mathcal{A}$
denote the subset of non-defaultable risk factors, directly or indirectly, associated with derivative $j$, with $d_j=|\mathcal{I}_j|$ representing the number of risk factors in this subset. We define $\mathcal{J}_j = \sum_{i=0}^{j-1} d_i$ (with $d_0$ taken to be 1), as an index that keeps track of all risk factors across the entire portfolio. Concretely, suppose we form a single vector whose entries collect all risk factors for the $P$ derivatives (noting that any risk factor associated with multiple derivatives appears multiple times in this vector). Then, the block of entries corresponding to derivative $j$ is indexed by $\{\mathcal{J}_j, \mathcal{J}_j + 1, \ldots, \mathcal{J}_j + d_j - 1\}\eqdef \mathcal{K}_j$. 

Defaultable risk factors consist of the asset processes of the bank and the counterparty, which form the basis of our structural model. We define two index sets for the defaultable risk factors, where $\widebar{\mathcal{E}}=\{1,2\}$ is used when considering only the defaultable risk factors, and $\mathcal{E}=\{d+1,d+2\}$ is used when all risk factors are included. To account for both non-defaultable and defaultable risk factors, we define $\mathcal{Q} \defeq \mathcal{A} \cup \mathcal{E}=\{1,2,\ldots,d+2\}$. 

\begin{example}
    Consider a portfolio of 2 derivatives with 4 non-defaultable risk factors, where risk factors 1-3 are tradeable and risk factor 4 is the stochastic volatility of risk factor 1 and hence, non-tradeable. Derivative 1 is written on assets 1,2 and, 3 and derivative 2 is written on asset 3. This setting implies: $d=4$, $P=2$, $\mathcal{I}_1=\{1,2,3,4\}$, $\mathcal{I}_2=\{3\}$, $d_1=4$, $d_2=1$, $\mathcal{Q}=\{1,2,3,4,5,6\}$, $\mathcal{J}_1=1$, $\mathcal{J}_2=5$, $\mathcal{K}_1=\{1,2,3,4\}$, $\mathcal{K}_2=\{5\}$. 
\end{example}

In this paper, we model both non-defaultable and defaultable risk factors using diffusion-type stochastic differential equations (SDEs). In this section, we first introduce the SDEs used to model non-defaultable and defaultable assets separately. We then combine these equations into a unified SDE that captures the dynamics of all risk factors. Finally, we outline the formal assumptions required to ensure the existence and uniqueness of a strong solution to the resulting risk factor SDE.

\subsection{Non-defautable risk factors}
The stochastic evolution of the non-defaultable risk factors is driven by a $d$-dimensional standard Brownian motion $\widehat{W} = (\widehat{W}_t)_{t \in [0,T]}$ defined on the filtered probability space $(\Omega, \mathcal{F}, \mathbb{\widehat{F}}, \mathbb{P})$. Here $\widehat{\F}\defeq\F^{\widehat{W}}\subset \F$ denotes the $\sigma$-algebra generated by $\widehat{W}$ and $\widehat{\mathbb{F}}\defeq \mathbb{F}^{\widehat{W}}$ is the filtration generated by $\widehat{W}$. For $t \in [0,T]$ and $i,j \in \mathcal{A}$, the correlation is given by
\begin{equation}\label{cor_S}
        \langle \widehat{W}^i, \widehat{W}^j \rangle_t = \int_0^t \widehat{\rho}^{i,j}(s)\,\mathrm{d}s.
\end{equation}
For $\widehat{x}_0\in\R^d$, Let $\widehat{b} = (\widehat{b}^1, \ldots, \widehat{b}^d)^\top$ and $\widehat{\sigma} = (\widehat{\sigma}^1,\ldots,\widehat{\sigma}^d)^T$, where for each $j \in \mathcal{A}$, $\widehat{b}^j : [0,T]\times\mathbb{R}^d \to \mathbb{R}$ and $\widehat{\sigma}^j : [0,T]\times\mathbb{R}^d \to (0,\infty)$. For $t \in [0,T]$ and $j \in \mathcal{A}$, the dynamics of the $j$-th risk factor are governed by the SDE
\begin{equation}\label{SDE_S}
    \widehat{X}_t^j = x_0^j + \int_0^t \widehat{b}^j(s,\widehat{X}_s)\,\mathrm{d}s + \int_0^t \widehat{\sigma}^j(s,\widehat{X}_s)\,\mathrm{d}\widehat{W}_s^j.
\end{equation}
We denote by $\widehat{X}^j = (\widehat{X}_t^j)_{t \in [0,T]}$ the process describing the $j$-th non-defaultable risk factor, and by $\widehat{X} = (\widehat{X}^1, \widehat{X}^2, \ldots, \widehat{X}^d)^\top$ the $d$-dimensional asset vector. As stated above, the bank's derivative portfolio with the counterparty consists of $P$
derivative contracts, each written on a subset of the $d$ underlying risk factors. For each $j\in\{1,2,\ldots,P\}$, the price vector of the risk factors directly or indirectly associated with derivative $j$ is defined as $\widehat{X}^{\mathcal{I}_j}=\{\widehat{X}^i\}_{i\in \mathcal{I}_j}$.

We introduce the theoretical \emph{instantaneous risk-free rate} $r$, assumed to be an $\widehat{\mathbb{F}}$-progressively measurable process with integrable paths where $r_t \in \mathbb{R}$ for $t \in [0,T]$. For $s,t \in [0,T]$ with $s \leq t$, we define
\begin{equation*}
    D_{s,t}(r) \coloneq \exp\!\biggl(-\int_s^t r_u\,\mathrm{d}u\biggr),
\end{equation*}
and for $s > t$, we set $D_{t,s}(r) \coloneq 0$. This rate is solely used for risk-free valuation of future cashflows and does not represent an actual borrowing/lending rate.  
\begin{assumption}[Interest rate assumption]\label{Ass:ir}
    An interest rate $r$ (including but not limited to the risk-free rate) is either:
    \begin{enumerate}
        \item Deterministic and continuous in time;
        \item Described by a diffusion type SDE.
    \end{enumerate}
    In the latter scenario, $r$ is assumed to be one of the non-defaultable risk factors, given by one of the components of $\widehat{X}$. 
\end{assumption}
\begin{assumption}[Discount condition]\label{Ass:discount}
 For $t,s\in[0,T]$, $D_{t,s}(r)$ has bounded moments.
\end{assumption}
Note that a sufficient condition for Assumption \ref{Ass:discount} to hold is that $r\in\mathbb{S}^2(\mathbb{\widehat{F}})$, which is true under standard conditions if $r$ satisfies an SDE.

\subsection{Defaultable risk factors}
We adopt a structural default model, in which the value of the assets of the bank and the counterparty, respectively are modelled by stochastic processes. A default is triggered when the value of the assets of either the bank, or the counterparty reaches a pre-defined threshold. The threshold is related to the company's debt structure, for instance, one common alternative for the threshold is the company's short term debt plus half of its long term debt, see \textit{e.g.,} \cite{crosbie2002modeling}. 

To model the banks and the counterparty's asset values we employ the following model. Let $\widebar{W}=(\widebar{W}_t)_{t\in[0,T]}$ be a 2-dimensional standard Brownian motion defined on $(\Omega, \F,\mathbb{\widebar{F}},\P)$, where $\mathbb{\widebar{F}}\defeq \mathbb{F}^{\widebar{W}}$. The correlation structure is given by 
\begin{equation}\label{cor_BC}
        \langle \widebar{W}^1,\widebar{W}^2 \rangle_t = \int_0^t \widebar{\rho}(s)\,\mathrm{d}s.
\end{equation}
At this point, we are ready to define the asset processes for the bank and the counterparty. For $\widebar{x}_0\in\R^2$, and $i\in\widebar{\mathcal{E}}$ let $\widebar{b}^i\colon [0,T]\times\R\to\R$, and $\widebar{\sigma}^i\colon[0,T]\times\R\to(0,\infty)$ be the drift and diffusion coefficients of the following SDEs representing the asset values for the bank and the counterparty
\begin{equation}\label{SDE_X_BC}
\widebar{X}_t^i=\widebar{x}_0^i+\int_0^t\widebar{b}^i(s,\widebar{X}_s^i)\d s + \int_0^s\widebar{\sigma}^i(s,\widebar{X}_s^i)\d \widebar{W}^i_s.
    \end{equation}
For $i\in\widebar{\mathcal{E}}$, the default thresholds are modeled by deterministic, time-dependent, sufficiently smooth functions $\xi_t^i\colon[0,T]\to(0,\infty)$ and the default times are given by the following $\mathbb{\widebar{F}}$-stopping times \begin{equation}\label{tau}
    \tau^\B=\inf\{t\in(0,T]\colon \widebar{X}_t^1\leq \xi_t^1\},\quad \tau^\C=\inf\{t\in(0,T]\colon \widebar{X}_t^2\leq \xi_t^2\}.
\end{equation} 
Moreover, define $\tau=\tau^\B\wedge\tau^\C$.

\subsection{Merging non-defaultable and defaultable risk factors}
As stated in the introduction, one of our main objectives is to formulate the xVA problem as a decoupled Markovian forward-backward stochastic differential equation (FBSDE). In the clean valuation setting, we consider only non-defaultable entities, and the forward SDE is given by \eqref{SDE_S}. However, to incorporate default-risk, we must merge \eqref{SDE_S} with \eqref{SDE_X_BC}. In this section, we present the resulting combined forward SDE and state the conditions ensuring the existence and uniqueness of a strong solution. This new, joint $d+2$-dimensional SDE forms our full state space. 

Our aim is to define the process $ X \coloneqq \text{Concat}(\widehat{X}, \widebar{X}) $ and to express an SDE for $ X $. To accurately account for \emph{wrong-way risk}, we must incorporate a full correlation structure among $(\widehat{W}, \widebar{W})$.

To achieve this, we define the overall filtration as the product sigma field $\mathbb{F} \coloneqq \widehat{\mathbb{F}} \otimes \widebar{\mathbb{F}}$. The concatenation of the Brownian motions $\widehat{W}$ and $\widebar{W}$ is then defined as the $(d + 2)$-dimensional standard Brownian motion $ W = \text{Concat}(\widehat{W}, \widebar{W}) $, with components given by
\begin{equation*}
  W^i = \widehat{W}^i, \text{ for }  i \in \mathcal{A}, \quad W^{d+i} = \widebar{W}^i, \text{ for } i \in \widebar{\mathcal{E}}.
\end{equation*} The correlation between the $ d + 2 $ components of $ W $ is governed by
\begin{equation*}
\langle W^i, W^j \rangle_t = \int_0^t \rho^{i,j}(s) \, ds, \quad \text{for } i, j \in \mathcal{A} \cup \mathcal{E},
\end{equation*}
with a structure specified as follows:
\begin{itemize}
\item  For $ i, j \in \mathcal{A} $, $\rho^{i,j} \coloneqq \widehat{\rho}^{i,j}$, represent the correlations among non-defaultable risk factors;
\item For $ i, j \in \widebar{\mathcal{E}} $, $\rho^{d+1, d+2} = \rho^{d+2, d+1} \coloneqq \widebar{\rho}$, representing the correlation between the defaultable risk factors;
\item For $ i \in \mathcal{A} $ and $ j \in \mathcal{E} $, $\rho^{i,j}$ captures the correlation between the non-defaultable and defaultable risk factors, reflecting wrong-way risk.
\end{itemize}
 For each $j\in\{1,2,\ldots,P\}$, the components of the Brownian motion directly or indirectly associated with derivative $j$ is denoted by $W^{\mathcal{I}_j}=\{W^i\}_{i\in \mathcal{I}_j}$.

To obtain an equation for $X$, we define the initial condition, drift and volatility coefficients by the concatenations $x_0\defeq\text{Concat}(\widehat{x}_0,\widebar{x}_0)$, $b\defeq\text{Concat}(\widehat{b}^1,\ldots,\widehat{b}^d,\widebar{b}^1,\widebar{b}^2)$ and
$\sigma\defeq\text{Concat}(\widehat{\sigma}^1,\ldots,\widehat{\sigma}^d,\widebar{\sigma}^1,\widebar{\sigma}^2)$. For $t\in[0,T]$ and $i\in\mathcal{Q}$, the full state process is given by
\begin{equation}
\begin{cases}
    X_t^i=x_0^i+\int_0^tb^i(s,\widehat{X}_s)\d s + \int_0^t\sigma^i(s,\widehat{X}_s)\d W_s^i,\quad \text{ for } i\in\mathcal{A},\\
    X_t^i=x_0^i+\int_0^tb^i(s,X_s^i)\d s + \int_0^t\sigma^i(s,X_s^i)\d W_s^i,\quad \text{ for } i\in\mathcal{E}
    \end{cases}
\end{equation}
or in vector form
\begin{equation}\label{eq:full_X}
    X_t=x_0+\int_0^tb(s,X_s)\d s+\int_0^t\sigma(s,X_s)\odot\d W_s.
\end{equation}
We allow the drift and diffusion coefficients of the non-defaultable risk factors to depend on the entire non-defaultable risk factor process. This flexibility is essential to account for scenarios where the direct underlying risk factors are influenced by indirect non-defaultable risk factors, such as stochastic volatility.
\begin{assumption}[Conditions on the state process]\label{Ass:X_cond}
Let $x_0 \in \mathbb{R}^{d+2}$ be the initial condition. Assume that the drift $b:[0,T]\times \mathbb{R}^{d+2}\to\mathbb{R}^{d+2}$ and the diffusion $\sigma:[0,T]\times \mathbb{R}^{d+2}\to\mathbb{R}^{d+2}$ satisfy the following conditions:
\begin{enumerate}
    \item 
 \textbf{Lipschitz continuity of $(b,\sigma)$ in $x$}. \newline There exists a constant $L > 0$ such that for all $t \in [0,T]$ and all $x_1,y_2 \in \mathbb{R}^{d+2}$,
\begin{equation*}
|b(t,x_1)-b(t,x_2)| + |\sigma(t,x_1)-\sigma(t,x_2)| \leq L|x_2-x_2|.
\end{equation*}
\item \textbf{Linear growth of $(b,\sigma)$ in $x$}.\newline There exist constants $K > 0$ and $\alpha \geq 0$ such that for all $t \in [0,T]$ and all $x \in \mathbb{R}^{d+2}$,
\begin{equation*}
|b(t,x)| + |\sigma(t,x)| \leq K(1 + |x|^\alpha).
\end{equation*}

\item \textbf{Hölder-$\tfrac{1}{2}$ continuity of $(b,\sigma)$ in $t$}.\newline There exists a constant $H > 0$ such that for all $x \in \mathbb{R}^{d+2}$ and all $0 \leq s < t \leq T$,
\begin{equation*}
|b(t,x)-b(s,x)| + |\sigma(t,x)-\sigma(s,x)| \leq H|t-s|^{1/2}.
\end{equation*}
\end{enumerate}

Additionally, the correlation structure $\rho^{i,j}:[0,T]\to(-1,1)$ for $i,j \in \mathcal{Q}$ is assumed to be measurable and satisfies
\begin{equation*}
\int_0^T (1-\rho^{i,j}(t)^2)^{-1}\,dt < \infty.
\end{equation*}
\end{assumption}Under Assumptions \ref{Ass:X_cond}, the SDE \eqref{eq:full_X} admits a unique strong solution $X \in \mathbb{S}^2(\mathbb{F})$. Moreover, \eqref{SDE_S} admits a unique strong solution $\widehat{X} \in \mathbb{S}^2(\mathbb{\widehat{F}})$.

\section{Portfolio valuation based on expected cashflows}\label{sec:expected_CF}
In this section we revisit the discounted cashflow approach to xVA of Brigo and co-authors (see e.g. \cite{brigo2024mild}) under the previous assumption of a structural model for credit risk. In the first part, we focus on the clean value, and in the second part we adjust the cashflow for default risk, funding costs/benefits and costs/benefits associated with initial margin and collateral (variation margin) accounts. In turn, this gives rise to an adjusted portfolio valuation.

\subsection{Clean valuation}
For the \emph{clean cashflow}, we only take the contractual derivative cashflows into account. For a derivative $j\in\{1,2,\ldots,P\}$, we define the cumulative stream of contractual cashflows by an $\widehat{\mathbb{F}}$-adapted (in fact, even $\mathbb{F}^{\widehat{X}^{\mathcal{I}_j}}$-adapted), càdlàg process $A^j=(A^j_t)_{t\in[0,T]}$ of finite variation. The process is initialized at $A_0^j=0$ and since $T_j$ is the maturity of the contract, it holds that $\d A_t^j=0$ for $t\geq T_j$, as no further cashflows occur beyond this point. Moreover, we denote by $A=(A_t)_{t\in[0,T]}$, the total contractual cashflow of the portfolio with $A=\sum_{j=1}^dA^j$. We denote the jumps of the process $A^j$ at $t \in [0, T]$ by $\Delta A_t^j = A_t^j - A_{t-}^j$, where $A_{t-}^j$ represents the left limit of $A^j$ at $t$. Similarly, the jumps of the aggregate process $A$ are given by $\Delta A_t = A_t - A_{t-}$. For $t\in[0,T)$ the total discounted cashflows on $(t,T]$ is then given by   
\begin{equation}\label{clean_CF_con}
\big(\widehat{\text{CF}}^\text{con}_{(t,T]}\big)^j\coloneq \int_{(t,T]} D_{t,s}(r)\d A_s^j,
\end{equation}
and for $t\geq T$, we define $\big(\widehat{\text{CF}}^\text{con}_{(t,T]}\big)^j\defeq 0$.
The total discounted cashflow for the entire portfolio is given by $\widehat{\text{CF}}^\text{con}_{(t,T]}\defeq \sum_{j=1}^P\big(\widehat{\text{CF}}^\text{con}_{(t,T]}\big)_j$. We assume that the bank is the seller of the derivatives and hence, from the bank's perspective the \emph{clean value} at $t\in[0,T]$ is given by
\begin{equation}\label{clean_V}
\widehat{V}_t\coloneq\E\big[-\widehat{\text{CF}}^\text{con}_{(t,T]}|\F_t\big],
\end{equation}
and for $t$ in the complement of $[0,T]$, we define $\widehat{V}_t\coloneq 0$. Similarily, for a derivative $j\in\{1,2,\ldots,P\}$, \begin{equation}\label{clean_V_j}
\widehat{V}_t^j\coloneq\E\Big[-\big(\widehat{\text{CF}}^\text{con}_{(t,T]}\big)^j\big|\F_t\Big].
\end{equation}
Clearly, by the linearity of conditional expectations, it holds that $\widehat{V}=\sum_{j=1}^P\widehat{V}^j$.
\begin{lemma}\label{lemma:cf_con_L2}
    Let $r\in\mathbb{S}^2(\widehat{\mathbb{F}})$ and for $j\in\{1,2,\ldots,J\}$ $A^j$, be an $\widehat{\mathbb{F}}$-adapted càdlàg process with finite variation. Then for $t\in[0,T]$ it holds that    \begin{equation*}\widehat{\text{CF}}^\text{con}_{(t,T]}\defeq\sum_{j=1}^P\int_{(t,T]} D_{t,s}(r)\d A_s^j\in\mathbb{L}^2(\widehat{\F}_T)\text{ and, }\widehat{V}\in\mathbb{S}^2(\mathbb{\widehat{F}}).\end{equation*}
    \begin{proof}
    By Cauchy–Schwarz,  
\begin{equation*}
    \mathbb{E}\biggl[\biggl(\int_{(t,T]}D_{t,s}(r)\,dA_s^j\biggr)^2\biggr]
\le \mathbb{E}\biggl[\biggl(\int_{(t,T]}D_{t,s}(r)^2|dA_s^j|\biggr)\biggl(\int_{(t,T]}|dA_s^j|\biggr)\biggr].
\end{equation*}
Since $A^j$ is finite variation, $\int_{(t,T]}|\d A_s^j|$ is almost surely finite. Also, $r\in \mathbb{S}^2(\widehat{\mathbb{F}})$ ensures $\mathbb{E}[D_{t,s}(r)^2]<\infty$ uniformly in $s$. Thus, by Fubini–Tonelli and dominated convergence, the right-hand side is finite, proving $\int_{(t,T]}D_{t,s}(r)\d A_s^j \in \mathbb{L}^2(\widehat{\F}_T)$. It then follows immediately that $\widehat{V}\in\mathbb{S}^2(\widehat{\mathbb{F}})$.\newline\newline
    \end{proof}
\end{lemma}
\begin{example}[European Call Option]\label{Ex:Call_Option}
Assume that we have a European call option, written on a single asset $\widehat{X}=(\widehat{X}_t)_{t\in[0,T]}$, with maturity $T_1<T$ and strike price $K$, then $A_t=\I_{\{t\geq T_1\}}(\widehat{X}_{T_1}-K)^+$. Note that $A$ is an $\widehat{\mathbb{F}}$-adapted càdlàg process, and under mild assumptions on $\widehat{X}$, (\textit{e.g}., continuity or càdlàg properties), 
$A$ has finite variation, as it only consists of a single jump at $t=T_1$. 
\end{example}

\subsection{Adjusted portfolio valuation}
In this section, we summarize and explain all cashflows that may occur during the lifespan of the portfolio. For each type of cashflow $Y\in\{\text{con},\text{col},\text{hed},\text{fun},\text{def}\}$, we define $\text{CF}_{(t,\tau\wedge T]}^Y$ as the random variable representing the cumulative cashflows of type $Y$ that occur over the time interval $(t,\tau\wedge T]$. Each of these cashflow processes is assumed to be constructed from $\mathbb{F}$-adapted, càdlàg processes of finite variation, ensuring that $\text{CF}_{(t,\tau\wedge T]}^Y$ is $\mathcal{F}_{\tau\wedge T}$-measurable.

If $\tau\wedge T \leq t$, the interval $(t,\tau\wedge T]$ is empty and we set $\text{CF}_{(t,\tau\wedge T]}^Y \defeq 0$. Thus, $\text{CF}_{(t,\tau\wedge T]}^Y$ is well-defined for all $\omega\in\Omega$ and $t\in[0,\tau \wedge T)$, and captures the cashflows of type $Y$ that have occurred strictly after time $t$ and up to $\tau \wedge T$.
\vspace{0.25cm}\newline
\noindent\textbf{Contractual derivative cashflow:}\newline
We define the defaultable contractual cashflow process as
\begin{equation*}
\widebar{A}_t = \I_{\{t<\tau\}}A_t + \I_{\{t\geq\tau\}}A_{\tau-}.
\end{equation*}
Here $A_{\tau-}$ represents the cumulative contractual payments made strictly up to, but excluding, the default time. This reflects the assumption that neither party is obligated to fulfill promised payments at the exact moment of default. For a more detailed discussion on this modeling approach and its implications, we refer to \cite{brigo2010dangers}.

For $t\in[0,\tau\wedge T]$, the risky contractual cashflow is given by
\begin{equation}\label{CF_con}
\text{CF}^{\text{con}}_{(t,\tau\wedge T]}\coloneq  \int_{(t,\tau\wedge T]} D_{t,s}(r)\d \widebar{A}_t.
\end{equation}
This formulation accounts for the possibility of default, ensuring that the valuation accurately reflects the associated risks and terminates cashflows upon default.
\vspace{0.25cm}\newline
\textbf{Costs and benefits associated with collateral:}\newline
Collateral plays a crucial role in managing counterparty credit risk in derivative transactions. By exchanging collateral, parties protect themselves against potential losses arising from default or adverse market movements. There are two main types of collateral in this context, Variation Margin (VM) and Initial Margin (IM). Variation margin addresses current exposure by reflecting the daily fluctuations in the mark-to-market of the derivative contract, whereas initial margin is designed to cover potential future exposure over a specified margin period of risk (MPR). Both types of collateral have associated costs and benefits, which we explore in the following sections.
\vspace{0.25cm}\newline
\noindent\textit{Variation Margin}\newline
Denote by $C=(C_t)_{t\in[0,T]}$, the variation margin process. We assume that the variation margin can be \emph{rehypothecated}, meaning the receiving party can reuse the collateral for other purposes. Hence, the variation margin is assumed to be either posted or received, and can therefore take both positive and negative values. For $C_t>0$, the bank is the taker of variation margin, and for $C_t<0$, the bank is the poster of variation margin. For simplicity, we assume that $C$ is a Lipschitz-continuous function of $\widehat{V}$. 

The interest rate on the variation margin may be different depending on which party is the poster and taker and we set the rate to
\begin{equation*}
    r^c_t=\I_{\{C_t\geq 0\}}r^{c,b}_t +\I_{\{C_t< 0\}}r^{c,l}_t.
\end{equation*}
Here $r^{c,b}=(r_t^{c,b})_{t\in[0,T]}$ and $r^{c,l}=(r_t^{c,l})_{t\in[0,T]}$ are $\mathbb{F}$-adapted, càdlàg processes describing the instantaneous borrowing and lending rates for collateral, respectively.
For $t\in[0,\tau\wedge T]$, this results in the following cashflow attributable to the variation margin
\begin{equation}\label{CF_col}
\text{CF}^\text{col}_{(t,\tau\wedge T]}\coloneq \int_t^{\tau\wedge T} D_{t,s}(r)(r_s-r^c_s)C_s\d s.
\end{equation}
\noindent
{\textit{Initial Margin}}\newline
Initial margin is imposed on a derivative deal to mitigate the risk of price movements during the period between the day one party defaults and the day the positions are closed or collateral is liquidated. As stated above, this period is known as the margin period of risk. The MPR, usually 10 or 20 days, reflects the time it takes to manage the default, including notifying parties, liquidating positions, and settling remaining obligations. Unlike variation margin, initial margin is typically not \emph{rehypothecable}, meaning it cannot be reused by the receiving party for other purposes. This restriction ensures that the collateral remains securely segregated and readily available if a default occurs.

We denote initial margin provided by the counterparty, and hence received by the bank, by $\text{IM}^\text{FC}=\big(\text{IM}_t^\text{FC}\big)_{t\in[0,T]}$. Similarly, we denote initial margin provided by the bank, and received by the counterparty by $\text{IM}^\text{TC}=\big(\text{IM}_t^\text{TC}\big)_{t\in[0,T]}$. Initial margin is defined as a risk measure, for instance value at risk (VaR) or expected shortfall (ES), of adverse price movements of the exposure (clean value of the portfolio) over the MPR. \begin{equation}
\text{IM}_t^\text{FC}=\varrho\Big[\big(\widehat{V}_{t+\text{MPR}_t}-\widehat{V}_t\big)^+\big|\widehat{\F}_t\Big],\quad \text{IM}_t^\text{TC}=-\varrho\Big[\big(\widehat{V}_{t+\text{MPR}_t}-\widehat{V}_t\big)^-\big|\widehat{\F}_t\Big].
\end{equation}
Here, $\text{MPR}_t=\min\{\text{MPR},T-t\}$, $\varrho[\cdot|\mathcal{F}]$, represents a generic risk measure conditioned on $\mathcal{F}$. Since initial margin is not rehypothecable, and must be held in a segregated account $\text{IM}^\text{FC}$ and $\text{IM}^\text{TC}$ are both posted simultaneously and cannot be netted against each other. We introduce the borrowing and lending rates for initial margin $r^{\text{IM},b}=(r^{\text{IM},b}_t)_{t\in[0,T]}$ and $r^{\text{IM},l}=(r^{\text{IM},l}_t)_{t\in[0,T]}$, which are assumed to be $\mathbb{F}$-adapted, càdlàg processes. We assume that the counterparty is remunerated with the rate $r^{\text{IM},b}$ for posting $\text{IM}^\text{FC}$ and the bank is remunerated with the interest rate $r^{\text{IM},l}$ for posting $\text{IM}^\text{TC}$.

 For $t\in[0,\tau\wedge T]$, the cashflow associated with funding costs of the initial margin is given by
\begin{equation}
    \text{CF}_{(t,\tau\wedge T]}^{\text{IM}}\coloneqq \int_t^{\tau\wedge T}D_{s,t}(r)\big((r_s-r_s^{\text{IM},l})\text{IM}_s^\text{TC} -  r_s^{\text{IM},b}\text{IM}_s^\text{FC}\big)\d s.
\end{equation}
\begin{assumption}[Risk measure]\label{Ass:risk_measure}
    We assume that the risk measure $\varrho$ is monotone and translation-invariant.
\end{assumption}
Note that under Assumption \ref{Ass:risk_measure}, it holds for $t\in[0,T]$ that $\varrho[\cdot|\widehat{\F}_t]\colon\mathbb{L}^2(\widehat{\F})\to\mathbb{L}^2(\widehat{\F}_t)$ and in turn $\text{IM}^\text{TC},\text{IM}^\text{FC}\in\mathbb{S}^2(\widehat{\mathbb{F}})$. 
\vspace{0.25cm}\newline
\textbf{Cashflows at default of one of the parties:}\newline
The actions both parties must take in the event of a default are governed by a \emph{close-out agreement}, which outlines procedures for early termination of the contract and settling outstanding obligations. This agreement covers various aspects, including the conditions under which a default can be declared, such as bankruptcy, missed payments, or breach of contract.

The most important part of this agreement, for our purposes, is the specification of the \emph{close-out amount}, the amount the parties agree to settle in the case of default. We denote this amount by $\theta_\tau$, which typically depends on several factors, including the time of default, which party is in default, the clean value of the portfolio, and any adjustments to portfolio valuations.

From the perspective of $t\in[0,\tau]$, the discounted future cashflow at default is then given by:
\begin{equation}\label{CF_def}
\text{CF}^\text{def}_{(t,\tau\wedge T]}\coloneq D_{t,\tau}(r)\theta_\tau\I_{\{\tau\leq T\}},
\end{equation}
and on the complement of $\{[0,\tau\wedge T]\}$, we define $\text{CF}^{\text{def}}_{t,\tau\wedge T}\coloneq0$. A common approach is to use a close-out amount based on the clean portfolio value. In the event of a default by either the counterparty or the bank, a \emph{recovery payoff}, denoted by $R_\tau$, is made. Along with the recovery payoff, the non-defaulting party retains the variation margin and the portion of the initial margin that was in their possession, \textit{i.e.,} for $\tau<T$ \begin{equation*}
\theta_\tau=R_\tau + C_\tau + \I_{\{\tau=\tau^\mathcal{C}\}}\text{IM}^\text{TB}_\tau + \I_{\{\tau=\tau^\mathcal{B}\}}\text{IM}^\text{TC}_\tau.    
\end{equation*}
For $\tau<T$, the recovery payoff is given by \begin{align*}
    R_\tau
    =&\I_{\{\tau=\tau^\mathcal{C}\}}\big((1-\text{LGD}^\mathcal{C})(Q_\tau-C_\tau-\text{IM}^\mathrm{FC}_\tau)^+ -(Q_\tau-C_\tau-\text{IM}^\mathrm{FC}_\tau)^-\big) \\
    &- \I_{\{\tau=\tau^\mathcal{B}\}}\big((1-\text{LGD}^\mathcal{B})(Q_\tau-C_\tau-\text{IM}^\mathrm{TC}_\tau)^- -(Q_\tau-C_\tau-\text{IM}^\mathrm{TC}_\tau)^+\big)
\end{align*}
Here, $\mathrm{LGD}^\mathcal{B}\in[0,1]$ and $\mathrm{LGD}^\mathcal{C}\in[0,1]$ are the losses given default of the bank and the counterparty, respectively. Finally, this yields
\begin{align}\label{clean_close_out}
\theta_\tau=Q_\tau - \I_{\{\tau=\tau^\mathcal{C}\}}\text{LGD}^\mathcal{C}(Q_\tau-C_\tau-\text{IM}_\tau^\text{FC})^+ + \I_{\{\tau=\tau^\mathcal{B}\}}\text{LGD}^\mathcal{B}(Q_\tau-C_\tau - \text{IM}_\tau^\text{TC})^-.
\end{align}
In the above we have defined the exposure by $Q_\tau\defeq\widehat{V}_\tau+\Delta A_\tau$ which is used instead of $\widehat{V}_\tau$ because $Q_\tau$ reflects the clean portfolio value just before default, incorporating all future discounted cashflows up to $\tau$. Since no contractual cashflows are expected exactly at default, any realized jump $\Delta A_\tau$ at $\tau$ must be explicitly added to the exposure, ensuring consistency with the portfolio’s obligations. Another alternative would be to consider a \emph{replacement close-out}, which would represent the cost of setting up a new similar portfolio with another counterparty, see for instance \cite{antonelli2023analysis}.
\begin{assumption}[Deterministic loss given default]
We assume that $\emph{LGD}^\B$ and $\emph{LGD}^\C$ are deterministic constants.
\end{assumption}
This assumption simplifies the model, making it easier to analyze and interpret, but can be relaxed if needed.
\vspace{0.25cm}\newline
\textbf{Costs and benefits from funding and hedging accounts:}\newline
Following \cite{brigo2019nonlinear}, the treasury funding costs related to the derivative stem from borrowing and lending at the rates $r^{f,b}$ and $r^{f,l}$ to fund the trading activity. 
However, only the uncollateralized portion of the portfolio requires hedging. Furthermore, the initial margin posted by the counterparty can be excluded from the funding costs, as it is held in a segregated account and does not need to be financed. In light of these considerations, by assuming that all assets are traded in the repo market, we define the funding process $F=(F_t)_{t\in[0,\tau\wedge T]}$ as 
\begin{equation}\label{eq:F_full}
    F_t:=V_t(\varphi)+ \Delta A_t\I_{\{t=\tau\}}-C_t-\text{IM}_t^\text{TC}
\end{equation}
The resulting cashflow from treasury funding the hedging portfolio is then for $t\in[0,\tau\wedge T]$ given by\begin{align}\label{eq:CF_fun}\begin{split}
\text{CF}_{(t,\tau\wedge T]}^{\text{fun}}(V(\varphi))\coloneqq&\int_t^{\tau\wedge T}D_{t,s}(r)\Big[(r-r_s^{f,b})\big(V_s(\varphi)+ \Delta A_t\I_{\{t=\tau\}}-C_t-\text{IM}_s^\text{TC}\big)^+\\
&-(r-r_s^{f,l})\big(V_s(\varphi)+ \Delta A_t\I_{\{t=\tau\}}-C_t-\text{IM}_s^\text{TC}\big)^-\Big]\d s,\end{split}\end{align} where the supersqript $V(\varphi)$ indicates that these cashflows are calculated under the assumption that the adjusted portfolio value is given by $V(\varphi)$. This does not imply that the adjusted portfolio value is replicable; indeed, we do not assume market completeness. 

By expressing the integrand in terms of these positive and negative parts, we distinguish between periods when the trasury desk of the bank requires external funding (the positive part, representing a cost) and periods when it effectively generates surplus funds (the negative part, representing a benefit). Thus, this decomposition naturally separates the funding costs from the funding benefits. Finally, regarding the costs and benefits from repo funding, we assume that the repo rates coincide with the risk-free rate $r$, which implies that there is no cash-flow contribution from the hedging procedure. 

\vspace{0.25cm}
\noindent\textbf{Adjusted portfolio value as a sum of expected cashflows:}\newline
We define the adjusted portfolio value as the conditional expectation of all cash flows from the bank's perspective. However, since the expression involves the adjusted portfolio value implicitly on both sides, we must first ensure that the expression is well-defined. This is established in the following lemma and  proposition.

\begin{lemma} For $t\in[0,\tau\wedge T)$ it holds that
    \begin{equation*}\emph{CF}^\emph{con}_{t,\tau\wedge T}, \emph{CF}^\emph{def}_{t,\tau\wedge T},  \emph{CF}^\emph{col}_{t,\tau\wedge T},  \emph{CF}^{\emph{IM}}_{t,\tau\wedge T} \in\mathbb{L}^2(\F_{\tau\wedge T}).\end{equation*}
    \begin{proof}
    The presence of the stopping time $\tau$ does not fundamentally alter the argument of Lemma \ref{lemma:cf_con_L2}. The same reasoning applies, because for each fixed $\omega\in\Omega$, $\tau(\omega)\wedge T \leq T$, and $\widebar{A}$ remains of finite variation on $[t,\tau(\omega)\wedge T]$. Thus, by replacing $T$ by $\tau\wedge T$ throughout the proof, the same inequalities and integrability arguments hold. Hence, $\text{CF}^{\text{con}}_{(t,\tau\wedge T]}\in \mathbb{L}^2(\mathcal{F}_{\tau\wedge T})$.

    Under our assumptions \ref{Ass:ir}, \ref{Ass:discount}, and \ref{Ass:X_cond} $(r-r^c)C\in\mathbb{S}^2(\mathbb{F})$ and $D_{t,\cdot}(r)$ has finite moments, and it follows that $\text{CF}^\text{col}_{(t,\tau\wedge T]}\in\mathbb{L}^2(\mathcal{F}_{\tau\wedge T})$.

    Under Assumption \ref{Ass:risk_measure}, it holds that $\varrho\colon\mathbb{L}^2(\F)\to\mathbb{S}^2(\mathbb{F})$ and hence by repeating the arguments from above, it holds that $\text{CF}_{(t,\tau\wedge T]}^{\text{IM}}\in\mathbb{L}^2(\F_{\tau\wedge T})$ and $\text{CF}^\text{col}_{(t,\tau\wedge T]}\in\mathbb{L}^2(\mathcal{F}_{\tau\wedge T})$.
    \end{proof}
\end{lemma}

\begin{proposition}\label{prp:V^tau}
For $t\in[0,\tau\wedge T]$, the following implicit equation is well defined
    \begin{equation*}  V_t = \mathbb{E}\big[-\emph{CF}^\emph{con}_{t,\tau\wedge T} + \emph{CF}^\emph{def}_{t,\tau\wedge T} + \emph{CF}^\emph{col}_{t,\tau\wedge T} + \emph{CF}^{\emph{IM}}_{t,\tau\wedge T} + \emph{CF}^\emph{fun}_{t,\tau\wedge T}(V_t) \,\big|\, \mathcal{F}_t\big], \end{equation*}
Moreover, $V=(V_t)_{t\in(0,\tau\wedge T]}\in\mathbb{S}^2(\mathbb{F})$ and $\text{CF}^\text{fun}_{t,\tau\wedge T}\defeq \text{CF}^\text{fun}_{t,\tau\wedge T}(V)\in\mathbb{L}^2(\F_{\tau\wedge T})$.
\begin{proof}
The result follows from standard arguments using Banach’s fixed-point theorem. Specifically, the mapping
\begin{equation*}
\Phi(\cdot) \defeq \mathbb{E}\big[-\text{CF}^\text{con}_{t,\tau\wedge T} + \text{CF}^\text{def}_{t,\tau\wedge T} + \text{CF}^\text{col}_{t,\tau\wedge T} + \text{CF}^{\text{IM}}_{t,\tau\wedge T} + \text{CF}^\text{fun}_{t,\tau\wedge T}(\cdot) \,\big|\, \mathcal{F}_t\big]
\end{equation*}
is a contraction on $\mathbb{S}^2(\mathbb{F})$, given the Lipschitz continuity of $\text{CF}^\text{fun}$ in $v$ and the square-integrability of all other terms. As previously shown for the other cashflow terms, $\text{CF}^\text{fun}_{t,\tau\wedge T}(V)\in\mathbb{L}^2(\F_{\tau\wedge T})$ whenever $V\in\mathbb{S}^2(\mathbb{F})$. 
\end{proof}
\end{proposition}
The adjusted portfolio value can then, for $t\in[0,\tau\wedge T]$, be defined as \begin{equation} \label{V} V_t \coloneq \mathbb{E}\big[-\text{CF}^\text{con}_{t,\tau\wedge T} + \text{CF}^\text{def}_{t,\tau\wedge T} + \text{CF}^\text{col}_{t,\tau\wedge T} + \text{CF}^{\text{IM}}_{t,\tau\wedge T} + \text{CF}^\text{fun}_{t,\tau\wedge T} \,\big|\, \mathcal{F}_t\big]. \end{equation}

\subsection{Valuation adjustments}
Without further manipulation, we can identify ColVA, MVA and FVA. These adjustments are defined for $t\in[0,\tau\wedge T]$ as
\begin{equation*}
    \text{ColVA}_t\defeq \E\big[\text{CF}^\text{col}_{(t,\tau\wedge T]}|\F_t\big],\quad     \text{MVA}_t\defeq \E\big[\text{CF}^\text{IM}_{(t,\tau\wedge T]}|\F_t\big], \quad   \text{FVA}_t\defeq \E\big[\text{CF}^\text{fun}_{(t,\tau\wedge T]}|\F_t\big],
\end{equation*}
representing the adjustments for the costs/benefits associated with variation margin, initial margin, and funding, respectively. For the contractual and default cashflows, we use
\begin{equation*}
    \E\bigg[-\int_{(t,\tau\wedge T]}D_{t,s}(r)\d \widebar{A}_s\;\Big|\;\F_t\bigg]+\E\big[D_{t,\tau}(r)\I_{\{\tau<T\}}Q_\tau\;\big|\;\F_t\big]=\widehat{V}_t.
\end{equation*} to derive
\begin{align*}
    \E\big[-\text{CF}^\text{con}_{(t,\tau\wedge T]}+\text{CF}^\text{def}_{(t,\tau\wedge T]}\big|\F_t\big]=&\widehat{V}_t-\E\Big[\I_{\{\tau\leq T\}}\I_{\{\tau=\tau^\mathcal{C}\}}D_{t,\tau}(r)\text{LGD}^\mathcal{C}(Q_\tau-C_\tau-\text{IM}_\tau^\text{FC})^+\;\big|\;\F_t\Big]\\
    &+\E\Big[\I_{\{\tau\leq T\}}\I_{\{\tau=\tau^\mathcal{C}\}}D_{t,\tau}(r)\text{LGD}^\mathcal{B}(Q_\tau-C_\tau-\text{IM}_\tau^\text{TC})^-\;\big|\;\F_t\Big].
\end{align*}
Above, we express the expected contractual and default cashflows as the clean portfolio value adjusted by two correction terms: one for the credit risk of the counterparty and one for the credit risk of the bank. Depending on the sign convention, these terms correspond to the CVA and DVA, where the first term is added and the second term is subtracted. This sign convention reflects the interpretation of these adjustments but does not alter their financial impact.
\begin{definition}[Valuation Adjustments]\label{def:valuation_adjustments}
For $X\in\{\emph{C},\emph{D},\emph{Col},\emph{M},\emph{F}\}$, we define the valuation adjustments $\emph{XVA}=(\emph{XVA}_t)_{t\in[0,\tau\wedge T]}$, which for  $t\in[0,\tau\wedge T]$, are defined as:
\begin{equation}\label{eq:XVA_all_definitions}
\boxed{
\begin{aligned}
\emph{CVA}_t \coloneqq& \, \mathbb{E}\Big[\mathbbm{1}_{\{\tau \leq T\}}\mathbbm{1}_{\{\tau = \tau^\mathcal{C}\}}D_{t,\tau}(r)\,\emph{LGD}^{\mathcal{C}}\bigl(Q_\tau - C_\tau - \emph{IM}_\tau^{\emph{FC}}\bigr)^+ \mid \mathcal{F}_t\Big], \\[6pt]
\emph{DVA}_t \coloneqq& \,\mathbb{E}\Big[\mathbbm{1}_{\{\tau \leq T\}}\mathbbm{1}_{\{\tau = \tau^\mathcal{B}\}}D_{t,\tau}(r)\,\emph{LGD}^{\mathcal{B}}\bigl(Q_\tau - C_\tau - \emph{IM}_\tau^{\emph{TC}}\bigr)^- \mid \mathcal{F}_t\Big], \\[6pt]
\emph{ColVA}_t \coloneqq &\,-\mathbb{E}\biggl[\int_t^{\tau\wedge T} D_{t,s}(r)\,(r_s - r_s^c)C_s\,\mathrm{d}s \;\bigg|\;\mathcal{F}_t\biggr], \\[6pt]
\emph{MVA}_t \coloneqq &\,-\mathbb{E}\biggl[\int_t^{\tau\wedge T} D_{s,t}(r)\bigl((r_s-r_s^{\emph{IM},l})\emph{IM}_s^{\text{TC}} -r_s^{\emph{IM},b}\emph{IM}_s^{\emph{FC}}\bigr)\mathrm{d}s \;\bigg|\;\mathcal{F}_t\biggr], \\[6pt]
\emph{FVA}_t \coloneqq &\,-\mathbb{E}\biggl[\int_t^{\tau\wedge T}D_{t,s}(r)\Big[(r-r_s^{f,b})\big(Q_s-C_s-\emph{IM}_s^\emph{TC}\big)^+\\
&-(r-r_s^{f,l})\big(Q_s-C_s-\emph{IM}_s^\emph{TC}\big)^-\Big]\d s \;\Big|\;\mathcal{F}_t\biggr].
\end{aligned}
}
\end{equation}
\end{definition}
The adjusted portfolio value can then be expressed in terms of the clean portfolio value and the valuation adjustments as \begin{align}\label{V_XVA} V_{t}^\tau = \widehat{V}_t - \text{CVA}_t + \text{DVA}_t - \text{FVA}_t - \text{ColVA}_t - \text{MVA}_t. \end{align}
Hence, at $t=0$, the bank sells the portfolio to the counterparty for the amount\begin{equation*}
    -V_0=-\widehat{V}_0+\text{CVA}_0 - \text{DVA}_0 + \text{FVA}_0 + \text{ColVA}_0 + \text{MVA}_0.
\end{equation*}
Alternatively, if we adopt the option valuation approach, in which the bank purchases the derivative, we obtain
\begin{equation*}
    V_0=\widehat{V}_0+\text{CVA}_0 - \text{DVA}_0 + \text{FVA}_0 + \text{ColVA}_0 + \text{MVA}_0.
\end{equation*}
Finally, denote the total valuation adjustment by $\mathrm{TVA}=(\mathrm{TVA}_t)_{t\in[0,\tau\wedge T]}$, which for $t\in[0,T]$ is defined by \begin{equation}\label{XVA}
    \mathrm{TVA}_t\coloneqq \widehat{V}_t-V_t=\text{CVA}_t-\text{DVA}_t+\text{FVA}_t+\text{ColVA}_t+\text{MVA}_t.
\end{equation}
\section{Valuation BSDEs}\label{sec:Val_BSDEs}
It is clear that the conditional expectation in \eqref{V} can be represented as a BSDE. In this section, we state this as a formal proposition. However, relying on a single BSDE formulation for $V$ has two main drawbacks: 
\begin{itemize}
    \item The resulting BSDE is highly complex and exhibits several non-standard features: different terminal conditions depending on the nature and timing of defaults, McKean--Vlasov anticipative terms arising from the chosen risk measure, multiple nonlinearities due to distinct lending and borrowing rates, and a high-dimensional state space. By contrast, if we instead formulate one BSDE per individual valuation adjustment, we can isolate specific complexities, ensuring that not all these challenging features appear simultaneously in a single equation. This separation can simplify both the conceptual understanding and the numerical solution procedure. \item If we only compute $V$ from one unified BSDE, there is no straightforward way to extract the individual valuation adjustments. Separate BSDE formulations for each valuation adjustment directly provide their respective contributions, enabling a clearer decomposition.
\end{itemize}
With these considerations in mind, we derive not only the full BSDE for the adjusted portfolio value but also distinct BSDEs for each individual valuation adjustment.
\subsection{Expressing the portfolio values as BSDEs}
The following proposition provides the final tool for formulating the clean and adjusted portfolio values as BSDEs.
\begin{proposition}\label{prp:BSDE}
For $t \in [0,T]$, let $\I_{\{\tau \leq T\}} \chi_\tau \in \mathbb{L}^2(\mathcal{F}_{\tau \wedge T})$, $\Lambda = (\Lambda_t)_{t \in [0,T]}$ be an $\mathbb{F}$-adapted càdlàg process satisfying $\int_{(t, \tau \wedge T]} D_{t,s}(r) \, \mathrm{d} \Lambda_s \in \mathbb{L}^2(\mathcal{F}_{\tau \wedge T})$, and $f^Y = (f_t^Y)_{t \in [0, \tau \wedge T]} \in \mathbb{H}^2(\mathbb{F})$.

Let $Y = (Y_t)_{t \in [0,T]} \in \mathbb{S}^2(\mathbb{F})$ be implicitly defined by the conditional expectation (which is well-defined by Proposition~\ref{prp:V^tau})
\begin{equation}\label{eq:BSDE_exp}
    Y_t = \mathbb{E}\biggl[\I_{\{\tau \leq T\}} \chi_\tau D_{t,s}(r) - \int_{(t, \tau \wedge T]} D_{t,s}(r) \, \mathrm{d} \Lambda_s 
    + \int_t^{\tau \wedge T} D_{t,s}(r) f_s^Y \, \mathrm{d}s \, \bigg| \, \mathcal{F}_t\biggr].
\end{equation}
Then, there exists a unique $Z = (Z_t)_{t \in [0, \tau \wedge T]} \in \mathbb{H}^2(\mathbb{F})$ such that $(Y, Z)$ solves the following BSDE
\begin{equation*}
    Y_t = \I_{\{\tau \leq T\}} \chi_\tau - \int_{(t, \tau \wedge T]} \mathrm{d} \Lambda_s + \int_t^{\tau \wedge T} \big(f_s^Y - r_s Y_s\big) \, \mathrm{d}s - \int_t^{\tau \wedge T} Z_s\cdot \mathrm{d}W_s.
\end{equation*}
\begin{proof}
The conditional expectation \eqref{eq:BSDE_exp} satisfies the BSDE structure due to the linear form of the discounting term $D_{t,s}(r)$, which ensures that the integrand remains square-integrable. The uniqueness of the solution $(Y, Z)$ follows directly from standard results in BSDE theory, since $f^Y-rY\in \mathbb{H}^2(\mathbb{F})$ and the stochastic integral with respect to $W$ is well-defined.
\end{proof}
\end{proposition}
To express the adjusted portfolio value as a BSDE we introduce some notation. Let $f^V\colon[0,T]\times\R\times\R\times\R\to\R$ such that for $(t,v,c,im_1,im_2)\in\colon[0,T]\times\R\times\R\times\R$, $f^V$ is defined as \begin{align*}
    f^V(t,v,c,im_1,im_2)=&(r_t-r_t^c)c+(r_t^{\text{IM},l}-r_t)im_1 - r_t^{\text{IM},b}im_2 + (r_t-r_t^{f,b})\bigl((v-c-im_1)^+ \\&- (r_t-r_t^{f,l})(v-c-im_1)^-.
\end{align*}

Using Proposition \ref{prp:BSDE} with specific choices of
$\theta=\chi$, $\widebar{A}=\Lambda$ and $f^V=f^Y$ we derive the BSDE for the adjusted portfolio value
\begin{equation}
    V_t=\theta_\tau\I_{\{\tau\leq T\}} -\int_{(t,\tau\wedge T]}\d \widebar{A}_s+\int_t^{\tau\wedge T}\big(f^V(s,V_s,C_s,\text{IM}_s^\text{TC},\text{IM}_s^\text{FC})-r_sV_s\big)\d s - \int_t^{\tau\wedge T}Z_s\cdot\d W_s.\label{dirty_V_BSDE}
\end{equation}
Similarly, for the clean portfolio value, defined in \eqref{clean_V}, we obtain the following BSDE 
\begin{equation}
        \widehat{V}_t= -\int_{(t,T]}\d A_s - \int_t^T r_s\widehat{V}_s\d s - \int_t^T\widehat{Z}_s\cdot\d\widehat{W}_s.\label{clean_V_BSDE}
\end{equation} 
Note that since $\widehat{V}$ is $\widehat{\mathbb{F}}$-adapted, the resulting BSDE is an $\widehat{\mathbb{F}}$-BSDE, meaning it is adapted to the filtration $\widehat{\mathbb{F}}$ and driven by the $\widehat{\mathbb{F}}$-Brownian motion $\widehat{W}$, rather than $W$.

\begin{example}[Continuation - European Call Option]
As in Example \ref{Ex:Call_Option}, assume a portfolio of a single European call option, written on $S=(S_t)_{t\in[0,T]}$, with maturity $T$ and strike price $K$. This setting allows for an equivalent continuous BSDE (since $A$ only consists of a single discontinuous jump at $t=T$). The clean value BSDE is reduced to
\begin{equation*}
        \widehat{V}_t= -(S_T-K)^+ - \int_t^T r_s\widehat{V}_s\emph{d} s - \int_t^T\widehat{Z}_s\cdot\emph{d}\widehat{W}_s.\label{eq:eaxmple_clean_BSDE}
\end{equation*}
Similarly, for the clean portfolio value, defined in \eqref{clean_V}, we obtain \begin{equation}
    V_t=\theta_\tau\I_{\{\tau\leq T\}} -(S_T-K)^+\I_{\{\tau> T\}} +\int_t^{\tau\wedge T}\big(f^V(s,V_s,C_s,\emph{IM}_s^\emph{TC},\emph{IM}_s^\emph{FC}) - r_sV_s\big)\emph{d} s - \int_t^{\tau\wedge T}Z_s\cdot\emph{d} W_s.\label{Eq:Example_dirty_V_BSDE}
\end{equation}
\end{example}

\subsection{Expressing the valuation adjustments as a system of coupled BSDEs}
As mentioned in the previous section, the valuation adjustments defined in \eqref{def:valuation_adjustments} can also be expressed as BSDEs using Proposition \ref{prp:BSDE}.

Let $f^\text{ColVa}\colon[0,T]\times\R\to\R$, $f^\text{MVA}\colon[0,T]\times\R\times\R\to\R$ and $f^\text{FVA}\colon[0,T]\times\R\times\R\times\R\times\R\to\R$ which for $(t,\tilde{v},tva,c,im_1,im_2)\in\colon[0,T]\times\R\times\R,\times\R\times\R\times\R$, are defined as
\begin{align*}
    f^\text{ColVA}(t,c)\defeq&(r_t-r_t^c)c,\quad
    f^\text{MVA}(t,im_1,im_2)\defeq(r_t-r_t^{\text{IM},l})im_1 - r^{\text{IM},b}_tim_2,\\
    f^\text{FVA}(t,\widehat{v},tva,c,im_1)\defeq&(r-r^{f,b}_t)\big(\widehat{v}-tva-c-im_1\big)^+-(r-r^{f,l}_t)\big(\widehat{v}-tva-c-im_1\big)^-.
\end{align*}
Using Proposition \ref{prp:BSDE} the valuation adjustments can be formulated as a system of coupled BSDEs. The dashed lines below indicate layers of depth in the system, where each layer corresponds to a set of equations. At each layer, the solution depends on the solutions of all BSDEs from the preceding layers. This hierarchical structure ensures that the computation of each valuation adjustment incorporates all necessary information from the preceding layers. These layers are particularly critical in the context of numerical computations, as they define the sequential order in which the computations must be performed to ensure consistency. For $t\in[0,\tau\wedge T]$, we have
\begin{center}
\label{box:XVA_BSDEs}
\fbox{%
\begin{minipage}{0.95\textwidth} 
\begin{align*}
\widehat{V}_t^j &= -\int_{(t,T]} \d A_s^j 
                   - \int_t^T r_s \widehat{V}_s^j \d s 
                   - \int_t^T \widehat{Z}_s^{\mathcal{I}_j}\cdot \d \widehat{W}_s^{\mathcal{I}_j},\quad j\in\{1,2,\ldots, P\},\quad\widehat{V}=\sum_{j=1}^P\widehat{V}^j, \\[6pt]
\noalign{\hdashrule[1.5ex]{\linewidth}{0.7pt}{4pt}}
-\ColVA_t &= \int_t^{\tau\wedge T} \bigl(f^\ColVA(s,C_s) - r_s \ColVA_s\bigr) \d s
            - \int_t^{\tau\wedge T} Z_s^\ColVA\cdot \d W_s, \\
\IM_t^\FC &= \varrho \bigl[(\widehat{V}_{t+\text{MPR}_t} - \widehat{V}_t)^+ \mid \widehat{\mathcal{F}}_t\bigr], \quad
\IM_t^\TC = -\varrho \bigl[(\widehat{V}_{t+\text{MPR}_t} - \widehat{V}_t)^- \mid \widehat{\mathcal{F}}_t\bigr], \\[6pt]
\noalign{\hdashrule[1.5ex]{\linewidth}{0.7pt}{4pt}}
\CVA_t &= \mathbbm{1}_{\{\tau \leq T\}} \mathbbm{1}_{\{\tau = \tau^\mathcal{C}\}} 
          \LGD^{\mathcal{C}} (\widehat{V}_\tau - C_\tau - \IM_\tau^{\FC})^+ 
          - \int_t^{\tau\wedge T} r_s \CVA_s \d s 
          - \int_t^{\tau\wedge T} Z_s^\CVA\cdot \d W_s, \\
\DVA_t &= \mathbbm{1}_{\{\tau \leq T\}} \mathbbm{1}_{\{\tau = \tau^\mathcal{B}\}} 
          \LGD^{\mathcal{B}} (\widehat{V}_\tau - C_\tau - \IM_\tau^{\TC})^- 
          - \int_t^{\tau\wedge T} r_s \DVA_s \d s 
          - \int_t^{\tau\wedge T} Z_s^\DVA\cdot \d W_s, \\
-\MVA_t &= \int_t^{\tau\wedge T} \bigl(f^\MVA(s,\IM^{\FC}_s,\IM^\TC_s) - r_s \MVA_s\bigr) \d s - \int_t^{\tau\wedge T} Z_s^\MVA\cdot \d W_s, 
\\[6pt]
\noalign{\hdashrule[1.5ex]{\linewidth}{0.7pt}{4pt}}
\textbf{XVA}_t &= (\ColVA_t, \CVA_t, \DVA_t, \MVA_t, \FVA_t), \\
-\FVA_t &= \int_t^{\tau\wedge T} \bigl(f^\FVA(s,\widehat{V}_s,C_s,\textbf{XVA}_t,\IM^\TC_s) 
         - r \FVA_s\bigr) \d s 
         - \int_t^{\tau\wedge T} Z_s^\FVA \cdot\d W_s.
\end{align*}
\end{minipage}%
}
\end{center}
It is important to note that solving the system of BSDEs above, along with defining $V^\XVA\defeq\widehat{V}-\ColVA-\CVA+\DVA-\MVA-\FVA$ is equivalent to solving the full BSDE in \eqref{dirty_V_BSDE}, \textit{i.e.,} $V^\XVA=V$. However, as highlighted above, extracting the individual valuation adjustments directly from \eqref{dirty_V_BSDE} is generally not feasible.

\subsection{Deep BSDE solver under multiple measures
}\label{sec:var_problem}
This section introduces two reformulations of the BSDEs from the previous section. The first reformulation involves applying a measure change to adjust the forward SDE that models the underlying assets. By doing so, we aim to increase the probability of default for the bank or the counterparty. This is important because, in practice, we can only draw a limited number of samples in the BSDE approximation, and very low default probabilities can make it difficult to capture default events accurately.

The second reformulation rewrites the BSDE as an equivalent variational problem, a standard technique when using the BSDE solver. Specifically, the backward SDE is converted into a forward SDE, and the objective becomes finding both the initial condition of the BSDE and the control process so that the terminal condition is satisfied in a mean-squared sense.

\subsubsection{Change of measure to increase default probabilities}
\label{sec:ChangeOfMeasure}
The BSDEs describing $\mathbf{XVA} = (\mathbf{XVA}_t)_{t \in [0, \,\tau \wedge T]}$ depend heavily on the default time $\tau = \tau^\mathcal{C} \wedge \tau^\mathcal{B}$. Since $\tau^\mathcal{C}$ and $\tau^\mathcal{B}$ represent the defaults of the counterparty and the bank, respectively, these events tend to be rare, especially for highly creditworthy entities. Our main goal is to approximate these BSDEs using neural network–based algorithms; however, a very low probability of default for either party can be problematic from a computational perspective, due low accuracy and/or high variance in our estimates. In this section, we outline how a change of measure technique can be utilized to increase the default probabilities. 

Consider the decoupled FBSDE, where the forward SDE is given by \eqref{eq:full_X} and the BSDE is of the form given in Proposition \ref{prp:BSDE}
\begin{equation}\label{eq:FBSDE_repeat}
\begin{dcases}
  X_t \;=\; x_0 \;+\; \displaystyle\int_0^t b\bigl(s,X_s\bigr)\,\mathrm{d}s 
           \;+\; \int_0^t \sigma\bigl(s,X_s\bigr)\,\odot\,\mathrm{d}W_s, 
  \\[6pt]
  Y_t
  \;=\; \mathbbm{1}_{\{\tau \le T\}}\chi_\tau
  \;-\; \displaystyle\int_{(t,\;\tau\wedge T]} \mathrm{d}\Lambda_s
  \;+\; \int_t^{\tau\wedge T}
    \bigl(f_s^Y \;-\; r_s\,Y_s\bigr)\,\mathrm{d}s
  \;-\; \int_t^{\tau\wedge T} Z_s\cdot\mathrm{d}W_s,
\end{dcases}
\end{equation}
under our baseline measure $\mathbb{P}$. 

To increase default probabilities, we introduce a ``tilt'' in the drift coefficient of $X$. We assume $q\colon [0,T]\times \mathbb{R}^{d+2} \;\to\;\mathbb{R}^{d+2}$
satisfies the same Lipschitz and linear‐growth bounds as in Assumption~\ref{Ass:X_cond}, and introduce the \emph{modified} or \emph{tilted} forward process $X^q=(X_t^q)_{t\in[0,T]}$ solving
\begin{equation}\label{eq:CM_FBSDE_forward}
  X_t^q 
  \;=\; 
  x_0 
  \;+\;
  \int_0^t\!\bigl[b\bigl(s,X_s^q\bigr)\;-\;q\bigl(s,X_s^q\bigr)\bigr]\,\mathrm{d}s
  \;+\;
  \int_0^t \sigma\bigl(s,X_s^q\bigr)\,\odot\,\mathrm{d}W_s.
\end{equation}
Here $\tau^q$ is the corresponding default time induced by $X^q$. 

By construction, $X^q$ under the \emph{original} measure $\mathbb{P}$ has drift $b - q$. 
However, define the following \emph{Girsanov kernel} and exponential martingale:
\[
\theta_t \;:=\; q\bigl(t,X_t\bigr)\,\oslash\,\sigma\bigl(t,X_t\bigr),
\quad
\Gamma_t^q 
\;:=\;
\exp\Bigl(
  -\!\int_0^t \theta_s\cdot \mathrm{d}W_s
  \;-\;
  \tfrac12 \int_0^t \bigl|\theta_s\bigr|^2 \,\mathrm{d}s
\Bigr).
\]
Under mild conditions (e.g.\ Novikov’s condition), $\Gamma_t^q$ is a true martingale, 
and we may define a new probability measure $\mathbb{P}^q$ on $\mathcal{F}_T$ by $\frac{\mathrm{d}\mathbb{P}^q}{\mathrm{d}\mathbb{P}}\Bigm|_{\mathcal{F}_T}=\Gamma_T^q$. From Girsanov’s theorem, the process
$W_\cdot^q \defeq W_\cdot+\int_0^\cdot \theta_s\mathrm{d}s$, 
is then a Brownian motion under $\mathbb{P}^q$. 
Rewriting the forward SDE in \eqref{eq:FBSDE_repeat} in terms of $W^q$ shows that $X^q$ under $\mathbb{P}$ 
has the \emph{same} distribution as $X$ under $\mathbb{P}^q$. 
Thus, simulating $X^q$ under $\mathbb{P}$ is equivalent to simulating the \emph{original} SDE $X$ 
but under a different measure $\mathbb{P}^q$. 
Consequently, the default event becomes more likely from the sampling perspective 
while still preserving the correct pricing PDE or BSDE structure.

Let $\bigl(Y^q,\,Z^q\bigr)$ solve the ``tilted'' BSDE associated with the default time $\tau^q$. Formally, under $\mathbb{P}^q$ we might write
\begin{equation*}
Y_t^q
\;=\;
\mathbbm{1}_{\{\tau^q \le T\}}\chi_{\tau^q}^q
\;-\;
\int_{(t,\;\tau^q\wedge T]} \mathrm{d}\Lambda_s^q
\;+\;
\int_t^{\tau^q\wedge T}
  \bigl(f_s^{Y,q} - r_s\,Y_s^q\bigr)\,\mathrm{d}s
\;-\;
\int_t^{\tau^q\wedge T}Z_s^q \cdot\mathrm{d}W^q_s.
\end{equation*}
Translating this back to $\mathbb{P}$ introduces an extra generator term. Indeed,
\begin{equation*}
\int Z_s^q\cdot\mathrm{d}W^q_s 
\;=\;
\int Z_s^q\cdot\mathrm{d}W_s
\;+\;
\int Z_s^q\cdot\theta_s\,\mathrm{d}s\;=\;\int Z_s^q\cdot\mathrm{d}W_s
\;+\;
\int \langle q(s,X_s^q)\oslash \sigma(s,X_s^q),Z_s^q\rangle\,\mathrm{d}s
\end{equation*}
Hence, under the original measure $\mathbb{P}$, $\bigl(Y^q,Z^q\bigr)$ satisfies
\begin{equation}\label{eq:CM_FBSDE}
    \begin{dcases}
    X_t^q=x_0 + \int_0^t \big(b(s,X_s^q)-q(s,X_s^q)\big)\d s + \int_0^t\sigma(s,X_s^q)\odot\d W_s,\\
    Y_t^q = \I_{\{\tau^q \leq T\}} \chi_{\tau^q}^q - \int_{(t, \tau^q \wedge T]} \mathrm{d} \Lambda_s^q + \int_t^{\tau^q \wedge T} \big(f_s^{Y,q} - r_s Y_s^q-\langle q(s,X_s^q)\oslash \sigma(s,X_s^q),Z_s^q\rangle\big) \, \mathrm{d}s \\
    \quad\quad\quad- \int_t^{\tau^q \wedge T} Z_s^q \cdot \mathrm{d}W_s.
    \end{dcases}
\end{equation}
This system precisely corresponds to an \emph{adjusted drift} in the forward SDE and an \emph{extra term} $\langle q\oslash\sigma,\;Z^q\rangle$ in the BSDE generator.
\vspace{0.25cm}\newline
\textbf{Practical advantage}:\newline
By simulating the \emph{tilted} SDE $X^q$, one observes default scenarios more frequently. 
This improves training stability and reduces the variance of neural network based methods, 
since they can learn from many more samples in the (otherwise rare) default region. 
After training, the result is mapped back to the original measure, 
so the final solution remains consistent with the true $\mathbf{XVA}$ under $\mathbb{P}$. Figure~\ref{fig:CM} presents a conceptual one-dimensional illustration of the formulations in \eqref{eq:FBSDE_repeat} and \eqref{eq:CM_FBSDE}. In this figure, the blue curves represent the processes $X$ and $Y$ under the original drift, while the red curves represent the processes $X^q$ and $Y^q$ under the adjusted drift that leads to an increased default probability.
\begin{figure}[htp]
\centering
\begin{tabular}{c}
          \includegraphics[width=160mm]{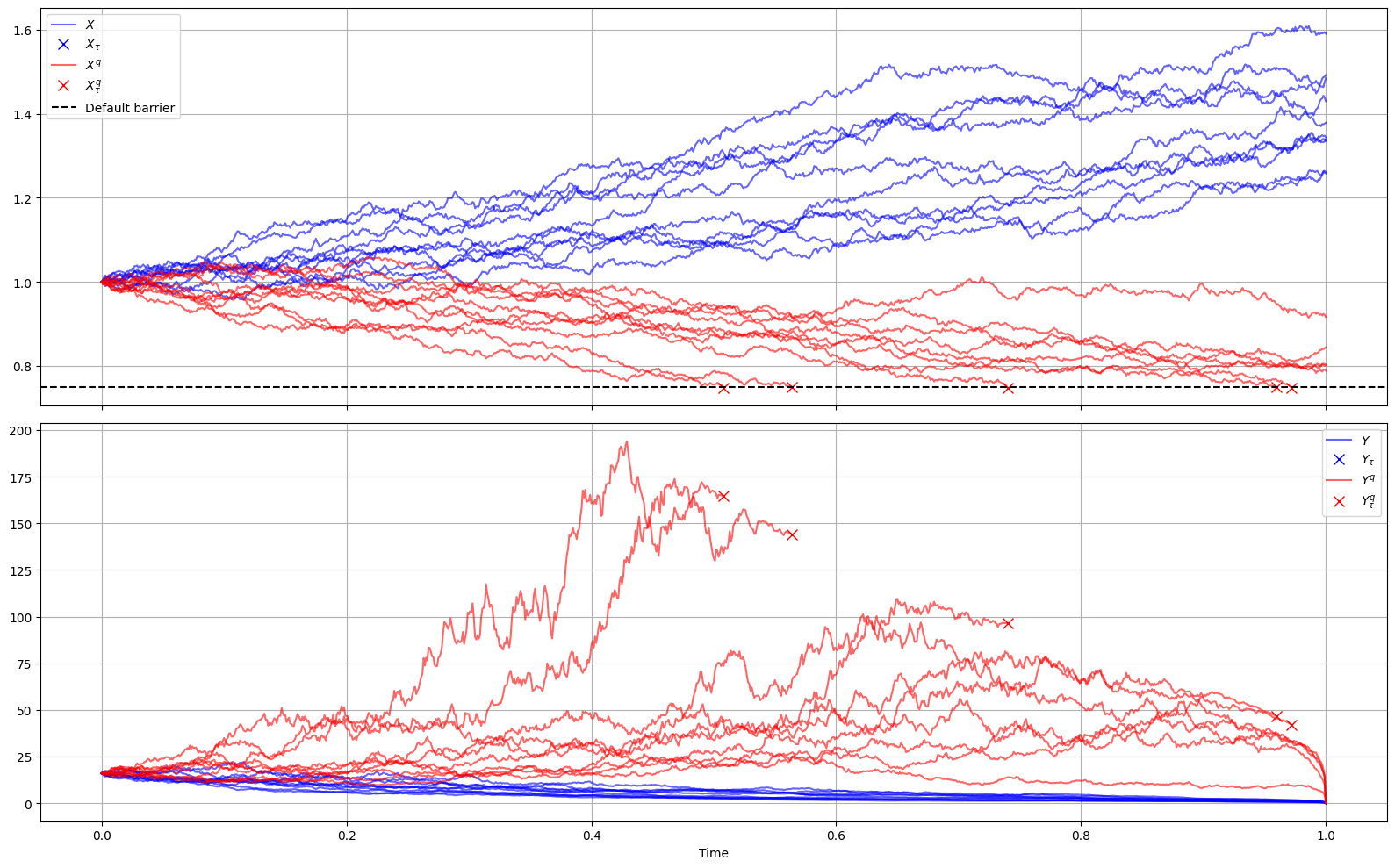}
\end{tabular}
\caption{Conceptual illustration of the formulations in \eqref{eq:FBSDE_repeat} and \eqref{eq:CM_FBSDE}. The blue curves represent the processes $X$ and $Y$ under the original drift, while the red curves represent the processes $X^q$ and $Y^q$ under the adjusted drift, which leads to an increased default probability.}\label{fig:CM}
\end{figure}

\begin{remark}[PDE approach and Itô's lemma]\label{rem:Markovian}
Assume that the BSDE in \eqref{eq:FBSDE_repeat} is Markovian in the state variable so that it can be associated with a PDE. Then the solution to \eqref{eq:FBSDE_repeat} can be represented by the viscosity solution $v$ of that PDE, and for a pre‐default state $(t,X_t)$ we have
\begin{equation*}
    Y_t^\tau \;=\; v(t,X_t), 
    \quad 
    Z_t^\tau \;=\; \sigma\bigl(t,X_t\bigr)\odot \mathrm{D}_x v\bigl(t,X_t\bigr).
\end{equation*}
Next, evaluating $v$ along the tilted SDE \eqref{eq:CM_FBSDE_forward} and applying Itô’s lemma recovers \eqref{eq:CM_FBSDE}. From this it becomes clear that the PDEs associated with \eqref{eq:FBSDE_repeat} and \eqref{eq:CM_FBSDE} coincide; the only difference is that in the former case one follows trajectories generated by $X$, and in the latter by $X^q$. Hence the PDE itself remains the same, but is traced out along different state‐space paths (blue and red lines in the upper panel in Figure~\ref{fig:CM}) under the two SDEs.
\end{remark}

\begin{remark}[Similarities and differences with importance sampling]
Both rewriting the BSDE with a new drift and importance sampling employ a change of measure to focus the forward SDE paths on critical regions. However, they differ fundamentally in their approach and objectives.  

Importance sampling primarily reduces variance in Monte Carlo simulations while leaving the original BSDE unchanged. It relies on simulating under a proposed measure and reweighting via a Radon--Nikod\'ym derivative $\Gamma$ to maintain correctness under the original measure. Although this yields an unbiased estimator, it can suffer from instability if $\Gamma$ becomes large.

Rewriting the BSDE with a new drift, on the other hand, modifies the forward dynamics so that the same Markovian solution $v$ (solution of the associated PDE) remains valid without requiring additional weighting. By integrating the “weighting” directly into the driver or differential operator, one avoids large likelihood-ratio effects.

In importance sampling, the red curves in the upper panel of Figure~\ref{fig:CM} correspond to a change of measure that increases the default probability, while the BSDE in the lower panel is solved under the original measure after correcting via the Radon--Nikod\'ym derivative. In contrast, our approach employs the same red curves for both sampling the forward process and solving the BSDE, as the change of measure is directly incorporated into the drift of the forward SDE and the driver of the backward SDE. 
\end{remark}

\subsubsection{General BSDE solver under multiple measures}
We now summarize the previous discussion by providing the general formulation of the deep BSDE solver under multiple measures. We start buy recalling that the FBSDE \eqref{eq:CM_FBSDE} can be equivalently reformulated as a variational problem: this reformulation reinterprets the search for a solution $(y_0, Z^q)$ in terms of minimizing a mean-square error criterion.

 We define the jump of the contractual payoff process $\Lambda^q$ by $
\Delta \Lambda^q_{\tau^q\wedge T} 
\;=\; \Lambda^q_{\tau^q \wedge T} \;-\; \Lambda^q_{(\tau^q \wedge T)^-}$.
Then, the FBSDE \eqref{eq:CM_FBSDE} can be recast as the optimization problem
\begin{equation}\label{eq:cont_Var_problem}
    \begin{dcases}
    \inf_{y_0,(Z_t^q)_{t\in[0,T]}}\E\Big[\big|\I_{\{\tau^q\leq T\}}\chi_{\tau^q}^q-\Delta\Lambda^q_{\tau^q\wedge T}-Y_{\tau^q\wedge T}^q\big|^2\Big],\quad\text{where,}\\
    X_t^q=x_0 + \int_0^t \big(b(s,X_s^q)-q(s,X_s^q)\big)\d s + \int_0^t\sigma(s,X_s^q)\odot\d W_s,\\
    Y_t^q = y_0 + \int_{(t, \tau^q \wedge T]} \mathrm{d} \Lambda_s^q - \int_t^{\tau^q \wedge T} \big(f_s^{Y,q} - r_s Y_s^q-\langle q(s,X_s^q)\oslash \sigma(s,X_s^q),Z_s^q\rangle\big) \, \mathrm{d}s \\
    \quad\quad\quad+ \int_t^{\tau^q \wedge T} Z_s^q \cdot \mathrm{d}W_s.
    \end{dcases}
\end{equation}

We search for an $\F_0$-measurable initial condition $y_0$ and an $\F_t$-adapted, square-integrable process $Z^q$. It is clear that the unique solution to \eqref{eq:CM_FBSDE} attains an objective value of zero, which implies it is also the unique minimizer of \eqref{eq:cont_Var_problem}. 

Since \eqref{eq:CM_FBSDE} and \eqref{eq:cont_Var_problem} conincide for any $q$ satisfying the drift conditions in Assumption \ref{Ass:X_cond} we can formulate a variational problem, which corresponds to multiple FBSDE.
\begin{equation}\label{eq:cont_Var_problem_multi}
    \begin{dcases}
    \inf_{y_0,(Z_t^{q_1},\ldots,Z_t^{q_K})_{t\in[0,T]}}\sum_{i=1}^K\E\Big[\big|\I_{\{\tau^{q_i}\leq T\}}\chi_{\tau^{q_i}}^{q_i}-\Delta\Lambda^{q_i}_{\tau^{q_i}\wedge T}-Y_{\tau^{q_i}\wedge T}^{q_i}\big|^2\Big],\quad\text{where for } i\in\{1,2,\ldots,K\},\\
    X_t^{q_i}=x_0 + \int_0^t \big(b(s,X_s^{q_i})-q_i(s,X_s^{q_i})\big)\d s + \int_0^t\sigma(s,X_s^{q_i})\odot\d W_s,\\
    Y_t^{q_i} = y_0 + \int_{(t, \tau^{q_i} \wedge T]} \mathrm{d} \Lambda_s^{q_i} - \int_t^{\tau^{q_i} \wedge T} \big(f_s^{Y,{q_i}} - r_s Y_s^{q_i}-\langle q_i(s,X_s^{q_i})\oslash \sigma(s,X_s^{q_i}),Z_s^{q_i}\rangle\big) \, \mathrm{d}s \\
    \quad\quad\quad+ \int_t^{\tau^{q_i} \wedge T} Z_s^{q_i} \cdot \mathrm{d}W_s.
    \end{dcases}
\end{equation}
\noindent
Note that changing the measure does not alter the BSDE’s initial condition. In particular, for every $i \in \{1,2,\ldots,K\}$, we still have $Y_0^{q_i} = y_0^*$, where $y_0^*$ is the first component of the minimizer to \eqref{eq:cont_Var_problem} or, equivalently, \eqref{eq:cont_Var_problem_multi}. Moreover, under the Markovian assumptions in Remark~\ref{rem:Markovian}, the second component of the minimizer is given by
\begin{equation*}
\bigl(Z_t^{q_i,*}\bigr)_{t\in[0,T]} 
\;=\; \bigl(\sigma\bigl(t,X_t^{q_i}\bigr) \odot \mathrm{D}_x v\bigl(t,X_t^{q_i}\bigr)\bigr)_{t\in[0,T]}.
\end{equation*}
Hence, the second component depends only on the trajectory $X_t^{q_i}$, while the underlying functional remains the same for all $q_i$. These observations are crucial from an algorithmic perspective.

\section{Temporal discretization and neural network schemes}\label{seq:discretization}
In this section, we present a fully implementable numerical scheme for approximating the BSDEs as well as the mapping that generates the risk measure introduced in Sections \ref{sec:Val_BSDEs} and \ref{sec:ChangeOfMeasure}. We begin by outlining our algorithm step by step at a high level. In the following sections, we describe in detail the employed temporal discretization algorithms, the BSDE approximation methods, and the neural network–based approach for approximating the risk measures.
\begin{enumerate}
    \item Approximation of the clean value BSDE;
    \item Approximation of the mapping generating the risk measure $\varrho$;
    \item Approximation of the ColVA, CVA, DVA and MVA BSDEs;
    \item Approximation of the FVA BSDE.
\end{enumerate}
It should be emphasized that the enumeration above is best interpreted as layers in our approximation scheme, where each deeper layer depends on approximations obtained from higher layers. Consequently, the computations must be performed in the order specified above. 
\subsection{Temporal Discretization}\label{sec:approximations}
Let $N \in \mathbb{N}^+$. We consider a uniform time grid $\pi(N) \coloneqq \{0 = t_0, t_1, \dots, t_N = T\}$ with constant step size $h \;=\; t_{n+1} - t_n$ for $n = 0,1,\ldots,N-1$. When there is no risk for confusion, we drop $N$ and denote $\pi=\pi(N)$. For simplicity, assume that each of the $P$ derivatives matures at one of these mesh points, that is, for each $p = 1, 2, \ldots, P$, there exists an integer $N_p$ such that $t_{N_p} = T_p$. For $n = 0, 1, \ldots, N-1$, let $\Delta W_n \;\sim\; \mathcal{N}\bigl(\mathbf{0}, h \rho\bigr)$
denote the correlated Gaussian increments, where $\mathbf{0}$ is the zero vector in $\mathbb{R}^{d+2}$ and $\rho$ is the $(d+2)\times(d+2)$ correlation matrix. We write $\Delta\widehat{W}_n$ for the first $d$ components of $\Delta W_n$. Moreover, for each derivative $j$, let $\Delta W^{\mathcal{I}_j}_n$ denote the subset of components of $\Delta W_n$ that are directly or indirectly associated with the $j$-th derivative.

We employ the Euler--Maruyama scheme to discretize the shifted forward SDE defining $X^q$. For the BSDEs, we rewrite the equations as forward SDEs, starting from the (unknown) initial condition. 
\begin{center}
\label{box:XVA_disc_BSDEs}
\fbox{%
\begin{minipage}{0.95\textwidth} 
\begin{align*}
    F&^{X,q}_n\big(X_{0:n-1}^{\pi,q},\Delta W_{0:n-1}\big)\\
    \defeq& X_0^{\pi,q}+ \sum_{k=0}^{n-1}\bigl[ b(t_k, X_k^{\pi,q}) - q(t_k, X_k^{\pi,q} \bigr]h 
    + \sum_{k=0}^{n-1}\sigma(t_k, X_k^{\pi,q})\odot\Delta W_k,\ 
    \widehat{X}^\pi=[X^{\pi,q}]_{i=1}^d,\\[7pt]
\noalign{\hdashrule[1.0ex]{\linewidth}{0.7pt}{4pt}}
F&^{\widehat{V}^j}_n\big(\widehat{X}_{0:n-1}^\pi,\widehat{V}_{0:n-1}^{j,\pi},\widehat{Z}_{0:n-1}^{j,\pi},\Delta \widehat{W}_{0:n-1}^{\mathcal{I}_j}\big) \\\defeq&\widehat{V}_0^{j,\pi} 
    + \sum_{k=0}^{n-1}\Delta A_{k}^{j,\pi} 
    + r_k^\pi \,\widehat{V}_k^{j,\pi}\,h 
    + \sum_{k=0}^{n-1}\widehat{Z}_k^{j,\pi}\cdot\Delta \widehat{W}_k^{\mathcal{I}_j},
    \ j\in\{1,2,\ldots, P\},\  \widehat{V}^{\pi}=\sum_{j=1}^P\widehat{V}^{j,\pi},\\[7pt]
\noalign{\hdashrule[1.0ex]{\linewidth}{0.7pt}{4pt}}
    F&^{\mathrm{ColVA},q}_n\big(X_{0:n-1}^{\pi,q},\widehat{V}^\pi_{0:n-1},\ColVA_{0:n-1}^{\pi,q},Z_{0:n-1}^{\ColVA,\pi,q},\Delta W_{0:n-1}\big) 
\defeq\ColVA_0^{\pi,q} \\
    &+\sum_{k=0}^{n-1}\Bigl( f^\mathrm{ColVA}\bigl(t_k, C_k^\pi\bigr) - r_k\,\ColVA_k^{\pi,q} - \mathrm{CM}_k^{\ColVA,\pi,q} \Bigr) h-\sum_{k=0}^{n-1}Z_k^{\ColVA,\pi,q}\cdot\Delta W_k,\\[7pt]
\noalign{\hdashrule[1.0ex]{\linewidth}{0.7pt}{4pt}}
    F&^{\mathrm{CVA},q}_n\big(X_{0:n-1}^{\pi,q},\CVA_{0:n-1}^{\pi,q},Z_{0:n-1}^{\CVA,\pi,q},\Delta W_{0:n-1}\big)\\\defeq&\CVA_0^{\pi,q}+ \sum_{k=0}^{n-1}\Bigl( r_k \,\CVA_k^{\pi,q} + \mathrm{CM}_k^{\CVA,\pi,q} \Bigr)h
    + \sum_{k=0}^{n-1}Z_k^{\CVA,\pi,q}\cdot\Delta W_k,\\
    F&^{\mathrm{DVA},q}_n\big(X_{0:n-1}^{\pi,q},\DVA_{0:n-1}^{\pi,q},Z_{0:n-1}^{\DVA,\pi,q},\Delta W_{0:n-1}\big)\\\defeq&\DVA_0^{\pi,q} + \sum_{k=0}^{n-1}\Bigl( r_k \,\DVA_k^{\pi,q} + \mathrm{CM}_k^{\DVA,\pi,q} \Bigr)h
    + \sum_{k=0}^{n-1}Z_k^{\DVA,\pi,q}\cdot\Delta W_k,\\
    F&^{\MVA,q}_n\big(X_{0:n-1}^{\pi,q},\mathrm{IM}^{\mathrm{FC},\pi}_{0:n-1},\mathrm{IM}^{\mathrm{TC},\pi}_{0:n-1},\MVA_{0:n-1}^{\pi,q},Z_{0:n-1}^{\MVA,\pi,q},\Delta W_{0:n-1}\big)\defeq \MVA_0^{\pi,q}\\
    &+ \sum_{k=0}^{n-1}\Bigl( f^\MVA\bigl(t_k, \IM^{\FC,\pi}_k, \IM^{\TC,\pi}_k\bigr) - r_k\,\MVA_k^{\pi,q} - \mathrm{CM}_k^{\MVA,\pi,q} \Bigr) h -\sum_{k=0}^{n-1}Z_k^{\MVA,\pi,q}\cdot\Delta W_k,
    \\[7pt]
\noalign{\hdashrule[1.0ex]{\linewidth}{0.7pt}{4pt}}
    F&^{\FVA,q}_n\big(X_{0:n-1}^{\pi,q},\mathbf{xva}_{0:n-1}^{\pi,q},\IM_{0:n-1}^{\TC,\pi},\FVA_{0:n-1}^{\pi,q},Z_{0:n-1}^{\FVA,\pi,q},\Delta W_{0:n-1}\big)\defeq \FVA_0^{\pi,q}\\
    &+\sum_{k=0}^{n-1}\Bigl( f^\FVA\bigl(t_k, \mathbf{xva}_k^{\pi,q}, \IM^{\TC,\pi}_k\bigr) - r_k\,\FVA_k^{\pi,q} - \mathrm{CM}_k^{\FVA,\pi,q} \Bigr) h -\sum_{k=0}^{n-1}Z_k^{\FVA,\pi,q}\cdot\Delta W_k.
\end{align*}
\end{minipage}%
}
\end{center}

Above, for $n = 0, 1, \ldots, N-1$, we have omitted the superscript $q$ for the unshifted forward SDE (\textit{i.e.,} $X_n^\pi\defeq X_n^{\pi,\mathbf{0}}$) and $\Delta A_n^{j,\pi}=A_{t_{n+1}}^{j,\pi}-A_{t_n}^{j,\pi}$. Set
$\mathrm{MPR}_n^\pi = \lfloor \frac{\text{MPR}_t}{h} \rfloor$,
\begin{equation*}
\IM_n^{\FC,\pi} 
= \varrho\Bigl[\bigl(\widehat{V}_{n+\mathrm{MPR}_n^\pi}^\pi - \widehat{V}_n^\pi\bigr)^+ \,\Big\vert\, X_n^\pi \Bigr],
\quad
\IM_n^{\TC,\pi}
= -\varrho\Bigl[\bigl(\widehat{V}_{n+\mathrm{MPR}_n^\pi}^\pi - \widehat{V}_n^\pi\bigr)^- \,\Big\vert\, X_n^\pi \Bigr].
\end{equation*}
Moreover, for $\mathrm{xva} \in \{\ColVA, \CVA, \DVA, \MVA, \FVA\}$, we use the shorthand notation below for the discretized BSDE compensators 
\begin{equation*}
\mathrm{CM}_n^{\mathrm{xva},\pi,q}
= \Big\langle q\bigl(t_n, X_n^{\pi,q}\bigr) \oslash \sigma\bigl(t_n, X_n^{\pi,q}\bigr),\; Z_n^{\mathrm{xva},\pi,q}\Big\rangle.
\end{equation*}
Note that although the above system is presented in the form of discretized forward SDEs, it cannot be solved solely with the Euler–Maruyama scheme. The reason is that, even though the BSDEs can be reformulated as forward SDEs, the initial conditions and the control processes are still unknown at this stage. Moreover, the risk measures, which are not yet specified, cannot typically be expressed in closed form and therefore require numerical approximation. In the following section, we describe how these approximations are carried out.

\subsection{Semi-discrete optimization problems}\label{sec:semi_discrete_opt}
As noted in the previous section, there are several unknown quantities in the discrete system of equations. Theses quantities are the initial conditions, which are real numbers, the control processes and the risk measures which are functions mapping the current time and state to a vector.

We assume that these BSDEs are Markovian in the state variable $X^q$, meaning that the control processes (denoted by $Z$) and the accumulate contractual cashflow (denoted by $A$) are functions of time and state.

To approximate the BSDEs, we adopt a modified version of the deep BSDE solver proposed in \cite{han2018solving}, 
differing in two key aspects. 
First, we approximate multiple BSDEs \emph{simultaneously} rather than solving each equation separately. 
Second, instead of training a distinct neural network at every time step, 
we use a \emph{single} neural network across all time points, with time included as an additional input feature.

Although the original deep BSDE solver starts from the variational problem in \eqref{eq:cont_Var_problem}, 
our framework begins with the multi-equation variant \eqref{eq:cont_Var_problem_multi}. 
To avoid overburdening the reader with notation, 
we present only the continuous-time variational problem in the form of a general BSDE, 
and then move directly to the semi-discrete formulations for each BSDE. 

Below, we present a hierarchy of semi-discrete optimization problems that motivate our fully implementable algorithms. These problems are structured into four layers of optimization, each building upon the solutions obtained at earlier layers. We denote by superscript $ ^{*} $ the solution entities at the first layer, by $ ^{\diamond} $ those at the second layer, by $ ^{\square} $ those at the third layer, and by $ ^{\bullet} $ those at the fourth (top) layer.

Concretely, the final-layer solution $ ^{\bullet} $ depends on the solutions $ ^{*} $, $ ^{\diamond} $, and $ ^{\square} $, the third-layer solution $ ^{\square} $ depends on $ ^{*} $ and $ ^{\diamond} $, the second-layer solution $ ^{\diamond} $ depends on $ ^{*} $, and the first-layer solution $ ^{*} $ is self-contained. We describe each layer’s problem in detail, highlighting how they interact within the overall framework.
 \vspace{0.25cm}\newline
\textbf{Layer 1} (Clean values):\newline
\begin{equation}\label{eq:clean_BSDE_solver}\begin{dcases}
\underset{\widehat{v}_0^1,\ldots \widehat{v}_0^P,\Z}{\mathrm{minimize }}\;\;
 \sum_{j=1}^P\E
 \Big[
   \big|
     \widehat{V}_{N_j}^{j,\pi}-\Delta A_{N_j}^{j,\pi}
   \big|^2
 \Big],\quad \text{where}\\
\widehat{X}_n^\pi=x_0+\sum_{k=0}^{n-1}\widehat{b}\big(t_k,\widehat{X}_k^\pi\big)h + \sum_{k=0}^{n-1}\widehat{\sigma}(t_k,\widehat{X}_k^\pi)\cdot\Delta \widehat{W}_k,\\
\widehat{V}^{j,\pi}_n=F^{\widehat{V}^j}_n\big(\widehat{X}_{0:n-1}^\pi,\widehat{V}_{0:n-1}^{j,\pi},\widehat{Z}_{0:n-1}^{j,\pi},\Delta \widehat{W}_{0:n-1}^{\mathcal{I}_j}\big),\ \widehat{V}^{j,\pi}_0=\widehat{v}_0^j,\
    \widehat{Z}_k^{j,\pi}=\big[\Z(t_k, \widehat{X}_k^\pi)\big]_{i\in\mathcal{K}_j},\ j\in\{1,2,\ldots,P\}.\end{dcases}
\end{equation}
Here $\big[\Z(t_k, \widehat{X}_k^\pi)\big]_{i\in\mathcal{K}_j}$ denotes elements $\mathcal{K}_j=\{\mathcal{J}_j,\mathcal{J}_j+1,\ldots,\mathcal{J}_j+d_j-1\}$ of $\Z(t_k, \widehat{X}_k^\pi)$ where $\mathcal{Z} \colon [0,T] \times \mathbb{R}^d \to \mathbb{R}^{\sum_{j=1}^P d_j}$, and $\widehat{Z}_{\cdot}^{j,\pi} \colon \{0,1,\dots,N-1\} \to \mathbb{R}^{d_j}$ for each $j \in \{1,2,\dots,P\}$. In BSDE terminology, $\widehat{Z}_{\cdot}^{j,\pi}$ approximates the control process for the $j$-th BSDE. We do not apply the deep BSDE solver to each BSDE individually because doing so would make computation time grow linearly with the (potentially large) number of derivatives. Moreover, we do not apply the method directly to $\widehat{V}$ since financial institutions often require clean values at the level of each derivative, rather than at an aggregated netting set. In addition, separately approximating each derivative and then summing the results often reduces bias compared to approximating the sum of derivatives directly, particularly for contracts with non-smooth payoffs. Finally, we use a single function approximator $\mathcal{Z}$ (instead of one per BSDE) to reduce memory usage.

Denote by $\widehat{v}_0^{1,*},\widehat{v}_0^{2,*},\ldots,\widehat{v}_0^{P,*}$ and $\mathcal{Z}^*$ the optimizers solving \eqref{eq:clean_BSDE_solver}.
\vspace{0.25cm}\newline
\textbf{Layer 2} (Risk measures):\newline
At this stage we must specify the risk measure and in this paper we use Value at risk (VaR) at level $\alpha\in(0,1)$ for $\mathrm{IM}^\mathrm{TC}$ and at level $1-\alpha$ for $\mathrm{IM}^\mathrm{FC}$. However, similar optimization problems can be formulated for other risk measures such as expected shortfall (ES), see for instance \cite{barrera2024statistical,crepey2024adaptive}. 

For $t \in [0, T]$, we define the VaR at $t$ over the margin period of risk, $\mathrm{MPR}_t$, as
\begin{equation*}
    \text{VaR}_{\alpha}\big(\widehat{V}_{t+\mathrm{MPR}_t} - \widehat{V}_t \,\big|\, \widehat{\mathcal{F}}_t\big) = \inf \bigl\{\,x \in \mathbb{R} \colon \mathbb{P}\big(\widehat{V}_{t+\mathrm{MPR}_t} - \widehat{V}_t \le x \,\big|\, \widehat{\mathcal{F}}_t\big) \;\ge\; \alpha \bigr\}.
\end{equation*}
To underpin our regression-based algorithms, we utilize the following equivalent formulation of VaR
\begin{equation*}
    \text{VaR}_{\alpha}\left(\widehat{V}_{t+\mathrm{MPR}_t} - \widehat{V}_t \,\big|\, \widehat{\mathcal{F}}_t\right) = \argmin_{q \in \mathbb{R}} \, \mathbb{E}\left[ \varkappa^\alpha(q;\widehat{V}_{t+\mathrm{MPR}_t}-\widehat{V}_t) \,\bigg|\, \widehat{\mathcal{F}}_t \right],
\end{equation*}
where
$\varkappa^\alpha(q;x)=\max\left( \alpha \cdot (x - q),\, (\alpha - 1) \cdot (x - q) \right)$.
This formulation leverages the properties of the check function in quantile regression, facilitating efficient estimation of the conditional VaR through a regression-based optimization technique proposed in \cite{koenker1978regression}.

Below is a semi-discrete optimzation problem for our specific problem setting
\begin{equation}\label{eq:VaR_regression}\begin{dcases}
\underset{q_n^+,q_n^-}{\mathrm{minimize }}\;\;
 \E
\Big[\varkappa^\alpha\bigl(q^+_n(\widehat{X}_n^\pi,\widehat{V}_n^{\pi,*});\bigl(\widehat{V}_{n+\mathrm{MPR}_n^\pi}^{\pi,*}-\widehat{V}_n^{\pi,*}\bigr)^+\bigr)
 \Big]\\
 \quad\quad\qquad+\E
 \Big[
     \varkappa^\alpha\bigl(q^-_n(\widehat{X}_n^\pi,\widehat{V}_n^{\pi,*});\bigl(\widehat{V}_{n+\mathrm{MPR}_n^\pi}^{\pi,*}-\widehat{V}_n^{\pi,*}\bigr)^-\bigr)
 \Big],\quad \text{where}\\
\widehat{X}_n^\pi=x_0+\sum_{k=0}^{n-1}\widehat{b}\big(t_k,\widehat{X}_k^\pi\big)h + \sum_{k=0}^{n-1}\widehat{\sigma}(t_k,\widehat{X}_k^\pi)\cdot\Delta \widehat{W}_k,\\
\widehat{V}^{j,\pi,*}_n=F^{\widehat{V}^j}_n\big(\widehat{X}_{0:n-1}^\pi,\widehat{V}_{0:n-1}^{j,\pi,*},\widehat{Z}_{0:n-1}^{j,\pi,*},\Delta \widehat{W}_{0:n-1}^{\mathcal{I}_j}\big),\quad \widehat{V}^{j,\pi,*}_0=\widehat{v}_0^j,\quad
    j\in\{1,2,\ldots,P\}\\ \widehat{V}^{\pi,*}=\sum_{j=1}^P\widehat{V}^{j,\pi,*}, \quad
    \widehat{Z}_n^{j,\pi,*}=\big[\Z^*(t_n, \widehat{X}_n^\pi)\big]_{i\in\mathcal{K}_j},\quad j\in\{1,2,\ldots,P\}.\end{dcases}
\end{equation}
Note that by solving the optimization problem above, we simultaneously obtain estimates for both the upper and lower tails of the VaR measure.

Denote by $q_0^{+,\diamond}, q_0^{-,\diamond}, q_1^{+,\diamond}, q_1^{-,\diamond}, \ldots, q_{N-1}^{+,\diamond}, q_{N-1}^{-,\diamond}$ the minimizers of the optimization problems defined in equation~\eqref{eq:VaR_regression}. Specifically, for each $n \in \{0, 1, \ldots, N-1\}$, define
\begin{equation*}
\mathrm{IM}_n^{\mathrm{FC},\pi,\diamond} \defeq q_n^{+,\diamond}(\widehat{X}_n^\pi,\widehat{V}_n^{\pi,*}) \quad \text{and} \quad \mathrm{IM}_n^{\mathrm{TC},\pi,\diamond} \defeq -q_n^{-,\diamond}(\widehat{X}_n^\pi,\widehat{V}_n^{\pi,*}),
\end{equation*}
where $\mathrm{IM}_n^{\mathrm{FC},\pi,\diamond}$ and $\mathrm{IM}_n^{\mathrm{TC},\pi,\diamond}$ represent the semi-discrete random variables describing the initial margin paid by the counterparty and payed to the counterparty, respectively. Including the quantity $\widehat{V}_n^{\pi,*}$ as an input to $q_n^+$ and $q_n^-$ is necessary because, in a discrete-time setting, the quantile of the future price movement is not Markovian in the state variable alone. By incorporating $\widehat{V}_n^{\pi,*}$, we effectively recover the Markov property for these quantile processes.
Superscript $''\diamond''$ indicates that the optimized entities depend not only on the current optimization problem but also on the optimization problem in the first layer, where the optimized entities are denoted with a superscript $''*''$.
\vspace{0.25cm}\newline
\textbf{Layer 3} (ColVA, MVA, CVA, DVA):\newline
\begin{equation}\label{eq:layer_3_BSDE_solver}\begin{dcases}
\underset{\mathbf{xva}_0, \mathcal{Z}^\mathrm{ColVA},\mathcal{Z}^\mathrm{MVA},\mathcal{Z}^\mathrm{CVA},\mathcal{Z}^\mathrm{DVA}}{\mathrm{minimize }}\;\;
 \sum_{i=1}^2\E
\Big[\big|\ColVA_{n_\tau\wedge N}^{\pi,q_i^\mathrm{ColVA}}\big|^2\Big]+\E\Big[\big|\MVA_{n_\tau\wedge N}^{\pi,q_i^\mathrm{MVA}}\big|^2\Big]\\
+\E\Big[\big|\CVA_{n_\tau\wedge N}^{\pi,q_i^\mathrm{CVA}}-\mathbbm{1}_{\{n_\tau\leq N \}} \mathbbm{1}_{\{n_\tau = n_{\tau^\mathcal{C}}\}} 
          \LGD^{\mathcal{C}} (\widehat{V}_{n_\tau}^{\pi,*} - C_{n_\tau}^{\pi,*} - \IM_\tau^{\FC,\pi,\diamond})^+\big|^2\Big],\\
+\E\Big[\big|\DVA_{n_\tau\wedge N}^{\pi,q_i^\mathrm{DVA}}-\mathbbm{1}_{\{n_\tau\leq N \}} \mathbbm{1}_{\{n_\tau = n_{\tau^\mathcal{B}}\}} 
          \LGD^{\mathcal{B}} (\widehat{V}_{n_\tau}^{\pi,*} - C_{n_\tau}^{\pi,*} - \IM_\tau^{\TC,\pi,\diamond})^-\big|^2\Big],\quad \text{where}\\
X_n^{\pi,q^\mathrm{xVA}_i}=F^{X,q_i^\mathrm{xVA}}_n\big(X_{0:n-1}^{\pi,q^\mathrm{xVA}_i},\Delta W_{0:n-1}\big),\quad X_0^{\pi,q^\mathrm{xVA}_i}=x_0,\quad \\
n_{\tau^\B}^i=\inf\{n\in\N\colon [X_n^{\pi,q_i^\mathrm{xVA}}]_{d+1}\leq \xi_{t_n}^1\},\quad n_{\tau^\C}^i=\inf\{n\in\N\colon [X_n^{\pi,q_i^\mathrm{xVA}}]_{d+2}\leq \xi_{t_n}^2\},\\
\widehat{V}^{j,\pi,*}_n=F^{\widehat{V}^j}_n\big(\widehat{X}_{0:n-1}^\pi,\widehat{V}_{0:n-1}^{j,\pi,*},\widehat{Z}_{0:n-1}^{j,\pi,*},\Delta \widehat{W}_{0:n-1}^{\mathcal{I}_j}\big),\  \widehat{V}^{j,\pi,*}_0=\widehat{v}_0^j,\ 
    j\in\{1,2,\ldots,P\},\  \widehat{V}^{\pi,*}=\sum_{j=1}^P\widehat{V}^{j,\pi,*},\\  
    \widehat{Z}_n^{j,\pi,*}=\big[\Z^*(t_n, \widehat{X}_n^\pi)\big]_{i\in\mathcal{K}_j},\ \mathrm{IM}_n^{\mathrm{xC},\pi,\diamond}= q_n^{+,\diamond}(\widehat{X}_n^\pi,\widehat{V}_n^{\pi,*})\ j\in\{1,2,\ldots,P\},\ x\in\{\mathrm{T},\mathrm{F}\},\\
     \ColVA_n^{\pi,q_i^{\ColVA}} = 
    F^{\ColVA,q_i^{\ColVA}}_n\big(X_{0:n-1}^{\pi,q_i^{\ColVA}},\widehat{V}^{\pi,*}_{0:n-1},\ColVA_{0:n-1}^{\pi,q_i^{\ColVA}},Z_{0:n-1}^{\ColVA,\pi,q_i^{\ColVA}},\Delta W_{0:n-1}\big),\\
    \CVA_n^{\pi,q_i^{\CVA}} = 
    F^{\CVA,q_i^{\CVA}}_n\big(X_{0:n-1}^{\pi,q_i^{\CVA}},\CVA_{0:n-1}^{\pi,q_i^{\CVA}},Z_{0:n-1}^{\CVA,\pi,q_i^{\CVA}},\Delta W_{0:n-1}\big),\\
    \DVA_n^{\pi,q_i^{\DVA}} = 
    F^{\DVA,q_i^{\DVA}}_n\big(X_{0:n-1}^{\pi,q_i^{\DVA}},\DVA_{0:n-1}^{\pi,q_i^{\DVA}},Z_{0:n-1}^{\DVA,\pi,q_i^{\DVA}},\Delta W_{0:n-1}\big),\\
    \MVA_n^{\pi,q_i^{\MVA}} =F^{\MVA,q_i^{\MVA}}_n\big(X_{0:n-1}^{\pi,q_i^{\MVA}},\mathrm{IM}^{\mathrm{FC},\pi,\diamond}_{0:n-1},\mathrm{IM}^{\mathrm{TC},\pi,\diamond}_{0:n-1},\MVA_{0:n-1}^{\pi,q_i^{\MVA}},Z_{0:n-1}^{\MVA,\pi,q_i^{\MVA}},\Delta W_{0:n-1}\big)\\
    \ColVA_0^{\pi,q_i^{\ColVA}} =\mathrm{colva}_0,\ \CVA_0^{\pi,q_i^{\CVA}} =\mathrm{cva}_0,\ \DVA_0^{\pi,q_i^{\DVA}} =\mathrm{dva}_0,\ \MVA_0^{\pi,q_i^{\MVA}} =\mathrm{mva}_0,\\
Z_n^{\mathrm{xVA},\pi,q_i^\mathrm{xVA}}=\Z^\mathrm{xVA}(t_n,X_n^{\pi,q_i^\mathrm{xVA}}),\quad i\in\{1,2\},\quad \mathrm{x}\in\{\mathrm{Col},\mathrm{C},\mathrm{D},\mathrm{M}\}.
    \end{dcases}
\end{equation}
In the above, we use the convention and define the infimum of an empty set as $\infty$. Denote by $\mathbf{xva}_0^\square = \{\mathrm{colva}_0^\diamond, \mathrm{cva}_0^\square, \mathrm{dva}_0^\square, \mathrm{mva}_0^\square\}$, $\Z^{\ColVA,\diamond}$, $\Z^{\CVA,\square}$, , $\Z^{\DVA,\square}$ and $\Z^{\MVA,\square}$ are the minimizers of the optimization problem above. Superscript $''\square''$ indicates that the optimized entities depend not only on the current optimization problem but also on the optimization problems in the first and second layers. Note that $\ColVA$ does not depend on the optimization problem in the second layer, hence all optimized entities relating to $\ColVA$ have superscript $''\diamond''$.

The summation over $i=1,2$ controls the measure change. Specifically, we assume that for $\mathrm{x} \in \{\mathrm{Col}, \mathrm{C}, \mathrm{D}, \mathrm{M}\}$, $q_1^\mathrm{xVA} \equiv 0$ and $q^\mathrm{xVA}_2$ are chosen appropriately for the relevant $\mathrm{xVA}$. For instance, when computing CVA, it may be appropriate to choose $q^\CVA_2$ such that defaults of the counterparty occur more frequently than in the original measure. The reason for also using the original measure is that, from a practical perspective, we approximate the valuation adjustments along stochastic trajectories. If we only approximate the problem under a fictitious measure, then we would not have an accurate approximation with the original measure, under which we want to evaluate the model.
\begin{remark}
Since there is no interdependency between the xVAs in this layer, the algorithm can be split into four separate ones, one for each xVA. This was done in the numerical experiments to mitigate RAM limitations, whereas for speed it is preferable to use the algorithm as presented above, which was tested with $N=101$ instead of $N=201$ time points.
\end{remark}
\noindent\textbf{Layer 4} (FVA):
\begin{equation}\label{eq:layer_4_BSDE_solver}\begin{dcases}
\underset{\mathrm{fva}_0, \mathcal{Z}^\mathrm{FVA}}{\mathrm{minimize }}\;\;
 \sum_{i=1}^2\E
\Big[\big|\FVA_{n_\tau\wedge N}^{\pi,q_i^\FVA}\big|^2\Big],\quad \text{where}\\
X_n^{\pi,q^\FVA_i}=F^{X,q}_n\big(X_{0:n-1}^{\pi,q^\FVA_i},\Delta W_{0:n-1}\big),\quad X_0^{\pi,q^\FVA_i}=x_0,\quad \\
n_{\tau^\B}^i=\inf\{n\in\N\colon [X_n^{\pi,q_i^\FVA}]_{d+1}\leq \xi_{t_n}^1\},\quad n_{\tau^\C}^i=\inf\{n\in\N\colon [X_n^{\pi,q_i^\FVA}]_{d+2}\leq \xi_{t_n}^2\},\\
\widehat{V}^{j,\pi,*}_n=F^{\widehat{V}^j}_n\big(\widehat{X}_{0:n-1}^\pi,\widehat{V}_{0:n-1}^{j,\pi,*},\widehat{Z}_{0:n-1}^{j,\pi,*},\Delta \widehat{W}_{0:n-1}^{\mathcal{I}_j}\big),\  \widehat{V}^{j,\pi,*}_0=\widehat{v}_0^j,\ 
    j\in\{1,2,\ldots,P\},\  \widehat{V}^{\pi,*}=\sum_{j=1}^P\widehat{V}^{j,\pi,*},\\  
    \widehat{Z}_n^{j,\pi,*}=\big[\Z^*(t_n, \widehat{X}_n^\pi)\big]_{i\in\mathcal{K}_j},\ j\in\{1,2,\ldots,P\},\ \mathrm{IM}_n^{\mathrm{x}_1\mathrm{C},\pi,\diamond}= q_n^{\mathrm{x}_2,\diamond}(\widehat{X}_n^\pi,\widehat{V}_n^{\pi,*}),\ \mathrm{x}\in\{(\mathrm{T},+),(\mathrm{F},-)\},\\
     \ColVA_n^{\pi,q_i^{\ColVA},\diamond} = 
    F^{\ColVA,q_i^\FVA}_n\big(X_{0:n-1}^{\pi,q_i^{\ColVA}},\widehat{V}^{\pi,*}_{0:n-1},\ColVA_{0:n-1}^{\pi,q_i^{\ColVA},\diamond},Z_{0:n-1}^{\ColVA,\pi,q_i^{\ColVA},\diamond},\Delta W_{0:n-1}\big),\\
    \CVA_n^{\pi,q_i^\FVA,\square} = 
    F_n^{\CVA,q_i^\FVA}\big(X_{0:n-1}^{\pi,q_i^\FVA},\CVA_{0:n-1}^{\pi,q_i^\FVA,\square},Z_{0:n-1}^{\CVA,\pi,q_i^\FVA,\square},\Delta W_{0:n-1}\big),\\
    \DVA_n^{\pi,q_i^\FVA,\square} = 
    F^{\DVA,q_i^\FVA}_n\big(X_{0:n-1}^{\pi,q_i^\FVA},\DVA_{0:n-1}^{\pi,q_i^\FVA,\square},Z_{0:n-1}^{\DVA,\pi,q_i^\FVA,\square},\Delta W_{0:n-1}\big),\\
    \MVA_n^{\pi,q_i^\FVA,\square} =F_n^{\MVA,q_i^\FVA}\big(X_{0:n-1}^{\pi,q_i^\FVA},\mathrm{IM}^{\mathrm{FC},\pi,\diamond}_{0:n-1},\mathrm{IM}^{\mathrm{TC},\pi,\diamond}_{0:n-1},\MVA_{0:n-1}^{\pi,q_i^\FVA,\square},Z_{0:n-1}^{\MVA,\pi,q_i^\FVA,\square},\Delta W_{0:n-1}\big),\\
    \ColVA_0^{\pi,q_i^\FVA,\diamond} =\mathrm{colva}_0^\diamond,\ \CVA_0^{\pi,q_i^\FVA,\square} =\mathrm{cva}_0^\square,\ \DVA_0^{\pi,q_i^\FVA,\square} =\mathrm{dva}_0^\square,\ \MVA_0^{\pi,q_i^\FVA,\square} =\mathrm{mva}_0^\square,\\
Z_n^{\ColVA,\pi,q_i^\FVA,\diamond}=\Z^{\ColVA,\diamond}(t_n,X_n^{\pi,q_i^\FVA}),\quad i\in\{1,2\},\\
Z_n^{\mathrm{xVA},\pi,q_i^\FVA,\square}=\Z^{\mathrm{xVA},\square}(t_n,X_n^{\pi,q_i^\FVA}),\quad i\in\{1,2\},\quad \mathrm{x}\in\{\mathrm{C},\mathrm{D},\mathrm{M}\},\\
    \FVA_n^{\pi,q_i^\FVA,\square} =F_n^{\FVA,q_i^\FVA}\big(X_{0:n-1}^{\pi,q_i^\FVA},\mathbf{xva}^{\pi,q_i^\FVA,\square}_{0:n-1},\mathrm{IM}^{\mathrm{TC},\pi,\diamond}_{0:n-1},\MVA_{0:n-1}^{\pi,q_i^\FVA,\square},Z_{0:n-1}^{\FVA,\pi,q_i^\FVA},\Delta W_{0:n-1}\big),\\
    \FVA_0^{\pi,q_i^\FVA}=\mathrm{fva}_0,\quad Z_n^{\FVA,\pi,q_i^\FVA} = \Z^\FVA(t_n,X_n^{\pi,q_i^\FVA})\quad i\in\{1,2\}.
    \end{dcases}
\end{equation}
Denote by $\mathrm{fva}_0^\bullet$ and $\Z^{\FVA,\bullet}$ the minimizers of the optimization problem above.
\begin{remark}
In the optimization problem \eqref{eq:layer_4_BSDE_solver}, the quantity $\Z^\FVA$ is approximated along stochastic trajectories given by $X_\cdot^{\pi,q_i^\FVA}$. This procedure relies on accurate approximations of 
$\ColVA_\cdot^{\pi,q_i^\FVA,\diamond}$, $\CVA_\cdot^{\pi,q_i^\DVA,\square}$, and $\MVA_\cdot^{\pi,q_i^\FVA,\square}$, 
which in turn depend on 
$\Z^{\ColVA,\diamond}$, $\Z^{CVA,\square}$, $\Z^{DVA,\square}$, and $\Z^{MVA,\square}$. 
However, in the level-four optimization problem, these functions are evaluated along trajectories of $X_\cdot^{\pi,q_i^\FVA}$, whereas in the level-three problem they are evaluated along $X_\cdot^{\pi,q_i^\mathrm{xVA}}$, for $\mathrm{x}\in\{\mathrm{Col}, \mathrm{C}, \mathrm{D}, \mathrm{M}\}$ and $i\in\{1,2\}$. 

In a continuous (undiscretized) setting, this distinction poses no issue so long as the underlying processes share the same support. Once we discretize the probability space and approximate expectations via sample means, however, it becomes crucial that the distributions of $X_\cdot^{\pi,q_i^\FVA}$ and $X_\cdot^{\pi,q_i^\mathrm{xVA}}$ exhibit sufficient overlap in their simulated paths. Without adequate overlap, the variance of the estimators can increase and the numerical approximations for quantities such as $\ColVA_\cdot^{\pi,q_i^\FVA,\diamond}$, $\CVA_\cdot^{\pi,q_i^\DVA,\square}$, and $\MVA_\cdot^{\pi,q_i^\FVA,\square}$ may become unreliable.
\end{remark}

\subsection{Fully discrete optimization problems}
To obtain a fully discrete optimization problem, we replace each expectation in the previous sections with sample means based on $M \in \mathbb{N}$ samples. Specifically, for each optimization problem, we draw $M$ realizations of  $\Delta W\defeq\{\Delta W_n\}_{n\in\{0,\ldots,N-1\}}$ (or $\Delta \widehat{W}\defeq\{\Delta\widehat{W}\}_{n\in\{0,\ldots,N-1\}}$ for layer 1 and 2). For instance, in the layer-3 optimization problem, we use the same realizations of the Gaussian increments for each of the trajectories $X_\cdot^{\pi,q_i^\mathrm{xVA}}$, $i\in\{1,2\}$, $\mathrm{x}\in\{\mathrm{Col}, \mathrm{C}, \mathrm{D}, \mathrm{M}\}$. As a result, these trajectories only differ by the shift in the drift coefficient, which reduces the number of simulated paths.

To avoid overburdening the reader with additional notation, we do not explicitly restate the fully discretized optimization problems here. Instead, the approach of sampling and the practical implementation details are directly analogous to the semi-discrete setting, with expectations simply replaced by empirical averages.

\subsection{Neural networks}
We next describe our neural network architectures and training procedures, structured across the various layers of the xVA problem. Below, we detail the network designs, hyperparameters, and optimization strategies used at each layer. For further details we refer to \url{https://github.com/AlessandroGnoatto/MultiLayerDeepXVA}.
\vspace{0.25cm}\newline
\textbf{Layer 1} (Clean values):\newline
We represent $\widehat{v}_0,\widehat{v}_1,\ldots,\widehat{v}_P$ with scalar valued trainable parameters and $\widehat{\Z}$ with a fully connected neural network $[0,T]\times\R^d\to\R^{\sum_{j=1}^Pd_j}$. The network consists of three hidden layers, each with 100 neurons and we apply the relu activation function to the input layer and each hidden layer. We use $2^{20}$ samples in the training (optimization) procedure, with a batch size $2^{11}$, with the Adam optimization algorithm, \cite{kingma2014adam}. 
\vspace{0.25cm}\newline
\textbf{Layer 2} (Quantile regression):\newline
For each $n \in \{0,1,\ldots,N-1\}$, the pairs $\{q^+_\alpha, q^-_\alpha\}$ are represented by a fully connected neural network, $\R^d\times\R^P\to\R^2$, with three hidden layers, each comprising 16 neurons. 
Again, we apply ReLU activations to the input and hidden layers. 
We use $2^{12}$ samples, a batch size of $2^{11}$, and the Adam optimizer. 
The data is split into training and validation sets; training continues until no improvement is observed on the validation set for 100 consecutive epochs.
We adopt this approach due to the separate network at each time point and the associated uncertainty in how much data is required at the current time step.
\vspace{0.25cm}\newline
\textbf{Layer 3} (ColVA, CVA, DVA, MVA):\newline
We represent the initial values of all four xVAs, $\mathbf{xva}_0$, by a trainable parameter in $\mathbb{R}^4$. 
The functions $\mathcal{Z}^{\mathrm{ColVA}}$, $\mathcal{Z}^{\mathrm{MVA}}$, $\mathcal{Z}^{\mathrm{CVA}}$, and $\mathcal{Z}^{\mathrm{DVA}}$ are collectively parameterized by a fully connected network $[0,T]\times\mathbb{R}^{d+2} \to \mathbb{R}^{d+2}$,
with three hidden layers of 50 neurons each and ReLU activation in the input and hidden layers. 
We use $2^{20}$ samples, a batch size of $2^{11}$, and train with Adam. The reduced network size relative to Layer 1 reflects the lower-dimensional outputs. 
\vspace{0.25cm}\newline
\textbf{Layer 4} (FVA):\newline
We represent $\mathbf{fva}_0$ by a trainable scalar parameter. The function $\mathcal{Z}^\FVA$ is modeled by a fully connected neural network $[0,T]\times\mathbb{R}^{d+2} \to \mathbb{R}^{d+2}$,
also with three hidden layers of 50 neurons each and ReLU activation at the input and hidden layers. 
As with the other layers, we use $2^{20}$ samples, a batch size of $2^{11}$, and the Adam optimizer. 
Here too, the reduced network size relative to Layer 1 stems from the lower-dimensional output.

\section{Error analysis}\label{sec:error_analysis}
In this section, we extend the a posteriori error estimate from \cite{han2020convergence} 
to BSDEs involving stopping times (or equivalently, BSDEs on domains). In addition, we allow 
the driver to take approximations 
coming from preceding layers as inputs. Consequently, our bound for the overall simulation error 
incorporates not only the a posteriori terms for the current layer 
(as in \cite{han2020convergence}) but also the a posteriori terms inherited 
from earlier layers.

\subsection{A posteriori error bounds for BSDEs with stopping times}
Our proof follows a line of argument similar to that of \cite{han2020convergence}, 
but we must rely on preliminary results from \cite{bouchard2009strong} (instead of \cite{bender2008time}), 
which provide estimates for backward discretization schemes on domains. 
As in \cite{han2020convergence}, we first establish a stability bound for the difference 
between (i) the BSDE approximation obtained via the scheme of \cite{bouchard2009strong} 
and (ii) the deep BSDE solver proposed here. 
We then apply the error bounds of \cite{bouchard2009strong} to derive the final 
simulation‐error estimate for our method.

Let $b\colon [0,T]\times\R^d\to\R^d$, $\sigma\colon [0,T]\times\R^d\to\R^d$, $f\colon [0,T]\times\R\times\R\times\R^d\to\R$, $g\colon[0,T]\times\R^d\to\R$ and $v\colon[0,T]\times\R^d\to\R$.
 and consider decoupled FBSDEs on the form 
\begin{equation}\label{eq:FBSDE_err_an}
\begin{cases}
        X_t=x_0 + \int_0^tb(s,X_s)\d s + \int_0^t\sigma(s,X_s)\d W_s,\\
        Y_t^v = g(\tau\wedge T,X_{\tau\wedge T}) +\int_{t}^{ T}\mathbb{I}_{\{s\leq \tau\}}f(s,v(s,X_s),Y_s^v,Z_s^v)\d s -\int_t^ TZ_s^v\cdot\d W_s,
        \end{cases}
\end{equation}
where $\tau\defeq\inf\{t\in[0,T]\colon X_t\notin\mathcal{O}\}$. Here the boundary is an open, piecewise smooth and connected set $\mathcal{O}\subset\R^d$. \begin{assumption}\label{ass:BSDE_coeff}
Let $t\in[0,T]$, $x,x_1,x_2\in\R^d$, $y,y_1,y_2\in\R$ and $z_1,z_2\in\R^d$, denote by $\Delta x=x_1-x_2$, $\Delta y=y_1-y_2$, $\Delta v=v_1(t,x)-v_2(t,x)$, $\Delta y=y_1-y_2$ and $\Delta z=z_1-z_2$. We assume
    \begin{enumerate}
        \item \textbf{One sided Lipschitz continuity of $f$ in $y$}.\newline There exists a, possibly negative, constant $k_y$ such that \begin{equation*}
            \big(f(t,v,y,z_1)-f(t,v,y,z_2)\big)\Delta y\leq k_y\Delta y^2.
        \end{equation*}
        \item \textbf{Lipschitz continuity of $(b,\sigma,f,g,v)$}.\newline There exist positive constants $L_b,L_\sigma,L_f,L_g,L_v$ such that\begin{align*}
        &|b(t,x_1)-b(t,x_2)|^2\leq L_b|\Delta x|^2, \quad
                |\sigma(t,x_1)-\sigma(t,x_2)|^2\leq L_\sigma|\Delta x|^2, \\
        &|f(t,v_1(t,x_1),y_1,z_1)-f(v_2(t,x_1),y_2,z_2)|^2\leq L_f\big(|\Delta v|^2 + |\Delta y|^2 + |\Delta z|^2\big),\\
            &|g(t,x_1)-g(t,x_2)|\leq L_g|\Delta x|^2,\quad    |v(t,x_1)-v(t,x_2)|\leq L_v|\Delta x|^2. 
        \end{align*} This guarantees that $f$ is Lipschitz continuous in $x$, with Lipschitz constant $L_f'=L_fL_v$, and we may define $f^v(t,x,y,z)\defeq f(t,v(t,x),y,z)$.
        \item \textbf{Linear growth of $(b,\sigma,f,g,v)$}.\newline
        There exist positive constants $K_b,K_\sigma,K_f,K_g,K_v$ such that
        \begin{align*}
        &|b(t,x_1)-b(t,x_2)|^2\leq K_b|\Delta x|^2, \quad
                |\sigma(t,x_1)-\sigma(t,x_2)|^2\leq K_\sigma|\Delta x|^2, \\
        &|f(t,v(t,x),y,z)|\leq K_f(1+|x| + |y| + |z|),\quad|g(t,x)|\leq K_g(1+|x|),\quad    |v(t,x)|\leq K_v(1+| x|). 
        \end{align*}
        \item \textbf{Hölder continuity of $(f,b,\sigma)$ in $t$}.\newline
        There exist positive constants $C_b,C_\sigma,C_f$ such that\begin{align*}
            &|b(t,x)-b(s,x)|\leq C_b|t-s|^{1/2}|,\quad |\sigma(t,x)-\sigma(s,x)|\leq C_\sigma|t-s|^{1/2}|,\\
            &f(t,v(t,x),y,z)-f(s,v(s,x),y,z)|\leq C_f|t-s|^{1/2}.
        \end{align*}
        This implicitly assumes Hölder-$\tfrac{1}{2}$ continuity of $v$ in $t$.
    \end{enumerate}
\end{assumption}
Under these standard assumptions, there exists a unique solution $(X,Y,Z)\in\mathbb{S}^2(\mathbb{F})\times \mathbb{S}^2(\mathbb{F})\times \mathbb{H}^2(\mathbb{F})$ to \eqref{eq:FBSDE_err_an}, see \textit{e.g.,} \cite{peng1991probabilistic,pardoux2005backward}.

We now present the backward discretization scheme proposed in \cite{bouchard2009strong}. For $n\in\{0,1,\ldots,N-1\}$ let $X_n^\pi$ be the Euler--Maruyama approximation of $X$ and   define the discrete stopping time as $n_{\tau^\pi}=\inf\{n\in\{0,1,\ldots N\}\colon X_n^\pi\notin\mathcal{O}\}$ and $\tau^\pi\defeq t_{n_{\tau^\pi}}$. For the backward SDE in \eqref{eq:FBSDE_err_an}, we define the  backward scheme:
\begin{equation}\label{eq:bouchard_scheme}
        \begin{cases}
\widebar{Y}_n^{\pi,v}=\E[\widebar{Y}_{n+1}^{\pi,v}] + \mathbb{I}_{\{t_n<\tau^\pi\}}f(t_n,v(t_n,X_n^\pi),\widebar{Y}_n^{\pi,v},\widebar{Z}_n^{\pi,v})h;\quad \widebar{Y}_N^{\pi,v}=g(\tau^\pi\wedge T,X_{\tau^\pi\wedge T}),\\
    \widebar{Z}_n^{\pi,v}=\frac{1}{h}\E[\widehat{Y}_{n+1}^{\pi,v}\Delta W_n\,|\,\F_{t_n}].
        \end{cases}
    \end{equation}
For $t\in[0,T]$, let $\phi(t)=\sup\{n\colon\,\pi\ni nh \leq t\}$. Define the discretization error obtained by the above scheme as
    \begin{equation*}
        \text{Err}(h)_\vartheta^2\defeq\E\big[\sup_{t\in[0,\vartheta]}|Y_t^v-\widebar{Y}_{\phi(t)}^{\pi,v}|^2\big] + \int_0^\vartheta\E|Z_t^v-\widebar{Z}_{\phi(t)}^{\pi,v}|^2\d t.
    \end{equation*}

The following assumptions regarding the domain is a non-technical summary of Assumptions \textbf{D1}, \textbf{D2} and \textbf{C} from \cite{bouchard2009strong}, to which we refer for a detailed presentation.

\begin{assumption}\label{ass:domain}
\begin{itemize}
    \item[\textbf{\emph{(D1)}}] \textbf{Domain regularity}.\\ The domain $\mathcal{O}$ is defined as the intersection of finitely many $C^2$ (twice continuously differentiable) and bounded regions. This ensures that the overall boundary of $\mathcal{O}$ is well smooth except possibly at finitely many "corner" points, which is crucial for constructing a smooth distance function and applying standard PDE results.
    
    \item[\textbf{\emph{(D2)}}] \textbf{Uniform exterior and interior conditions}.\\ For every boundary point, there exists an external ball (touching the domain only at that point) and an internal cone (with its vertex at that point) contained within the domain. This prevents the boundary from having extreme flatness or inward cusps, ensuring that the process has a uniformly controlled mechanism to exit the domain.
    
    \item[\textbf{\emph{(C)}}] \textbf{Non-degeneracy and ellipticity near the boundary}.\\ The diffusion is uniformly non-degenerate in the normal direction at regular parts of the boundary and uniformly elliptic near corner points. In practical terms, the process does not "stick" to the boundary; it is able to exit properly. This is essential for controlling the error in the time-discretization, ultimately leading to the $O(h^{1/4-\varepsilon})$ convergence rate for the BSDE approximation.
\end{itemize}
\end{assumption}

\begin{assumption}[Smooth terminal condition with bounded derivatives]\label{ass:terminal_cond}
The terminal condition function $g$ is assumed to be one time continuously differentiable in $t$ and twice continuously differentiable in $x$ with bounded partial derivatives (1 times in $t$ and twice in $x$), \textit{i.e.,} there exists a positive constant $K$ such that $g \in C^{1,2}([0,T]\times \mathbb{R}^d)\text{ and }\|\partial_t g\|_\infty + \|D_x g\|_\infty + \|D^2_{xx} g\|_\infty \leq K$.
  Here $\|\cdot\|_\infty$ represents the infinity-norm.
\end{assumption}

\begin{theorem}[Bouchard and Menozzi, 2009, \cite{bouchard2009strong}]\label{thm:bounchard}
Let Assumptions~\ref{ass:BSDE_coeff}-\ref{ass:terminal_cond} be satisfied, then for each $\epsilon\in(0,1/2)$, there exists a constant $C_\epsilon$, such that \begin{equation*}
    \emph{Err}(h)_{\tau\wedge\tau^\pi}^2\leq C_\epsilon h^{1-\epsilon},\text{ and
    }    \emph{Err}(h)_T^2\leq C_\epsilon h^{1/2-\epsilon}.
\end{equation*}
\end{theorem}
We define two additional schemes for the BSDE in \eqref{eq:FBSDE_err_an}, namely the (not fully implementable) forward scheme \begin{equation}\label{eq:forward_scheme}
\begin{cases}
        \widetilde{Y}_0=\mathcal{Y}_0^\pi,\quad \widetilde{Z}_0=\mathcal{Z}_0^\pi,\\
        \widetilde{Y}_{n+1}=\widetilde{Y}_n-\mathbb{I}_{\{t_n<\tau^\pi\}}f(t_n,v(t_n,X_n^\pi),\widetilde{Y}_n^{\pi,v},\widetilde{Z}_n^{\pi,v})h +\widetilde{Z}_n^{\pi,v}\cdot\Delta W_n,\\
        \widetilde{Z}_n^{\pi,v}=\mathcal{Z}(t_n,X_n^\pi).
        \end{cases}
\end{equation}
and the generic (neither forward nor backward) scheme 
\begin{equation}\label{eq:generic_scheme}
        \begin{cases}
Y_n^{\pi,v}=\E[Y_{n+1}^{\pi,v}] + \mathbb{I}_{\{t_n<\tau^\pi\}}f(t_n,v(t_n,X_n^\pi),Y_n^{\pi,v},Z_n^{\pi,v})h,\\
    Z_n^{\pi,v}=\frac{1}{h}\E[Y_{n+1}^{\pi,v}\Delta W_n\,|\,\F_{t_n}].
        \end{cases}
    \end{equation}
Note that both the backward scheme \eqref{eq:bouchard_scheme} and the forward scheme \eqref{eq:forward_scheme} can be cast into the form of \eqref{eq:generic_scheme}.
Finally, we impose the following assumptions to guarantee that \eqref{eq:forward_scheme} is well-defined.
\begin{assumption}\label{ass:ANN}
    The initial conditions of the forward scheme, $\mathcal{Y}_0^\pi\in\R$ and $\mathcal{Z}_0^\pi\in\R^d$ are bounded. For $t\in[0,T)$, the functions $\mathcal{Z}(t,\cdot)$ belongs to the set of measurable functions $\R^d\to\R^d$ and satisfies a linear growth condition. 
\end{assumption}
The following Theorem generalizes an a posteriori estimate in \cite{han2020convergence} to BSDEs in a domain. However, for our purpose, it is sufficient to consider decoupled BSDEs. This simplifies the error analysis and relaxes some strong assumptions on the coeffiecients.
\begin{theorem}\label{thm:sim_error}Let assumptions \ref{ass:BSDE_coeff}-\ref{ass:ANN} be satisfied, then for each $\epsilon\in(0,1/2)$, there exists a constant $C_\epsilon$ such that
\begin{align*}
    &\E\big[\sup_{t\in[0,T]}|X_t-X^\pi_{\phi(t)}|^2 + \underset{k\in\{0,1,\ldots,N-1\}}{\mathrm{max}} \E\big[\sup_{t\in[t_k,t_{k+1}]}\I_{\{t<\vartheta\}}|Y_t^{v_1}-\widetilde{Y}_{\phi(t)}^{\pi,v_2}|^2\big] + \int_0^\vartheta\E|Z_t^{v_1}-\widetilde{Z}_{\phi(t)}^{\pi,v_2}|^2\d t\\
    &\leq \begin{cases}C(h^{1-\epsilon}+  \underset{n\in\{0,\ldots,N-1\}}{\max}\E|v_1(t_n,X_n^\pi)-v_2(t_n,X_n^\pi)|^2 + \E\big|g(\tau^\pi\wedge T,X_{\tau\wedge T}^\pi)-\widetilde{Y}_N^{\pi,v}\big|^2)&\vartheta=\tau\wedge\tau^\pi,\\
    C(h^{1/2-\epsilon} +  \underset{n\in\{0,\ldots,N-1\}}{\max}\E|v_1(t_n,X_n^\pi)-v_2(t_n,X_n^\pi)|^2 +\E\big|g(\tau^\pi\wedge T,X_{\tau\wedge T}^\pi)-\widetilde{Y}_N^{\pi,v}\big|^2)&\vartheta=T.
    \end{cases}
\end{align*}
    \begin{proof}
    First, note that the error induced from the approximation of $X$ is the standard Euler--Maruyama discretization error of order $\tfrac{1}{2}$.
    
For the errors induced by the approximations of $(Y^{v_1},Z^{v_1})$ we follow closely the proof in \cite{han2020convergence} in which the special case $\mathcal{O}=\R^d$ and $v_1=v_2$ is proven (although more general in the coupling between the forward and backward SDEs). 

Define $\delta Y_n=Y_n^{\pi,1}-Y_n^{\pi,2}$, $\delta Z_n=Z_n^{\pi,1}-Z_n^{\pi,2}$,  $\delta f_n=f(t_n,v(t_n,X_n^\pi),Y_n^{\pi,1},Z_n^{\pi,1})\\-f(t_n,v(t_n,X_n^\pi),Y_n^{\pi,2},Z_n^{\pi,2})$, and $\delta v_n=v_1(t_n,X_n^\pi)-v_2(t_n,X_n^\pi)$        
where $\{Y_n^{\pi,1}\}_{n\in\{0,\ldots,N\}}$,\\ $\{Z_n^{\pi,1}\}_{n\in\{0,\ldots,N-1\}}$ and $\{Y_n^{\pi,2}\}_{n\in\{0,\ldots,N\}}$, $\{Z_n^{\pi,2}\}_{n\in\{0,\ldots,N-1\}}$ are two solutions to \eqref{eq:generic_scheme}. Directly from \eqref{eq:generic_scheme} we have $\delta Y_n=\E[\delta Y_{n+1}|\F_{t_n}]+\mathbb{I}_{\{t_n<\tau^\pi\}}\delta f_n$ and $\delta Z_n=\frac{1}{h}\E[\delta Y_n\Delta W_n|\F_{t_n}]$. By the martingale representation theorem we know that there exists a unique process $(\delta Z_t)_{t\in[t_n,t_{n+1}]}\in\mathbb{H}^2(\mathbb{F})$ such that $\delta Y_{n+1}=\delta Y_n + \int_{t_n}^{t_{n+1}}\delta Z_t \d W_s$. By the It\^o isometry and subsequently by the root-mean square and geometric
mean inequality (RMS-GM inequality) and Lipschitz continuity of $f$ in the second and fourth components, for any $\lambda>0$ we have\begin{align*} 
    \E|\delta &Y_{n+1}|^2=\E|\delta Y_n - \mathbb{I}_{\{t_n<\tau^pi\}}\delta f_n h|^2 + \int_{t_n}^{t_{n+1}}|\delta Z_t|^2\d t\\
    \geq&\E|\delta Y_n|^2 - 2h\E\Big[\big( f(t_n,v_1(t_n,X_n^\pi),Y_n^{\pi,1},Z_n^{\pi,1}) - f(t_n,v_1(t_n,X_n^\pi),Y_n^{\pi,2},Z_n^{\pi,1})\big)\mathbb{I}_{\{t_n<\tau^\pi\}}\delta Y_n\Big] \\
    &- 2h\E\Big[\big( f(t_n,v_1(t_n,X_n^\pi),Y_n^{\pi,2},Z_n^{\pi,1}) - f(t_n,v_2(t_n,X_n^\pi),Y_n^{\pi,2},Z_n^{\pi,2})\big)\mathbb{I}_{\{t_n<\tau^\pi\}}\delta Y_n\Big]+ \int_{t_n}^{t_{n+1}}|\delta Z_t|^2\d t\\ 
    \geq&\E|\delta Y_n|^2 - 2hk_f\E|\delta Y_n|^2 -h(\lambda\E|\delta Y_n|^2+ \lambda^{-1}(L_v\E|\delta v_n|^2+L_f\E|\delta Z_n|^2))+ \int_{t_n}^{t_{n+1}}|\delta Z_t|^2\d t.
\end{align*}
Noting that $\int_{t_n}^{t_{n+1}}|\delta Z_t|^2\d t=h\E|\delta Z_n|^2$ and collecting terms yield
\begin{equation}\label{eq:split_YZ}
     \E|\delta Y_{n+1}|^2\geq(1-h(2k_y+\lambda))\E|\delta Y_n|^2 +h(1-\lambda^{-1}L_f)\E|\delta Z_n|^2 -h\lambda^{-1}L_v\E|\delta v_n|^2
\end{equation}
Choosing $\lambda\geq L_f$, and assuming sufficiently small $h$ such that $(2k_y+\lambda)<1$ yield
\begin{equation*}
    \E|\delta Y_n|^2 \leq(1-h(2k_y+\lambda))^{-1}\big(\E|\delta Y_{n+1}|^2 + h\lambda^{-1}L_v\E|\delta v_n|^2\big). 
\end{equation*}
Setting $A=(1-h(2k_y+\lambda))^{-1}$, and using induction (a Discrete Grönwall type inequality), we obtain \begin{equation*}
    \E|\delta Y_n|^2\leq A^N\E|\delta Y_N|^2 + TL_v\lambda^{-1}A^N\max_{k\in\{0,1,\ldots,N\}}\E|\delta v_k|^2.
\end{equation*}
To find a bound for $\E|\delta Z_n|^2$ we choose $\lambda=2L_f$ in \eqref{eq:split_YZ} to obtain
\begin{equation*}
    \E|\delta Z_n|^2h\leq 2(\E|\delta Y_{n+1}|^2-\E|\delta Y_n|^2) + 4h(k_y+L_f)\E|\delta Y_n|^2) + h\lambda L_vL_f^{-1}\E|\delta v_n|^2.
\end{equation*}
By summing the above over all $n$, and noting that the first term creates a telescoping sum, we obtain
\begin{align*}
    \sum_{n=0}^{N-1}\E|\delta Z_n|^2h&\leq 2\E|\delta Y_N|^2 + 4T(k_y+L_f)\max_{k\in\{0,1,\ldots,N\}}\E|\delta Y_k|^2 + TL_vL_f^{-1}\max_{k\in\{0,1,\ldots,N\}}\E|\delta v_k|^2\\
    &\leq  (2+4T(k_y+L_f)A^N)\E|\delta Y_N|^2 + T\big(4TL_v(k_y+L_f)+L_vL_f^{-1}\big)\lambda^{-1}A^N\max_{k\in\{0,1,\ldots,N\}}\E|\delta v_k|^2
\end{align*}
Finally, Letting $Y^{\pi,1}=\widebar{Y}^{\pi,v_1}, Z^{\pi,1}=\widebar{Z}^{\pi,v_1}$ and $Y^{\pi,1}=\widetilde{Y}^{\pi,v_2}, Z^{\pi,1}=\widetilde{Z}^{\pi,v_2}$, add and subtract the exact solution $(Y^{v_1},Z^{v_1})$ and applying Theorem~\ref{thm:bounchard} prove the claim.
\end{proof}
\end{theorem}
\subsection{Applying the error bounds for the recursive xVA BSDEs}
Next, we apply Theorem~\ref{thm:sim_error} to the recursive xVA BSDE.
We use the following notation for the simulation errors for the clean values and the xVA BSDEs 
    \begin{equation*}
\mathrm{Err}^V(h)^2\defeq\sum_{j=1}^P\E\Big[\sup_{t\in[0,T]}\big|\widehat{V}^j_t-\widehat{V}_{\phi(t)}^{j,\pi,*}\big|^2\Big] + \sum_{j=1}^P\int_0^T\E\big|\widehat{Z}_t^j-\widehat{Z}_{\phi(t)}^{j,\pi,*}\big|^2,
\end{equation*}
and for $(\mathrm{x},\mathrm{y})\in\{(\mathrm{Col},*),(\mathrm{M},\diamond),(\mathrm{C},\diamond),(\mathrm{D},\diamond),(\mathrm{F},\bullet)\}$
\begin{equation*}
        \mathrm{Err}^\mathrm{xVA}(h)^2\defeq      \E\Big[\sup_{t\in[0,T]}\big|\mathrm{xVA}_t-\mathrm{xVA}_{\phi(t)}^{\pi,\mathrm{y}}\big|^2\Big] + \int_0^T\E\big|Z_t^\mathrm{xVA}-Z_{\phi(t)}^{\mathrm{xVA},\pi,\mathrm{y}}\big|^2\d t,
    \end{equation*}
and we define the a posteriori error terms as
    \begin{align*}
        &\epsilon_{V^j}\defeq \E|\widehat{V}_t^j-\widehat{V}_{\phi(t)}^{j,\pi,*}|^2,\quad \epsilon_\ColVA \defeq \E\big|\ColVA_{n_{\tau^\pi}\wedge N}^{\pi,\diamond}\big|^2,\quad \epsilon_\MVA \defeq \E\big|\MVA_{n_{\tau^\pi}\wedge N}^{\pi,\square}\big|^2,\\ &\epsilon_\CVA \defeq \E\big|\CVA_{n_{\tau^\pi}\wedge N}^{\pi,\square}-\mathbbm{1}_{\{n_\tau\leq N \}} \mathbbm{1}_{\{n_\tau = n_{\tau^\mathcal{C}}\}} 
          \LGD^{\mathcal{C}} (\widehat{V}_{n_\tau}^{\pi,*} - C_{n_\tau}^{\pi,*} - \IM_\tau^{\FC,\pi,\diamond})^+\big|^2,\\
    &\epsilon_\DVA\defeq\E\Big[\big|\DVA_{n_\tau\wedge N}^{\pi,\square}-\mathbbm{1}_{\{n_\tau\leq N \}} \mathbbm{1}_{\{n_\tau = n_{\tau^\mathcal{B}}\}} 
          \LGD^{\mathcal{B}} (\widehat{V}_{n_\tau}^{\pi,*} - C_{n_\tau}^{\pi,*} - \IM_\tau^{\TC,\pi,\diamond})^-\big|^2\Big],\ \\
          &\epsilon_\FVA \defeq \E\big|\FVA_{n_{\tau^\pi}\wedge N}^{\pi,\bullet}\big|^2.
    \end{align*}  
Finally, for $t\in[0,T]$ and $n\in\{0,1,\ldots,N-1\}$, we denote by $\IM_t=\mathrm{Concat}(\IM_t^\TC,\IM_t^\FC)$ and $\IM_n^{\pi,\diamond}=\mathrm{Concat}(\IM_n^{\TC,\pi,\diamond},\IM_n^{\FC,\pi,\diamond})$, the vector concatenations of the IMs to and from the counterparty and their approximations, respectively. 
 \begin{theorem}
 Let Assumptions~\ref{ass:BSDE_coeff} and \ref{ass:ANN} be satisfied for each of the $P$ clean value BSDEs and let Assumptions~\ref{ass:BSDE_coeff}-\ref{ass:ANN} be satisfied for each of the xVA BSDEs. Then, for any $\epsilon\in(0,1/2)$, there exists a constant $C$ such that
         \begin{align*} 
    &\mathrm{Err}^V(h)^2
    \leq 
    C\big(h +\epsilon_V\big),\\
    &\mathrm{Err}^\ColVA(h)^2
    \leq 
    C\big(h^{1/2-\epsilon} +\epsilon_V  +\epsilon_\ColVA\big),\\
&\mathrm{Err}^\MVA(h)^2
    \leq 
    C\big(h^{1/2-\epsilon} +\epsilon_V +\max_{n\in\{0,\ldots,N-1\}}\E|\IM_{t_n}-\IM_n^{\pi,\diamond}|^2 +\epsilon_\MVA\big),\\
&\mathrm{Err}^\CVA(h)^2
    \leq 
    C\big(h^{1/2-\epsilon} +\epsilon_V +\max_{n\in\{0,\ldots,N-1\}}\E|\IM_{t_n}-\IM_n^{\pi,\diamond}|^2 +\epsilon_\CVA\big),\\
&\mathrm{Err}^\DVA(h)^2
    \leq 
    C\big(h^{1/2-\epsilon} +\epsilon_V +\max_{n\in\{0,\ldots,N-1\}}\E|\IM_{t_n}-\IM_n^{\pi,\diamond}|^2 +\epsilon_\DVA\big),\\
&\mathrm{Err}^\FVA(h)^2
    \leq 
    C(h^{1/2-\epsilon} +\epsilon_V +\max_{n\in\{0,\ldots,N-1\}}\E|\IM_{t_n}-\IM_n^{\pi,\diamond}|^2 +\epsilon_\ColVA+\epsilon_\MVA+\epsilon_\CVA+\epsilon_\DVA+\epsilon_\FVA).
    \end{align*}
 \end{theorem}  
 \begin{proof}
     Under Assumption~\ref{ass:BSDE_coeff}, Theorem 1' in \cite{han2020convergence} can be used to obtain the first bound. Note that the smoothness assumptions on $g$ can be relaxed as long as $g$ is Lipschitz continuous and satisfies a linear growth condition.

     The remaining bounds follow directly from applying Theorem~\ref{thm:sim_error} recursively. 
\end{proof}
The theorem above provides error bounds for the recursive simulation errors in terms of the step size $h$, a posteriori terms, and the mean-squared errors of our IM approximations. While the first two can be controlled in the current framework, controlling the mean-squared error requires extra work. For example, one could adapt the approach proposed in \cite{barrera2024statistical}, which employs Rademacher bounds and an a posteriori Monte Carlo procedure to estimate distances to the true conditional VaR.\vspace{0.25cm}\newline
\noindent\textbf{Qualitative discussion about the assumptions}:\newline
Some of the key assumptions introduced above might initially appear overly restrictive. However, by employing several standard techniques, one can demonstrate that these conditions are indeed met by realistic financial models. In the following, we discuss each assumption in detail.
\begin{itemize}
    \item \textbf{Bounded domain}:\newline
    The domain is bounded in the spatial components that are related to the bank’s and the counterparty’s asset processes. For the remaining, unbounded components, we may impose artificial boundary conditions to make $\mathcal{O}$ bounded. 
    \item\textbf{Smooth terminal condition function}:\newline
    This assumption is satisfied by the BSDEs for clean values, ColVA, MVA, and FVA. However, for FVA, the error in the BSDE is bounded by a posteriori terms associated with the CVA and DVA BSDEs. Since these BSDEs incorporate the pay-off functions of the derivatives in the portfolio, they may not fulfill the smoothness requirements imposed on $g$. On the other hand, assuming that the domain $\mathcal{O}$ is bounded (for instance by imposing artificial boundary conditions), the condition holds for linear products such as swaps, futures, and forwards. Moreover, for European payoffs, smoothing techniques can be applied effectively.
    \item \textbf{Hölder-$\tfrac{1}{2}$ continuous driver in $t$}:\newline
    Although Hölder-$\tfrac{1}{2}$ continuity is not generally required for the existence 
    and uniqueness of BSDE solutions, it is often imposed to achieve a $\tfrac{1}{2}$-order 
    convergence rate in discretization schemes. However, in our models, the payoffs of 
    derivatives that mature before time $t$ appear in the BSDE drivers, causing 
    discontinuities at a finite number of points. A more realistic assumption, therefore, 
    is to require piecewise Hölder-$\tfrac{1}{2}$ continuity. By ensuring that the 
    discontinuity points coincide with the time grid $\theta$, the BSDE drivers will 
    satisfy the continuity requirement on each sub-interval $[t_n, t_{n+1})$. In practice, 
    this is typically sufficient to maintain the $\tfrac{1}{2}$-order convergence rate.

\end{itemize}
We note that the arguments presented above, which relax certain restrictive assumptions, require a rigorous justification. Nevertheless, the methodological ideas we employ align with techniques commonly used for related problems.

\section{Numerical experiments}\label{sec:numerics}
We consider Geometric Brownian motions both for the risky underlyings and for the bank and the counterparty asset processes, which for $t\in[0,T]$ follow 
\begin{equation}
\widehat{X}_t^i=1+\int_0^tr\widehat{X}_s^i\d s+\int_0^t\sigma^i\widehat{X}_s^i\d \widehat{W}_s^i,\ i\in\{1,\ldots,d\}\quad \widebar{X}_t^j = 1 + \int_0^tr\widebar{X}_s^j\d s +\int_0^t\widebar{\sigma}^j\widebar{X}_s^j\d \widebar{W}_s^j,\ j\in\{\B,\C\}.
\end{equation}
We set $T=1$, $r=0.05$,
\begin{equation*}
    \quad d = 5,\quad\widebar{\sigma}=\begin{pmatrix}
        0.2 \\
        0.3 
    \end{pmatrix}, \quad 
    \widehat{\sigma} = \begin{pmatrix}
        0.2 \\
        0.25 \\
        0.25 \\
        0.25 \\
        0.3
    \end{pmatrix}, \quad 
    \rho = \begin{pmatrix}
        1.0 & 0.9 & 0.2 & 0.5 & 0.1 & 0.1 & 0.2 \\
        0.9 & 1.0 & 0.4 & 0.3 & 0.2 & 0.3 & 0.2 \\
        0.2 & 0.4 & 1.0 & 0.2 & 0.75 & 0.15 & 0.25 \\
        0.5 & 0.3 & 0.2 & 1.0 & 0.35 & 0.05 & 0.15 \\
        0.1 & 0.2 & 0.75 & 0.35 & 1.0 & 0.15 & 0.05 \\
        0.1 & 0.3 & 0.15 & 0.05 & 0.15 & 1.0 & 0.25 \\
        0.2 & 0.2 & 0.25 & 0.15 & 0.05 & 0.25 & 1.0
    \end{pmatrix}.
\end{equation*}
Here, $\rho$ is the correlation matrix for $W=\text{Concat}(\widehat{W},\widebar{W})$. We consider a portfolio of $P=33$ European basket call options with respective maturities
$T^P=\{T_1,T_2,\ldots,T_P\}=\{1 , 1 , 1 , 0.8, 0.8, 0.6, 0.6, 0.4, 0.4,\\ 0.2, 0.2, 1 , 1 ,1 , 0.7, 0.7, 0.5, 0.5, 0.3, 0.3, 0.1, 0.1, 1 , 0.7, 0.7, 0.5,
       0.5, 0.3, 1 , 0.8, 0.8, 0.6, 0.6\}$. For each of the 33 derivatives the pay-off structure is given by $A_j=-\mathbb{I}_{\{t=T_j\}}\Big(\big(\prod_{i\in \I_j}\widehat{X}_{T_j}^i\big)^{\frac{1}{d_j}}-K_j\Big)^+$ with $K=\{1.05, 1.1 , 1.05, 1.05, 0.7 , 0.7 , 0.75, 0.75, 1,\\ 0.9 , 0.8 ,
       0.9, 1.1, 1.05, 0.85, 0.9 , 0.9 , 1.05, 1, 1  , 0.9 , 0.95,
       1.05, 0.7 , 0.7 , 0.75, 0.75, 1  , 0.9 , 0.8 , 0.9 , 1.1 , 1.05\}$.\\
Note that from the bank's perspective, this pay-off function is equivalent to buying a portfolio of call options from the counterparty. Each of the 33 derivatives has a pay-off depending on a subset of the $5$ risky underlying assets as follows 
\begin{align*}
    \mathcal{I} &= \{\mathcal{I}_1, \mathcal{I}_2, \ldots, \mathcal{I}_{33}\} 
      = \big\{ 
            [1,2,3,4,5],\ 
            [2,3,4,5],\ 
            [1,3,4,5],\ 
            [1,2,4,5],\ 
            [1,2,3,5], \\
    &\quad [1,2,3,4],\ 
            [1,2,3],\ 
            [1,2,4],\ 
            [1,2,5],\ 
            [1,3,4], 
     [1,3,5],\ 
            [2,3,4],\ 
            [2,3,5],\ 
            [3,4,5],\ 
            [1,2], \\
    &\quad [2,3],\ 
            [1,3],\ 
            [1],\ 
            [2],\ 
            [3],\ 
            [4],\ 
            [5], 
     [1,2],\ 
            [2,3],\ 
            [1,3],\ 
            [1],\ 
            [1,2,3,4,5],\ 
            [1,2,3,4,5], \\
    &\quad [1,2,3,4,5],\ 
            [2,3,4,5],\ 
            [1,2,3],\ 
            [1,2],\ 
            [2,3]
        \big\}.
\end{align*}
Finally, we set $N=200$, where $N+1$ is the number of discrete time points. 
\vspace{0.2cm}\newline
\noindent\textbf{Rationale behind our presentation choices}\newline
For each example across the four layers, we present empirical evidence of accuracy for representative paths and the solution distribution. In Layer one, where a closed-form analytic solution is available, we compare the empirical mean, 99th, and 1st percentiles of our approximate solutions to their analytical counterparts. In particular, the percentiles provide insight into the quality of our approximations in the tails.

In subsequent layers, it is too computationally expensive to generate many reference solutions, so we compare a few representative paths of our approximations to those from a costly nested Monte--Carlo method. To assess the entire distribution, we focus on the terminal and initial conditions of the BSDEs. The terminal condition is relatively inexpensive to sample since it is determined by a function evaluated at the forward SDE solution. We then compare the empirical distribution of the differences between our approximate and reference terminal conditions, which indicates how well our terminal conditions are satisfied. However, as discussed in \cite{andersson2023convergence}, accuracy in the terminal condition alone is insufficient because the approximation before the terminal time may still be poor unless strict conditions are met, see \cite{han2020convergence}.

The initial condition is deterministic and requires no nested Monte--Carlo approximations, allowing us to generate a reference solution relatively cheaply. However, as noted in several papers (see e.g., \cite{wang2023solving}), convergence in the initial condition does not guarantee convergence for $t>0$ when using the deep BSDE solver. This can easily be confirmed by terminating the training procedure early and noting that initial condition is accurately approximated while the full solution path is not.  

In summary, while accuracy in either the initial or terminal condition alone is insufficient, achieving both gives us greater confidence in our approximations.

In Appendix~\ref{sec:reference_solutions}, we present the methodology used to compute the reference solutions in each layer.

\subsection{Layer 1 - Clean values}
For the clean values, the objective is to simultaneously approximate 33 BSDEs, where the dimension of each BSDE is given by $d_j=|\mathcal{I}_j|$. In total, this yields a 93-dimensional, highly non-symmetric problem with five diffusive factors.

Since we consider a Geometric Brownian motion for the asset process $\widehat{X}$, we have access to a closed form solution. Moreover, the Geometric average of multiple Geometric Brownian motions, is itself a Geometric Brownian motion. This can be leveraged in order to rewrite each basket option with a Geometric mean pay-off into a standard European call option. Hence, we have access to closed form solutions for each derivative and in turn the clean portfolio value.

The left panel in Figure~\ref{fig:clean_value} displays the clean portfolio value, for three representatives realizations. The right panel displays the mean, 99th and 1st percentiles, respectively for an empirical distribution with $2^{16}$ samples. 
\begin{figure}[htp]
\centering
\begin{tabular}{cc}
          \includegraphics[width=80mm]{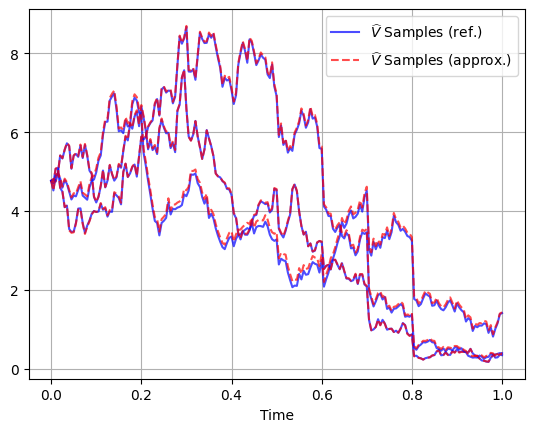}&   
          \includegraphics[width=80mm]{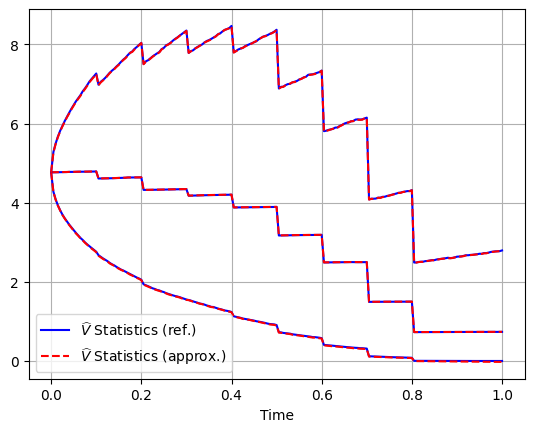}
\end{tabular}
\caption{Approximate clean portfolio values compared with their analytical counterparts. \textbf{Panel 1:} Three representative samples. \textbf{Panel 2:} Empirical mean, 99th and 1st percentiles.}\label{fig:clean_value}
\end{figure}

\subsection{Layer 2 - Initial Margin}
For initial margin (IM) computations, we choose VaR at the level $\alpha = 0.99$. As described in Section~\ref{sec:semi_discrete_opt}, the IM is given by the $\alpha$-level (or $(1-\alpha)$-level) VaR of positive (or negative) changes in the portfolio value over the MPR. We set $\mathrm{MPR}_n^\pi=\min\{\mathrm{MPR},N-n\}$ with $\mathrm{MPR}=8$. With this, we implicitly assume that the maximum MPR is 14.6 days (if time is measured in years).     

Unfortunately, we do not have an analytical reference solution for direct comparison. However, we can simulate $M = 2^{16}$ samples of the portfolio values (generated in layer~1 of our method) and then compute the empirical VaR measures. Because this requires nested Monte Carlo sampling, we are only able to generate reference solutions for a limited number of paths, which prevents us from comparing averages or percentiles of the IM as in Figure~\ref{fig:clean_value}.

In Figure~\ref{fig:IM}, we compare three representative IM paths obtained by our procedure with corresponding reference solutions generated by nested Monte Carlo. These examples illustrate how our approach compares to the nested Monte Carlo benchmark on a path-by-path basis.

For this experiment, we use three hidden layers with 16 nodes in each, $2^{16}$ training samples where 80\% is for training and 20\% for validation, 5000 epochs with 
patience of 100 epochs (terminate training if inference results on the validation set does not improve for 100 epochs), and a batch size of $2^{12}$.
\begin{figure}[htp]
\centering
\begin{tabular}{c}
          \includegraphics[width=160mm]{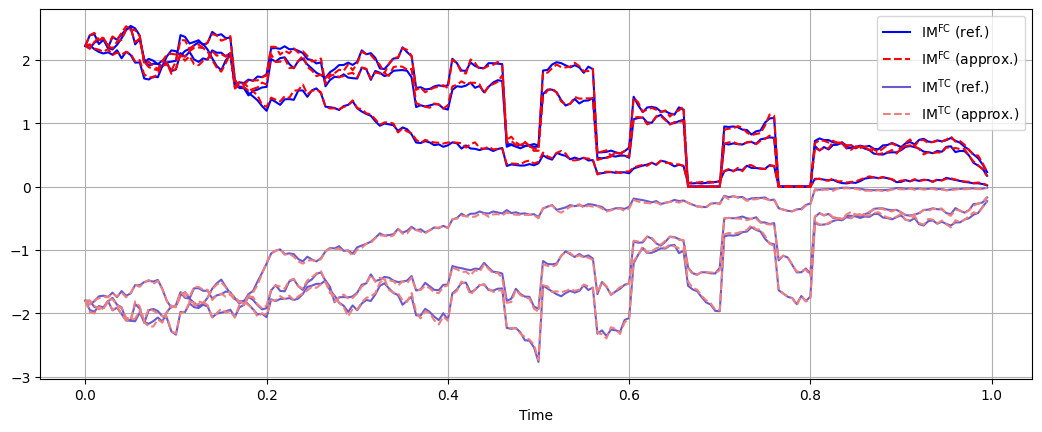}
\end{tabular}
\caption{Three representative samples of approximate IM compared with reference solutions obtained by nested Monte--Carlo sampling.}\label{fig:IM}
\end{figure}

\subsection{Layer 3 - ColVA, MVA, DVA and CVA}
This is the first layer that considers defaultable parties. We set $r_t^{c,l}=0.085$, $r_t^{c,b}=0.075$, $r_t^{\IM,l}=0.065$ and $r_t^{\IM,b}=0.05$, $\mathrm{LGD}^\C=\mathrm{LGD}^\B=0.3$ and assume that $50\%$ of the clean portfolio value is collateralized, \textit{i.e.,}, for $t\in[0,\tau\wedge T]$, $C_t=0.5\widehat{V}_t$, and we let the default barriers for the bank and the counterparty be defined by $\xi_t^\B\defeq0.575$, $\xi_t^\C\defeq 0.675$. Moreover, for $\mathrm{x}\in\{\mathrm{Col},\mathrm{M},\mathrm{C},\mathrm{D}\}$, we set $q_1^\mathrm{xVA}= \mathbf{0}$ and \begin{equation*}
    q_2^\ColVA=\begin{pmatrix}
        [\mathbf{0}^4]\\0.2\\0.35
    \end{pmatrix},\quad
        q_2^\MVA=\begin{pmatrix}
        [\mathbf{0}^4]\\0.2\\0.35
    \end{pmatrix},\quad q_2^\CVA=\begin{pmatrix}
        [\mathbf{0}^4]\\0\\0.35
    \end{pmatrix},\quad    q_2^\DVA=\begin{pmatrix}
        [\mathbf{0}^4]\\0.2\\0
    \end{pmatrix}.
\end{equation*} Here $[\mathbf{0}^4]$ is a column vector of four zeros. The specific parameter choices above are arbitrary, but they reflect the following motivations.
For ColVA and MVA, it is important to increase the default probability of both the bank and the counterparty, 
since the possibility of either party defaulting is relevant in these adjustments. 
For CVA, emphasis falls on the counterparty's default probability in order to capture the effect accurately, 
whereas for DVA, the emphasis shifts to the bank's own default probability. 
These modifications ensure that each type of valuation adjustment is properly accounted for 
in scenarios where the relevant party is more likely to default.

\begin{figure}[htp]
\centering
\begin{tabular}{cc}
          \includegraphics[width=80mm]{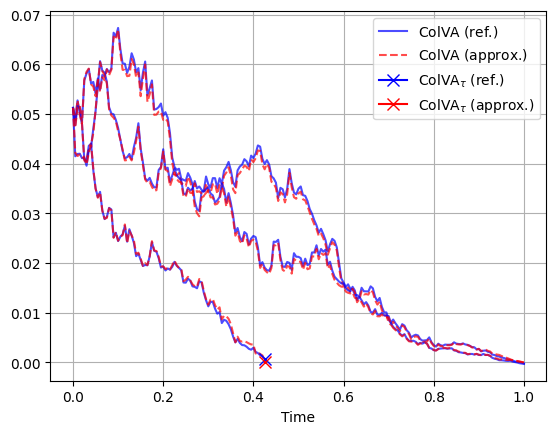}& 
          \includegraphics[width=80mm]{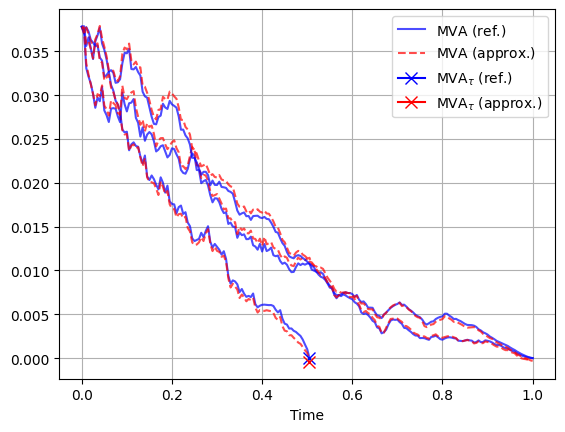}\\
        \includegraphics[width=80mm]{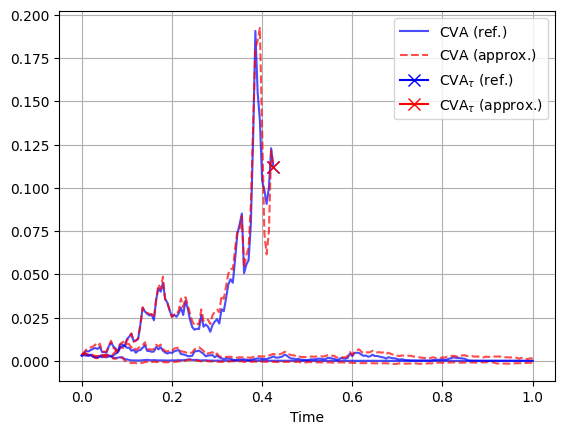}& 
          \includegraphics[width=80mm]{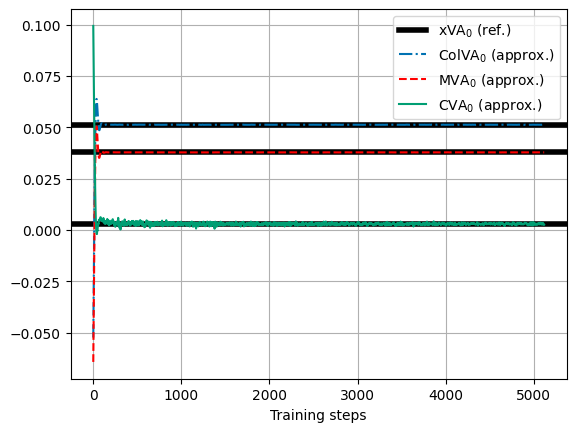}
\end{tabular}
\caption{\textbf{Panels 1-3}: The approximate ColVA, MVA and CVA compared with their respective reference solutions. \textbf{Panel 4}: The approximate initial conditions $\ColVA_0,\ \MVA_0$ and $\CVA_0$ plotted against training batches.} \label{fig:layer_3}
\end{figure}
In Figure~\ref{fig:layer_3}, panels 1-3, we compare three representative approximate solutions with their references for ColVA, MVA and CVA. For ColVA and MVA, we have one trajectory where the bank defaults and two where no default occur and for CVA, we have one trajectory where the counterparty defaults and two where no default occur. Panel 4 displays the convergence of the initial conditions during training. Note that for this specific problem DVA, is identically zero for all $t\in[0,T]$, which is why we do not include any results on DVA.
\begin{figure}[htp]
\centering
\begin{tabular}{cc}
          \includegraphics[width=80mm]{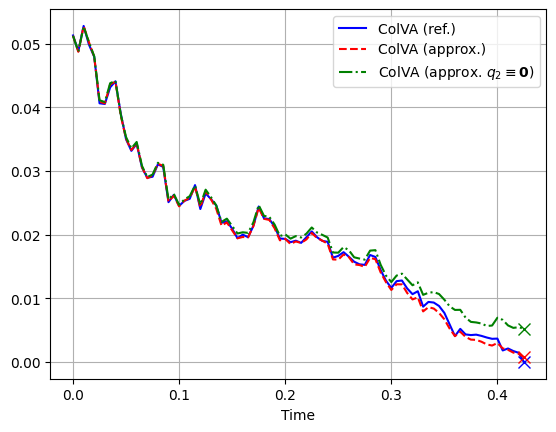}&   
          \includegraphics[width=80mm]{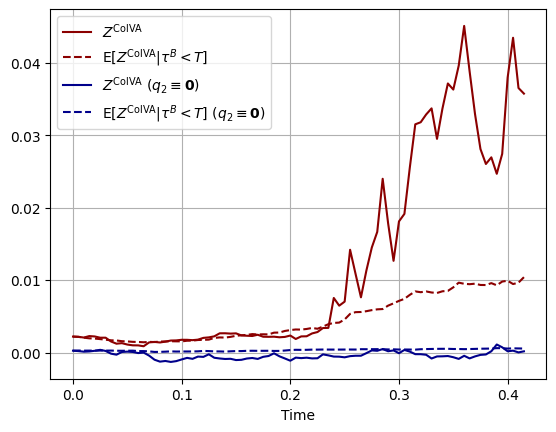}
\end{tabular}
\caption{A trajectory where default of the bank occurs at $t=0.42$. \textbf{Panel 1:} Comparison of the reference solution and our approximations with and without the measure change technique. \textbf{Panel 2:} Comparison of the control process of the BSDE in the 6th component which is associated with default of the bank.}\label{fig:ColVA_compare_2}
\end{figure}

In Figure~\ref{fig:ColVA_compare_2}, we compare our ColVA approximation with and without the measure-change technique for a trajectory where the bank defaults at $t = 0.42$. We observe that, when the measure change is applied, the terminal condition is accurately captured. By contrast, without the measure change, ColVA near the default event is poorly approximated.

The right panel of Figure~\ref{fig:ColVA_compare_2} illustrates this by plotting the sixth component of $Z^{\mathrm{ColVA}}$ with and without the measure-change technique for a representative sample path that defaults (solid lines), and the average over all samples that default prior to $T$ (dashed lines). We see that, if no measure change is used, $Z^{\mathrm{ColVA}}$ remains near zero at all times. This happens because the neural network rarely ``sees'' default scenarios under the original measure, and thus it cannot learn the correct control behavior close to the bank's default. Building on the insight from the right panel, we see that the green ColVA curve aligns with the scenario in which the bank's default probability is negligible (i.e.\ effectively zero).

As pointed out above, it is too expensive to sample many samples from the reference solutions of ColVA, MVA and CVA and hence, we only compare with the reference solutions for representative trajectories. However, we can efficiently sample from the reference terminal conditions and compare with our approximations to gain some insight into accuracy of the full distribution of our approximations. In Figure~\ref{fig:terminal_errors}, we display empirical distributions of terminal condition errors for ColVA, MVA and CVA, i.e., the difference between our approximate value at either default or the terminal time and the reference values.  
\begin{figure}[htp]
\centering
\begin{tabular}{c}
          \includegraphics[width=165mm]{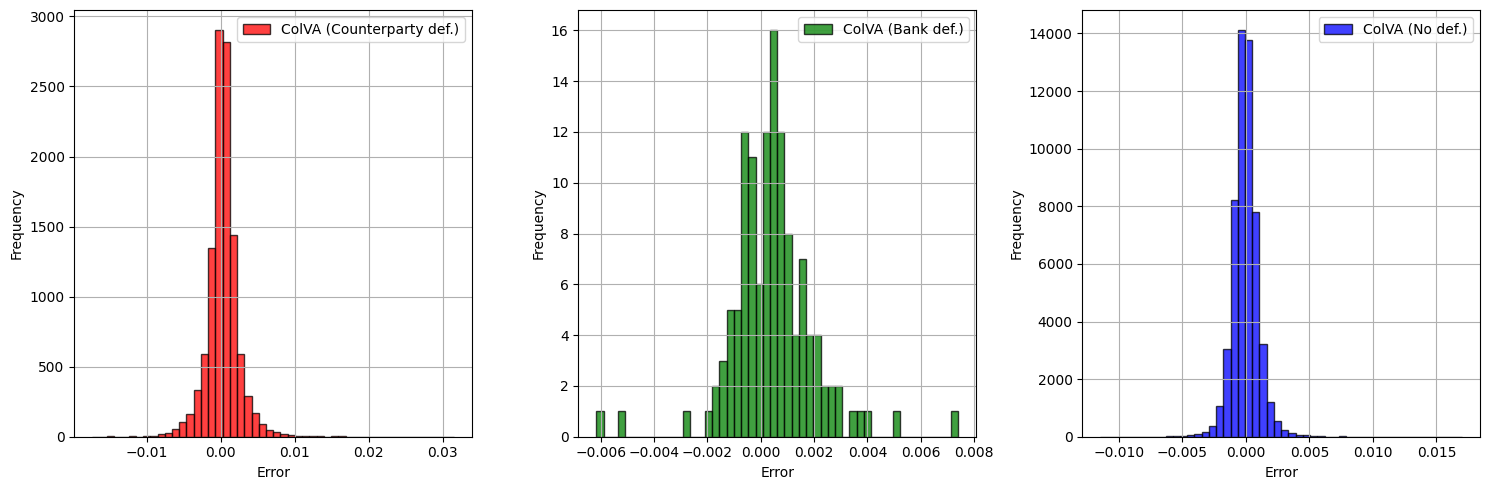}\\
          \includegraphics[width=165mm]{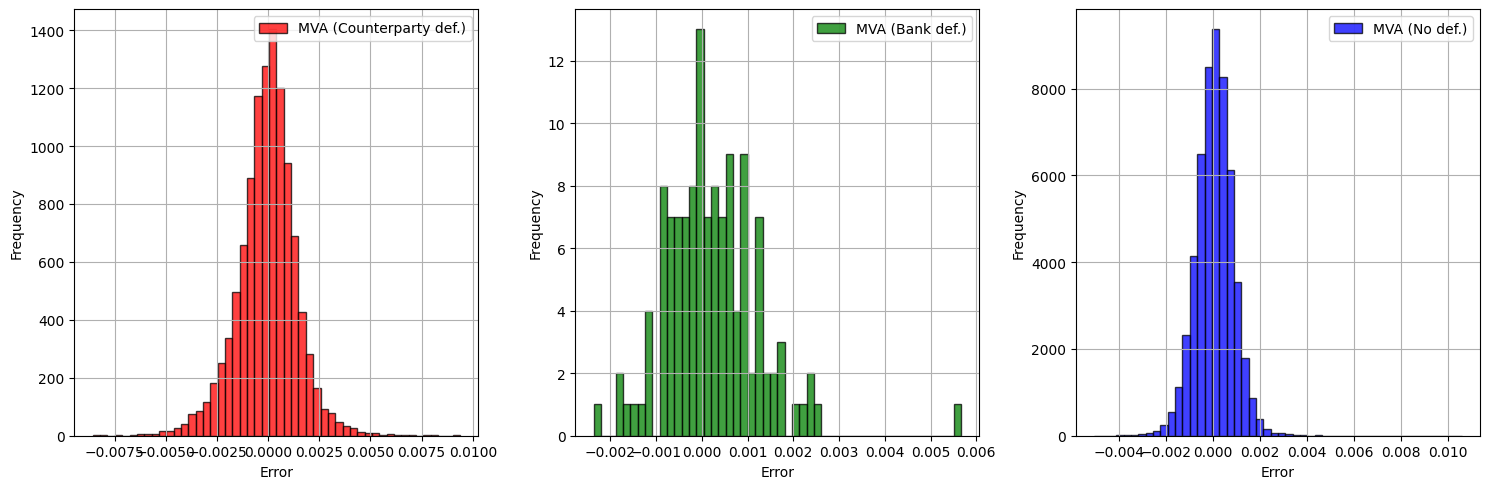}\\
          \includegraphics[width=165mm]{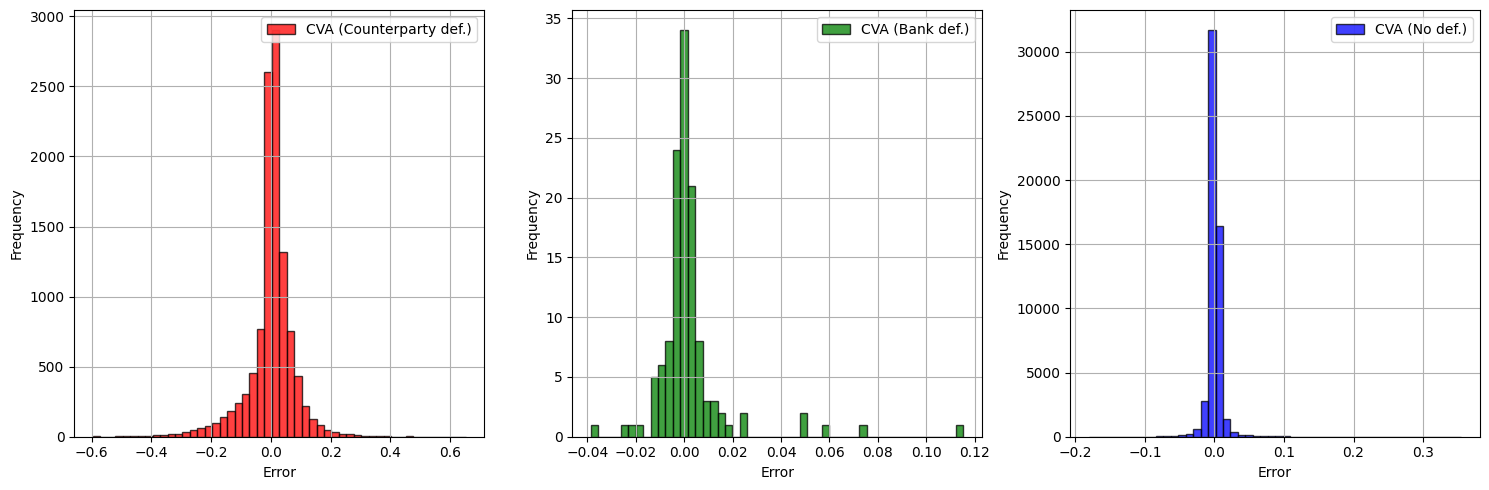}
\end{tabular}
\caption{Empirical distributions of terminal condition errors for ColVA, MVA and CVA.}\label{fig:terminal_errors}
\end{figure}
\vspace{0.25cm}
\newline
\noindent\textbf{Stress test of our algorithm}:\newline
To push our algorithm to its limits, we perform an experiment in which all parameters are held constant except for $\xi_t^\C$, which is decreased from 0.675 to 0.45. As a result, the probability that the counterparty defaults before the bank—and that this default occurs before time $T$—drops to approximately 0.6\%. Under these conditions, only a few trajectories register a default. Moreover, when we also require that the close-out amount be positive (since collateral and IM can sometimes lead to negative values), this probability further declines to about 0.2\%. In Figure~\ref{fig:CVA_compare} (panel 1), we observe that the empirical distribution of terminal errors obtained with the measure change technique has a higher concentration of probability mass around zero and is more symmetric about zero compared to the distribution without the technique. Moreover, in panels 2-3, we observe that for two representative trajectories, even though the terminal error is comparable with and without the measure change technique, the approximation performed without the measure change technique fails to capture the fluctuations caused by the bank's asset process, leading to poor path-wise accuracy.

\begin{figure}[htp]
\centering
\resizebox{\textwidth}{!}{
  \begin{tabular}{ccc}
    \includegraphics{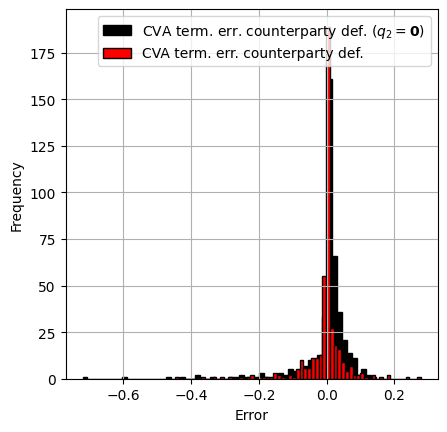} &   
    \includegraphics{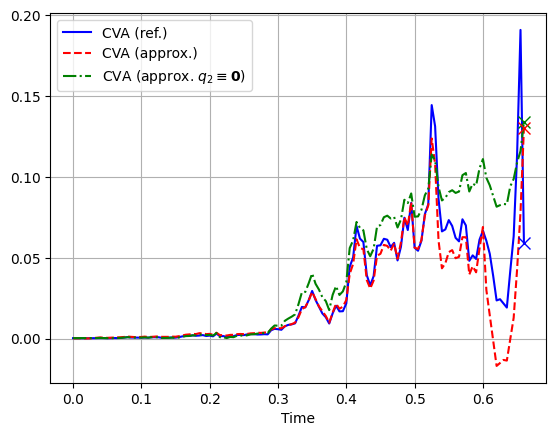} &
    \includegraphics{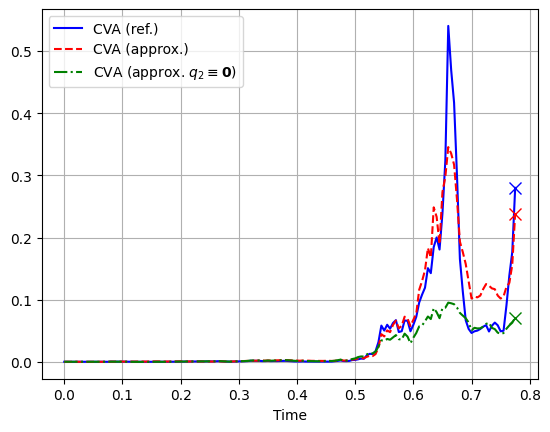}
  \end{tabular}
}
\caption{CVA in scenarios where the counterparty defaults before the terminal time and prior to the bank's default. 
\textbf{Panel 1:} Empirical distributions of terminal errors for CVA approximations with and without the measure change technique.  
\textbf{Panels 2-3:} Comparison of the reference solution with our approximations (with and without the measure change technique) for representative scenarios where the counterparty defaults at $t=0.66$ (Panel 2) and $t=0.775$ (Panel 3).
}\label{fig:CVA_compare}
\end{figure}

\subsection{Layer 4 - FVA}
We set $r^{f,b}=0.75$, $r^{f,l}=0.65$, $q_1^\FVA= \mathbf{0}$ and $q_1^\FVA=(0,0,0,0,0,0.35,0.2)^\top$.

In Figure~\ref{fig:layer_4}, panel 1, we compare three representative approximate solutions with their references for FVA. In one of the trajectories the counterparty defaults and in two no default occur. Panel 2 displays the convergence of the initial conditions during training. 
\begin{figure}[htp]
\centering
\begin{tabular}{cc}
        \includegraphics[width=80mm]{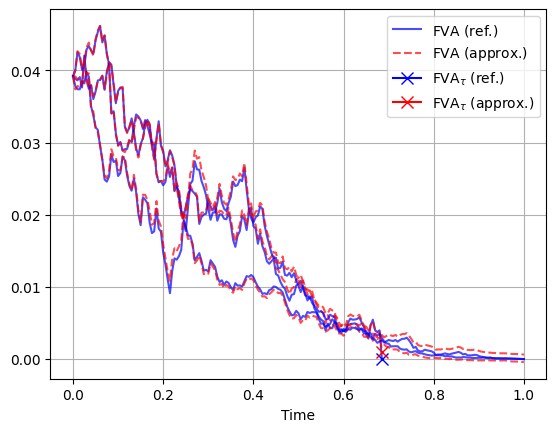}& 
          \includegraphics[width=80mm]{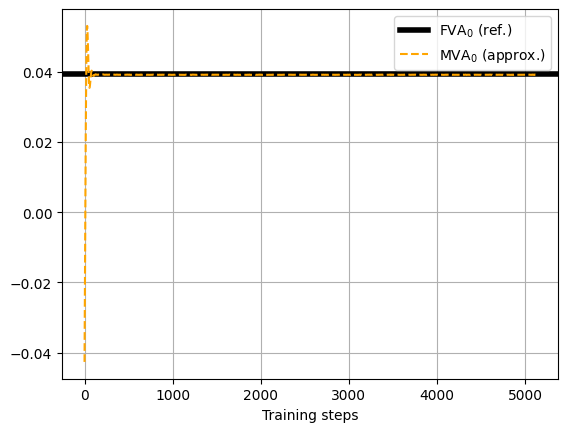}
\end{tabular}
\caption{\textbf{Panel 1}: The approximate FVA compared with reference solutions. \textbf{Panel 2}: The approximate initial condition, $\FVA_0$, plotted against training batches.} \label{fig:layer_4}
\end{figure}
Figure~\ref{fig:terminal_error_FVA} displays the empirical distributions of terminal condition errors for FVA. 
\begin{figure}[htp]
\centering
\begin{tabular}{c}
\includegraphics[width=165mm]{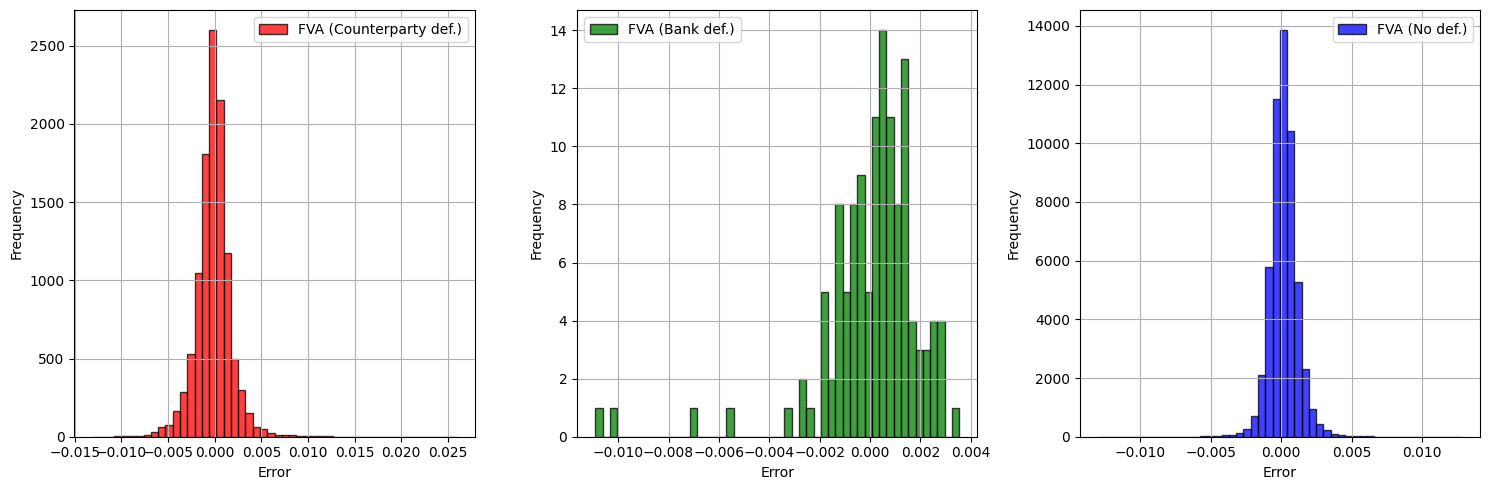}
\end{tabular}
\caption{Empirical distributions of terminal condition errors for FVA.}\label{fig:terminal_error_FVA}
\end{figure}

\section{Conclusions and potential future research directions}\label{sec:conclusion}
This paper presents a structural model for computing multi-layer valuation adjustments (xVAs) by formulating them as a high-dimensional system of backward stochastic differential equations (BSDEs). We employ iterative deep BSDE techniques, including strategic measure changes, to effectively capture rare default events and mitigate the need for costly nested Monte Carlo schemes. Our numerical results demonstrate both the flexibility and high accuracy of this approach across four layers even when confronted with substantial dimensionality, asymmetric portfolio structures, numerous risk factors, and heavily interdependent BSDEs.
This is evidenced by comparisons to closed-form references (where available) and carefully constructed nested Monte Carlo benchmarks.  In addition, we establish the first error analysis for deep BSDE methods with stopping times, showing that domain-restricted BSDEs retain a reduced $\mathcal{O}(h^1/4-\varepsilon)$ convergence rate under suitable assumptions.

A natural extensions of the present paper would be to allow for multiple counterparties, but this is clearly feasible by means of a divide and conquer approach that leverages our results. Further extensions could involve derivatives with early exercise features and/or jumps in the asset dynamics.  We remark that the inclusion of the capital value adjustment (KVA) would also not represent a significant challenge to our approach. We stress again that our algorithm, under the assumption of a structural credit model, could be applied to alternative specifications of the xVA BSDEs where for example FVA and KVA are computed at the level of the whole bank's portfolio and not at the single netting set as we do.

As for the second extension, one could either formulate reflected BSDEs for the clean values or adopt the deep optimal stopping method proposed in \cite{becker2019deep} to compute optimal stopping decisions. In the latter approach, the optimal stopping algorithm is applied in an initial phase and in a second phase, the clean values are treated as barrier options, which can be solved with our algorithm for BSDEs with random stopping times (without optionality which is handled in the initial phase). A similar strategy, applying one algorithm for execution decisions and another for pathwise valuation using neural network-based regression, was proposed in \cite{andersson2021deep} and later extended to entire netting sets in \cite{andersson2021deep_2}.

Allowing for Lévy jumps in the asset dynamics would remove the predictability of default events and address the common criticism that structural models underestimating short-term default risk, see \textit{e.g.,} \cite{BrigoVrins2018,ballotta2015counterparty}. In this setting, one could extend our method by incorporating a deep approach for jump BSDEs; see, for example, \cite{alasseur2024deep,andersson2022deep,lu2024temporal}.

\section*{Acknowledgments}
We thank Athena Picarelli and Adam Andersson for their valuable insights and fruitful discussions.

.


\appendix
\section{Reference solutions}
\label{sec:reference_solutions}

In this section, we describe how we generate reference solutions for the system 
\eqref{eq:FBSDE_repeat}, which in its conditional-expectation form can be written as 
\begin{align}
\label{eq:FBSDE_cond_exp}
\begin{split}
    \begin{dcases}
  X_t \;=\; x_0 
  \;+\; \displaystyle\int_0^t b\bigl(s,X_s\bigr)\,\mathrm{d}s 
  \;+\; \int_0^t \sigma\bigl(s,X_s\bigr)\,\odot\,\mathrm{d}W_s, 
  \\[6pt]
  Y_t
  \;=\; \E\biggl[\,
    \mathbbm{1}_{\{\tau \le T\}}\chi_\tau
    \;-\; \displaystyle\int_{(t,\;\tau\wedge T]} \mathrm{d}\Lambda_s
    \;+\; \int_t^{\tau\wedge T}
      \bigl(f_s^Y \;-\; r_s\,Y_s\bigr)\,\mathrm{d}s
  \;\Big|\; X_t \;=\; x
  \biggr].
\end{dcases}
\end{split}
\end{align}
The aim is to obtain a sufficiently accurate approximation of $(X_t, Y_t)$ 
for each relevant $(t,x)$, which we refer to
as the \emph{reference solution} 
in subsequent numerical comparisons.

\begin{enumerate}[\itshape i)]
\item \textbf{Euler--Maruyama for $X$ and $\tau$.}  
We approximate the process $X$ by applying the Euler--Maruyama scheme over a refined time grid of size $N_{\mathrm{ref}} \in \mathbb{N}^{+}$ (potentially large). Denote by $h_{\mathrm{ref}} = \tfrac{T}{N_{\mathrm{ref}}}$ the corresponding step size. We then define the discrete stopping time $n_{\tau}$ as the smallest grid index $n$ such that $X_{n}^\pi$ crosses the given default boundary.

Under standard assumptions on the SDE coefficients, Euler--Maruyama yields a strong approximation of order $\mathcal{O}\bigl(\sqrt{h_{\mathrm{ref}}}\bigr)$. However, because $\tau$ is defined as a hitting time of a boundary, the error also depends on the geometry of that boundary; see, e.g., \cite{kloeden1992stochastic,mikulevicius1991rate,marion2002convergence}.
\item \textbf{Quadrature for the integrals.}
For each simulated path, we approximate the integrals
\begin{equation*}
  \int_{(t,\;\tau\wedge T]} \mathrm{d}\Lambda_{s}
  \quad\text{and}\quad
  \int_t^{\tau\wedge T}
    \bigl(f_{s}^{Y} - r_{s}\,Y_{s}\bigr)\,\d s
\end{equation*}
by a left Riemann sum on the same time grid, truncating at the stopping index. Under standard assumptions on the integrands, these approximations incur an error of order~$\mathcal{O}\bigl(h_{\mathrm{ref}}\bigr)$.
\item \textbf{Monte Carlo approximation.} 
Finally, the conditional expectation $\E[\ \cdot \mid X_t=x]$
is replaced by the corresponding sample average over 
the $M^\mathrm{ref}\in\N$ simulated paths that started at $(t,x)$. In particular, this implies that \textit{i} and \textit{ii} are repeated $M_\text{ref}$ times and the approximation error is of order $\mathcal{O}(M_\text{ref}^{-1/2})$ in a mean-squared sense (or in probability).
\end{enumerate}
By taking both $N^\mathrm{ref}$ (the number of time steps) and 
$M^\mathrm{ref}$ (the number of sample paths) sufficiently large, 
we obtain a solution $Y^{\pi,\mathrm{ref}}_{t_n}(x)$ with an error 
that is typically of order 
\begin{equation*}
    O\bigl(h_\mathrm{ref}^{1/2}\bigr) + O\bigl(N_{\mathrm{MC}}^{-1/2}\bigr),
\end{equation*}
reflecting both the time-discretization (including default-time approximation) 
and the Monte Carlo sampling. In practice, one chooses $N^\mathrm{ref}$ and 
$M^\mathrm{ref}$ large enough that this reference solution is stable and 
accurate to the desired precision.

The initial margins in this paper are given by the $\alpha$- and $(1-\alpha)$-quantiles of price movements over the margin period of risk. At each discrete time step, we approximate these upper and lower quantiles by their empirical counterparts from an $M_q$-sample, $M_q \in \mathbb{N}$. In analogy to a Monte Carlo approximation of an expectation, the empirical quantiles converge at rate $\mathcal{O}(M_q^{-1/2})$ in a mean-squared sense. However, the constant in this convergence rate depends on both $\alpha$ and the underlying distribution (specifically on the density near the true $\alpha$-quantile).

\begin{remark}
Although we can control the approximation errors of our reference solutions,
it is important to note that in a nested approach each higher-level approximation 
relies on approximations computed at lower levels. As a result, any error 
introduced at the lower levels carries over to, and can be compounded in, subsequent 
levels. Moreover, applying this nested approach more than once quickly becomes 
prohibitively expensive in terms of computational resources. Therefore, in this 
work we treat the neural network approximations at lower levels as a given ``ground truth'' and 
use these directly as inputs. In this way, the reference-solution scheme outlined above is applied only at the final level. For instance, approximations at level 4 would use our neural network approximation at levels 1-3 and apply the nested approach at level 4. This can be improved if we have closed form solutions at the first layer.   
\end{remark}

\end{document}